\documentclass[a4paper,11pt,reqno]{amsart}
\usepackage[margin=2.6cm]{geometry}

\makeatletter

\usepackage{graphicx} 
\usepackage{braket}
\usepackage{amsmath,amssymb,amstext}
\usepackage{amsthm}
\usepackage{dsfont}

\newtheorem{lem}{Lemma}[section]
\newtheorem{rem}[lem]{Remark}
\newtheorem{thm}[lem]{Theorem}
\newtheorem{cor}[lem]{Corollary}

\numberwithin{equation}{section}

\title[Proof of the Landau-Pekar Formula for the effective Mass of the Polaron]{Proof of the Landau-Pekar Formula for the effective Mass of the Polaron at strong coupling}
\author{Morris Brooks}

\begin{document}

\maketitle
\begin{abstract} 
We study the Fröhlich polaron in the regime of strong coupling and prove the asymptotically sharp lower bound on the effective mass $m_{\mathrm{eff}}(\alpha)\geq \alpha^4 m_{\mathrm{LP}}-C\alpha^{4-\epsilon}$, where $m_{\mathrm{LP}}$ is an explicit constant. Together with the corresponding upper bound, which has been verified recently in \cite{BS2}, we confirm the validity of the celebrated Landau-Pekar formula \cite{Lpekar} from 1948 for the effective mass $\underset{\alpha\rightarrow \infty}{\lim}\alpha^{-4}m_{\mathrm{eff}}(\alpha)=m_{\mathrm{LP}}$ as conjectured by Spohn \cite{Sp} in 1987. 
\end{abstract}


\section{Introduction and Main Result}
\label{Sec:Introduction_and_Main_Result}
In this manuscript we study the Fröhlich polaron, which is a model describing the interactions of a charged particle, e.g. an electron, with a polarizable medium \cite{F37}. From a mathematical point of view, the Fröhlich polaron is a popular toy model of a quantum field theory, as it is simple enough to allow for rigorous mathematical proofs while still giving rise to complex and non-trivial phenomena, such as the effective increase of the electrons mass due to its interactions with the quantized excitations of the medium. It will be the objective of this article to study the dependence of the effective mass $m_\mathrm{eff}(\alpha)$ on the coupling strength $\alpha$ as $\alpha\rightarrow \infty$ goes to infinity. Due to a conjecture by Spohn \cite{Sp}, it is expected that $m_\mathrm{eff}(\alpha)$ grows with a quartic power in $\alpha$ according to the Landau-Pekar formula \cite{Lpekar}, i.e. there exists a (rather explicit) constant $m_{\mathrm{LP}}\in (0,\infty)$ such that
\begin{align}
\label{Eq:Landau_Pekar_Conjecture}
    \underset{\alpha\rightarrow \infty}{\lim}\alpha^{-4}m_{\mathrm{eff}}(\alpha)=m_{\mathrm{LP}}.
\end{align}
Making use of the creation and annihilation operators $a^*$ and $a$ acting on the Fock space $\mathcal{F}$ over $L^2\! \left(\mathbb R^3\right)$, which satisfy for $f,g\in L^2\! \left(\mathbb R^3\right)$ the re-scaled canonical commutation relations
\begin{align*}
    [a(f),a^*(g)]=\alpha^{-2}\langle f,g\rangle,
\end{align*}
we introduce the Fröhlich Hamiltonian $\mathbb H$ as the self-adjoint operator
\begin{align}
\label{Eq:First_Definition_Hamiltonian}
    \mathbb H:=-\Delta_x+\mathcal{N}-a^*(v_x)-a(v_x)
\end{align}
acting on the Hilbert space $L^2\! \left(\mathbb R^3\right)\otimes \mathcal{F}$, where $v_x(y):=\pi^{-\frac{3}{2}}|y-x|^{-2}$ and $\mathcal{N}$ is the (re-scaled) particle number operator defined in terms of an orthonormal basis $\{u_j:j\in \mathbb N\}$ of $L^2\! \left(\mathbb R^3\right)$ as
\begin{align}
\label{Eq:Definition_Particle_Number_II}
    \mathcal{N}:=\sum_{j=0}^\infty a^*(u_j)a(u_j).
\end{align}
Note that $\Delta_x$ refers to the Laplace operator acting only on the $L^2\! \left(\mathbb R^3\right)$ factor in $L^2\! \left(\mathbb R^3\right)\otimes \mathcal{F}$ and $x$ refers to the position of the electron, i.e. the position operator in $L^2\! \left(\mathbb R^3\right)$. A detailed introduction to Fock space formalism and the objects appearing in the definition of $\mathbb H$ in Eq.~(\ref{Eq:First_Definition_Hamiltonian}) is given in Section \ref{Sec:Fock_Space_Formalism}. It is a central observation that $\mathbb H$ is invariant under a joint translation of the electron $x\mapsto x-z$ and the polarization field $a(f)\mapsto a(f_z)$, where we define the shifted function $f_z(y):=f(y-z)$, which is generated by a family of self-adjoint operators
\begin{align*}
    \mathbb P=(\mathbb P_1,\mathbb P_2,\mathbb P_3).
\end{align*}
Consequently, the joint spectrum $\sigma(\mathbb P,\mathbb H)$ is well-defined and we can introduce the ground state energy as a function of the total momentum $p\in \mathbb R^3$ as
\begin{align*}
    E_\alpha(p):=\inf\{E\in \mathbb R: (p,E)\in \sigma(\mathbb P,\mathbb H)\}.
\end{align*}
It is known that $E_\alpha(p)$ obtains its global minimum at $p=0$, see \cite{DS}, and therefore the ground state energy $E_\alpha:=\inf \sigma(\mathbb H)$ is given by $E_\alpha(0)$. In the past decades numerous results have been obtained on the asymptotic behaviour of $E_\alpha$ and $E_\alpha(p)$ in the regime of strong coupling, such as the seminal work \cite{DV}, which established the convergence of the ground state energy 
\begin{align}
\label{Eq:Introduction_Convergence_Semi_Classical_Energy}
    \underset{\alpha\rightarrow \infty}{\lim}E_\alpha=e^\mathrm{Pek}
\end{align}
to the minimal value $e^\mathrm{Pek}$ of the corresponding semi-classical functional  
\begin{align}
\label{Eq:Pekar_Functional}
    \mathcal{F}^\mathrm{Pek}(\varphi):=\|\varphi\|^2+\inf \sigma \Big(\! -\Delta-2 v*\mathfrak{Re}[\varphi]\Big),
\end{align}
where $v(y):=\pi^{-\frac{3}{2}}|y|^{-2}$ and $\varphi\in L^2\! \left(\mathbb R^3\right)$. Notably the proof of Eq.~(\ref{Eq:Introduction_Convergence_Semi_Classical_Energy}) in \cite{DV} is based on a functional integral representation of $E_\alpha$ using a Feynman-Kac formula for the semi-group $e^{-T\mathbb H}$, see \cite{F}. Later, an elementary proof of Eq.~(\ref{Eq:Introduction_Convergence_Semi_Classical_Energy}) has been established in \cite{LT} using a functional analytic approach. Regarding the asymptotic behaviour of the energy-momentum relation $E_\alpha(p)$ and the effective mass $m_\mathrm{eff}(\alpha)$, where
\begin{align}
\label{Eq:Effective_Mass_Def}
    m_\mathrm{eff}(\alpha):=\underset{P\rightarrow 0}{\lim}\frac{|p|^2}{2(E_\alpha(p)-E_\alpha)},
\end{align}
there has been significant progress in recent years. It has been verified in \cite{LS} that 
\begin{align}
    \label{Eq:Divergence_of_Effective_Mass}
    m_\mathrm{eff}(\alpha)\underset{\alpha \rightarrow \infty}{\longrightarrow }\infty,
\end{align}
i.e. the effective mass diverges in the limit of large $\alpha$. Using a functional integral representation of $m_\mathrm{eff}(\alpha)$, see \cite{Sp,DS} and \cite{MV,BP}, the rate of divergence in Eq.~(\ref{Eq:Divergence_of_Effective_Mass}) has been quantified in \cite{BP} by a lower bound of the form $m_\mathrm{eff}(\alpha)\geq c\alpha^{\frac{2}{5}}$, which has been improved in \cite{S} by the almost quartic lower bound $m_\mathrm{eff}(\alpha)\geq c(\log \alpha)^{-6}\alpha^{4}$ and in \cite{BMSV} by a lower bound of the form
\begin{align}
\label{Eq:Quartic_Lower_Bound}
 m_\mathrm{eff}(\alpha)\geq c\alpha^{4} 
\end{align}
for a suitable constant $c\in (0,\infty)$. It is worth pointing out, that Eq.~(\ref{Eq:Quartic_Lower_Bound}) already captures the correct quartic divergence $\alpha^4$ of the effective mass as $\alpha \rightarrow \infty$. In our main Theorem \ref{Th:Main} we are going to improve Eq.~(\ref{Eq:Quartic_Lower_Bound}) by establishing the lower bound
\begin{align}
\label{Eq:In_Text_Main_Result}
     m_\mathrm{eff}(\alpha)\geq \alpha^{4}m_{\mathrm{LP}}-C\alpha^{4-\epsilon}
\end{align}
for a suitable $\epsilon>0$, where $m_{\mathrm{LP}}:=\frac{2}{3}\left\|\nabla \varphi^\mathrm{Pek}\right\|^2$ and $\varphi^\mathrm{Pek}>0$ is the unique radial minimizer of the Pekar functional defined in Eq.~(\ref{Eq:Pekar_Functional}), see \cite{Li}. The corresponding upper bound 
\begin{align}
\label{Eq:Upper_Bound_Effective_Mass}
     m_\mathrm{eff}(\alpha)\leq \alpha^{4}m_{\mathrm{LP}}+C\alpha^{4-\epsilon}
\end{align}
has been proven recently in \cite{BS2}, using the results in \cite{Po}. Therefore we conclude that the constant $m_{\mathrm{LP}}$ in Eq.~(\ref{Eq:In_Text_Main_Result}) is the optimal one such that an inequality of the form Eq.~(\ref{Eq:In_Text_Main_Result}) holds. Combining Eq.~(\ref{Eq:In_Text_Main_Result}), see Theorem \ref{Th:Main}, and Eq.~(\ref{Eq:Upper_Bound_Effective_Mass}) confirms the Landau-Pekar-Spohn conjecture on the effective mass stated in Eq.~(\ref{Eq:Landau_Pekar_Conjecture}). 

Regarding the upper bound in Eq.~(\ref{Eq:Upper_Bound_Effective_Mass}), it has been shown in \cite{BS2}, using an upper bound on $E_\alpha$ from \cite{MMS}, that the energy-momentum relation $E_\alpha(p)$ satisfies a lower bound of the form
\begin{align}
\label{Eq:Non_Small_p_Lower_Bound_Energy}
    E_\alpha(p)-E_\alpha\geq \min\left\{\frac{|p|^2}{2\alpha^4 m_\mathrm{LP}},\alpha^{-2}\right\}-C\alpha^{-(2+\epsilon)}.
\end{align}
Since the error term $C\alpha^{-(2+\epsilon)}$ does not scale like $|p|^2$ in  $p$, this lower bound alone is clearly insufficient to obtain a non-trivial upper bound on $m_\mathrm{eff}(\alpha)$. However, together with the results in \cite{Po}, which especially imply that $E_\alpha(p)\leq E_\alpha+\frac{|p|^2}{2m_\mathrm{eff}(\alpha)}$, Eq.~(\ref{Eq:Upper_Bound_Effective_Mass}) follows from Eq.~(\ref{Eq:Non_Small_p_Lower_Bound_Energy}). A corresponding upper bound to Eq.~(\ref{Eq:Non_Small_p_Lower_Bound_Energy}) has been verified previously in \cite{MMS}
\begin{align}
\label{Eq:Non_Small_p_Upper_Bound_Energy}
    E_\alpha(p)-E_\alpha\leq \min\left\{\frac{|p|^2}{2\alpha^4 m_\mathrm{LP}},\alpha^{-2}\right\}+C\alpha^{-\left(\frac{5}{2}+\epsilon\right)},
\end{align}
using a lower bound on $E_\alpha$ from \cite{BS1}. Again, due to the lack of a $|p|^2$ scaling in the error term $C\alpha^{-\left(\frac{5}{2}+\epsilon\right)}$, Eq.~(\ref{Eq:Non_Small_p_Upper_Bound_Energy}) alone is insufficient to obtain a non-trivial lower bound on $m_\mathrm{eff}(\alpha)$.

\begin{thm}
\label{Th:Main}
    Let $m_{\mathrm{LP}}$ be the Landau-Pekar constant introduced below Eq.~(\ref{Eq:In_Text_Main_Result}) and $m_\mathrm{eff}(\alpha)$ the effective mass defined in Eq.~(\ref{Eq:Effective_Mass_Def}). Then there exist $C,\epsilon>0$ such that
    \begin{align*}
        m_\mathrm{eff}(\alpha)\geq \alpha^{4}m_{\mathrm{LP}}-C\alpha^{4-\epsilon}.\\
    \end{align*}
\end{thm}

\textbf{Proof strategy of Theorem \ref{Th:Main}.} By the definition of the effective mass $m_\mathrm{eff}(\alpha)$ in Eq.~(\ref{Eq:Effective_Mass_Def}), it is clearly enough to show the upper bound
\begin{align}
\label{Eq:Introduction_Upper_Bound_Energy-Momentum_Relation}
    E_\alpha(p)\leq E_\alpha+\frac{|p|^2}{2\alpha^4 m_{\mathrm{LP}}}+C\alpha^{-(4+\epsilon)}|p|^2.
\end{align}
In order to find a more convenient expression for $E_\alpha(p)$, let us switch to the Lee-Low-Pines picture, \cite{LLP}, where the frame for the polarization field moves along with the position of the electron $x$. After applying this unitary transformation, the operator $\mathbb H$ reads
\begin{align*}
    \mathbb H_{\frac{1}{i}\nabla_x}=\left( \frac{1}{i}\nabla_x- \mathcal P\right)^2+\mathcal{N}-a^*(v)-a(v),
\end{align*}
where $\mathcal P$ is the generator of the translations $a(f)\mapsto a(f_z)$ and $v(y):=\pi^{-\frac{3}{2}}|y|^{-2}$, and we define for $p\in \mathbb R^3$ the fiber Hamiltonian $\mathbb H_p$ acting only on the Fock space $\mathcal{F}$ as
\begin{align}
\label{Eq:Introduction_Fiber_Hamiltonian}
    \mathbb H_p:=\left( p - \mathcal P\right)^2+\mathcal{N}-a^*(v)-a(v).
\end{align}
Since $\frac{1}{i}\nabla_x$ is the Lee-Low-Pines transform of the total momentum operator $\mathbb P$, we can therefore express $E_\alpha(p)$ as the ground state energy of the fiber operator $\mathbb H_p$, i.e.
\begin{align*}
    E_\alpha(p)=\inf \sigma(\mathbb H_p).
\end{align*}
In order to establish the upper bound in Eq.~(\ref{Eq:Introduction_Upper_Bound_Energy-Momentum_Relation}), we first note that by the variational characterization of the ground state energy, $E_\alpha(p)$ is bounded from above by
\begin{align*}
    E_\alpha(p)\leq \left\langle \Psi_{\alpha,p},\mathbb H_p \Psi_{\alpha,p}\right\rangle_\mathcal{F},
\end{align*}
for any element $\Psi_{\alpha,p}\in \mathcal{F}$ with $\|\Psi_{\alpha,p}\|=1$. It is therefore enough to find a good trial state $\Psi_{\alpha,p}$, such that
\begin{align}
\label{Eq:Introduction_Trial_state_Upper_Bound}
    \left\langle \Psi_{\alpha,p},\mathbb H_p \Psi_{\alpha,p}\right\rangle_\mathcal{F}\leq E_\alpha+\frac{|p|^2}{2\alpha^4 m_{\mathrm{LP}}}+C\alpha^{-(4+\epsilon)}|p|^2.
\end{align}
Naively, the most natural candidate would be $\Psi_{\alpha,p}:=\Psi_\alpha$, where $\Psi_\alpha$ is the ground state of the operator $\mathbb H_0$, i.e. $\langle \Psi_\alpha,\mathbb H_0 \Psi_\alpha\rangle_\mathcal{F}=E_\alpha$ and $\|\Psi_\alpha\|=1$. In this case however 
\begin{align*}
    \left\langle \Psi_{\alpha,p},\mathbb H_p \Psi_{\alpha,p}\right\rangle_\mathcal{F}=E_\alpha+|p|^2
\end{align*}
would not have the right quartic scaling in $\alpha$, and would only yield a lower bound of the form $m_\mathrm{eff}(\alpha)\geq \frac{1}{2}$, i.e. we would not even observe the divergence of the effective mass. The problem is of course, that the term $\left( p - \mathcal P\right)^2$ in Eq.~(\ref{Eq:Introduction_Fiber_Hamiltonian}) contains a $p$ without the right $\alpha$ scaling. In order to obtain an improved trial state $\Psi_{\alpha,p}$, let $\mathbb B=(\mathbb B_1,\mathbb B_2,\mathbb B_3)$ be a (pseudo) boost for $\mathcal P$, i.e. a collection of commuting self-adjoint operators $\mathbb B_j$ acting on $\mathcal{F}$ that satisfies for $p\in \mathbb R^3$
\begin{align}
\label{Eq:Introduction_Boost_Relations_For_P}
  e^{-ip\cdot \mathbb B}\mathcal Pe^{ip\cdot \mathbb B} \approx \mathcal P+p,
\end{align}
or equivalently for $z\in \mathbb R^3$
\begin{align}
\label{Eq:Introduction_Boost_Relations}
    e^{iz\cdot \mathcal P}\mathbb Be^{-iz\cdot \mathcal P} \approx \mathbb B+z,
\end{align}
and define $\Psi_{\alpha,p}:=e^{ip\cdot \mathbb B}\Psi_\alpha$. We write $\approx$, respectively pseudo-boost, as it is not possible to satisfy Eq.~(\ref{Eq:Introduction_Boost_Relations_For_P}) and Eq.~(\ref{Eq:Introduction_Boost_Relations}) as an identity. With this choice we have a cancellation of the unwanted $p$-term
\begin{align*}
    \left\langle \Psi_{\alpha,p}, \left( p - \mathcal P\right)^2\Psi_{\alpha,p}\right\rangle_{\mathcal{F}}\approx \left\langle \Psi_{\alpha},  \mathcal P^2 \Psi_{\alpha}\right\rangle_{\mathcal{F}}.
\end{align*}

For the purpose of finding the correct boost $\mathbb B$ that yields the desired upper bound in Eq.~(\ref{Eq:Introduction_Trial_state_Upper_Bound}), let us formally apply the semi-classical correspondence principle, which tells us that the energy of the quantum system corresponds to the Peka functional $\mathcal{F}^\mathrm{Pek}$ defined in Eq.~(\ref{Eq:Pekar_Functional}), the ground state $\Psi_\alpha$ corresponds to a minimizer $\varphi^\mathrm{Pek}$ of $\mathcal{F}^\mathrm{Pek}$, the ground state energy $E_\alpha$ corresponds to the minimum $e^\mathrm{Pek}$ of $\mathcal{F}^\mathrm{Pek}$ and the unitary transformation $e^{ip\cdot \mathbb B}$ corresponds formally to a  Hamiltonian flow $p\mapsto \varphi_p\in  L^2\! \left(\mathbb R^3\right)$ with $p\in \mathbb R^3$ and values in the infinite dimensional space $L^2\! \left(\mathbb R^3\right)$ 
\begin{align}
   \label{Eq:Introduction_Hamiltonian_Diff._Eq.}
   \begin{cases}
\alpha^2 \frac{\mathrm{d}}{\mathrm{d}p}\varphi_p=i\nabla_{\overline{\varphi}}B\! \left(\varphi(p)\right),\\
\varphi_0=\varphi^\mathrm{Pek},
\end{cases} 
\end{align}
for a suitable function $B:L^2\! \left(\mathbb R^3\right)\longrightarrow \mathbb R^3$ corresponding to the operator ${\mathbb  B}$. Since all of these considerations are of purely formal nature, we are not going to worry about the notion of differentiability. It is however worth pointing out that Eq.~(\ref{Eq:Introduction_Boost_Relations}) corresponds to
\begin{align}
\label{Eq:Introduction_Translation_Covariant}
    B(\varphi_z)=B(\varphi)+z
\end{align}
for $\varphi\neq 0$. A natural candidate for a function $B$ is the mean-value (we ignore at this point that this is ill-defined for some $\varphi\in L^2\! \left(\mathbb R^3\right)$) with respect to the measure $\|\varphi\|^{-2}|\varphi(y)|^2\mathrm{d}y$, i.e.
\begin{align}
\label{Eq:Introduction_Mean_Value_Boost}
    B(\varphi):=\|\varphi\|^{-2}\int_{\mathbb R^3} \! y\, |\varphi(y)|^2\mathrm{d}y.
\end{align}
This choice of $B$ clearly satisfies Eq.~(\ref{Eq:Introduction_Translation_Covariant}), and the solution of Eq.~(\ref{Eq:Introduction_Hamiltonian_Diff._Eq.}) reads
\begin{align*}
    \varphi_p(y)=e^{i\alpha^{-2}\|\varphi\|^{-2} p\cdot y}\varphi^\mathrm{Pek}(y).
\end{align*}
However, since $\nabla_{\overline{\varphi}}B\! \left(\varphi^\mathrm{Pek}\right)=\|\varphi\|^{-2} y \varphi^\mathrm{Pek}(y)$ is not an $L^2\! \left(\mathbb R^3\right)$ function, $\mathcal{F}^\mathrm{Pek}(\varphi_p)$ does not exhibit the right quadratic scaling in $p$. To be precise, one can show that there exists a constant $C>0$ such that for $p$ small enough
\begin{align*}
 \mathcal{F}^\mathrm{Pek}(\varphi_p)\geq e^\mathrm{Pek}+C\alpha^{-2}|p|.
\end{align*}

In order to make sure that $\nabla_{\overline{\varphi}}B\! \left(\varphi^\mathrm{Pek}\right)\in L^2\! \left(\mathbb R^3\right)$, let us take the mean-value only in a bounded region. For this purpose let $x^\lambda$ denote the $\lambda$-quantile of a measure on $\mathbb R$, let $\rho_j$ be the $j$-th marginal distribution of a measure $\rho$ on $\mathbb R^3$ and define $\mathrm{d}\rho_\varphi:=|\varphi(y)|^2\mathrm{d}y$. Then we introduce for $q>0$ in components
\begin{align}
\label{Eq:Introduction_regularized_median}
    m_q(\rho)_j: &=\left(\int_{x^{\frac{1}{2}- q}(\rho_j)}^{x^{\frac{1}{2}+ q}(\rho_j)} \mathrm{d}\rho_j \right)^{-1}\int_{x^{\frac{1}{2}- q}(\rho_j)}^{x^{\frac{1}{2}+ q}(\rho_j)} y \mathrm{d}\rho_j(y),\\
    \label{Eq:H_B_as_reg._med.}
    B(\varphi): & =m_q(\chi* \rho_\varphi),
\end{align}
where $\chi\geq 0$ is a smooth function, see also Section \ref{Sec:Construction_of_a_Trial_State} for a comprehensive definition of $m_q$. The statistical quantity $m_q$ has been used previously in the study of translation-invariant Bose gases \cite{BS0} as well as in the study of the Fröhlich polaron \cite{BS1,BS2}, where it was referred to as the regularized median. Clearly the choice of $B$ in Eq.~(\ref{Eq:H_B_as_reg._med.}) satisfies Eq.~(\ref{Eq:Introduction_Translation_Covariant}) and since
\begin{align*}
\nabla_{\overline{\varphi}}B\! \left(\varphi^\mathrm{Pek}\right)\in L^2\! \left(\mathbb R^3\right),    
\end{align*}
 we have the right scaling
\begin{align}
\label{Eq:Introduction_Expansion_Pekar_Energy}
    \mathcal{F}^\mathrm{Pek}(\varphi_p)= e^\mathrm{Pek}+C\alpha^{-4}|p|^2+o\! \left(\alpha^{-4}|p|^2\right),
\end{align}
for a suitable constant $C>0$, where we write $o\! \left(\alpha^{-4}|p|^2\right)$ for a term that is small compared to $\alpha^{-4}|p|^2$ as $p\rightarrow 0$ and $\alpha\rightarrow \infty$. 

Finally, we want to modify $B$, such that the corresponding constant $C$ in Eq.~(\ref{Eq:Introduction_Expansion_Pekar_Energy}) is the optimal one $C=\frac{1}{2m_\mathrm{Pek}}$. We note that the Landau-Pekar constant $m_\mathrm{Pek}$ emerges naturally as
\begin{align}
\label{Eq:Classical_Polaron}
     \mathcal{F}^\mathrm{Pek}\! \left(\varphi^\mathrm{Pek}-i\frac{p}{\alpha^2 m_{\mathrm{LP}}}\cdot \nabla \varphi^\mathrm{Pek}\right)=e^\mathrm{Pek}+\frac{|p|^2}{2\alpha^4 m_{\mathrm{LP}}}+o\! \left(\alpha^{-4}|p|^2\right),
\end{align}
see for example \cite{FRS}, where the effective mass is identified for the classical polaron described by the Landau-Pekar equations. Motivated by Eq.~(\ref{Eq:Classical_Polaron}), we want to find a function $ B$ satisfying Eq.~(\ref{Eq:Introduction_Translation_Covariant}) and 
\begin{align}
\label{Eq:Introduction_Boost_Condition}
   \nabla_{\overline{\varphi}}B\! \left(\varphi^\mathrm{Pek}\right)=-\frac{1}{m_\mathrm{LP}} \nabla \varphi^\mathrm{Pek}. 
\end{align}
In the following let us define for $f:\mathbb R^3\longrightarrow \mathbb R^3$ the modified boost $B$ as
\begin{align}
\label{Eq:Introduction_H_B_with_modifier}
     B(\varphi):  =m_q(\chi* \rho_\varphi)+\int_{\mathbb R^3} f\Big(y-m_q(\chi* \rho_\varphi)\Big)\mathrm{d}\rho_\varphi(y),
\end{align}
which satisfies Eq.~(\ref{Eq:Introduction_Translation_Covariant}). An elementary computation shows that 
\begin{align}
\label{Eq:Introduction_Gradient_Modified_result}
   \nabla_{\overline{\varphi}}B\! \left(\varphi^\mathrm{Pek}\right)=\nabla_{\overline{\varphi}}m_q(\chi* \rho_\varphi)|_{\varphi^\mathrm{Pek}}+f\varphi^\mathrm{Pek}
\end{align}
in case $f\varphi^\mathrm{Pek}\perp \nabla \varphi^\mathrm{Pek}$, i.e. we can write any vector in the affine space
\begin{align}
\label{Eq:Introduction_Affine_Space}
    \xi\in \mathcal{A}:=\nabla_{\overline{\varphi}}m_q(\chi* \rho_\varphi)|_{\varphi^\mathrm{Pek}}+\{\nabla \varphi^\mathrm{Pek}\}^\perp
\end{align}
as a gradient $\xi= \nabla_{\overline{\varphi}}B\! \left(\varphi^\mathrm{Pek}\right)$ of a boost $B$ with the right choice of $f$. Notably, 
\begin{align*}
  -\frac{1}{m_{\mathrm{LP}}} \nabla \varphi^\mathrm{Pek}\in \mathcal{A}  
\end{align*}
is an element of $\mathcal{A}$, and therefore we can find a boost $B$ such that Eq.~(\ref{Eq:Introduction_Boost_Condition}) holds. Hence
\begin{align}
\label{Eq:Introduction_Boost_Condition_Satisfied}
    \mathcal{F}^\mathrm{Pek}(\varphi_p)= e^\mathrm{Pek}+\frac{|p|^2}{2\alpha^4 m_{\mathrm{LP}}}+o\! \left(\alpha^{-4}|p|^2\right) .
\end{align}

Having the correct semi-classical boost $B$ at hand such that Eq.~(\ref{Eq:Introduction_Boost_Condition_Satisfied}) holds, we are going to construct the corresponding boost operator ${\mathbb  B}$ and define the trial state
\begin{align*}
   \Psi_{\alpha,p}:=e^{ip\cdot {\mathbb  B}}\Psi_\alpha. 
\end{align*}
There are two main technical challenges in this manuscript. The first one is to establish a rigorous correspondence between the quantum objects $\Psi_{\alpha,p},\mathbb H_p,{\mathbb  B},\dots $ and their semi-classical counterparts $\varphi_p, \mathcal{F}^\mathrm{Pek},B,\dots $, i.e. we have to verify that
\begin{align}
\label{Eq:Introduction_Correspondence_rigorous}
    \left\langle \Psi_{\alpha,p},\mathbb H_p \Psi_{\alpha,p}\right\rangle_\mathcal{F}-E_\alpha = \mathcal{F}^\mathrm{Pek}(\varphi_p) - e^\mathrm{Pek}+o\! \left(\alpha^{-4}|p|^2\right)=\frac{|p|^2}{2\alpha^4 m_{\mathrm{LP}}}+o\! \left(\alpha^{-4}|p|^2\right).
\end{align}
In \cite{BS1}, it has been shown that suitable low energy states satisfy Bose-Einstein condensation, which loosely speaking means that a typical phonon configuration is expected to be close to the semi-classical optimizer $\varphi^\mathrm{Pek}$. In order to derive Eq.~(\ref{Eq:Introduction_Correspondence_rigorous}), we are going to improve this result, by showing that the empirical distribution $\rho_{Y_n}$ of a collection of phonons at position $Y_n=(y_1,\dots ,y_n)$ is expected to be close to the Pekar measure $\mathrm{d}\rho^\mathrm{Pek}:=|\varphi^\mathrm{Pek}|^2\mathrm{d}y$ with respect to the distance
\begin{align}
\label{Eq:Introduction_Distance}
\|\chi*(\rho_{Y_n}-\rho^\mathrm{Pek})\|_{\mathrm{TV}},   
\end{align}
where $\|\cdot \|_{\mathrm{TV}}$ refers to the total variation and $\chi$ is a mollifier. Furthermore, we show that the notion of distance in Eq.~(\ref{Eq:Introduction_Distance}) is the correct one to discuss continuity and differentiability aspects of the statistical quantity $m_q(\chi* \rho_\varphi)$ used in the definition of $B$ in Eq.~(\ref{Eq:Introduction_H_B_with_modifier}). The second technical difficulty concerns the Taylor expansion of the unitary group $e^{ip\cdot \mathbb B}$, which is necessary to express the left hand side of Eq.~(\ref{Eq:Introduction_Correspondence_rigorous}). To be precise, we need to show that the moments of the boost operator ${\mathbb  B}$ are bounded in the ground state $\Psi_\alpha$, i.e. we have to show that there exist constants $C_m$ such that for all $\alpha$ large enough
\begin{align}
\label{Eq:Introduction_moments_boost}
    \left\langle \Psi_\alpha, |{\mathbb  B}|^m \Psi_\alpha \right\rangle\leq C_m.
\end{align}
For this purpose, we use a layer-cake representation of $ \left\langle \Psi_\alpha, |{\mathbb  B}|^m \Psi_\alpha \right\rangle$ and show that the ground state has to be supported mostly on the spectral subspace where $|{\mathbb  B}|$ is small. Related to the question of whether $\nabla_{\overline{\varphi}}B\! \left(\varphi^\mathrm{Pek}\right)$ is an element of $L^2\! \left(\mathbb R^3\right)$, we want to emphasise that Eq.~(\ref{Eq:Introduction_moments_boost}) certainly does not hold for any boost ${\mathbb  B}$. For example the mean-value boost in Eq.~(\ref{Eq:Introduction_Mean_Value_Boost}) has already on the semi-classical level unbounded second moments
\begin{align*}
    \int_{\mathbb R^3}|y|^2 \mathrm{d}\rho^\mathrm{Pek}=\|y\varphi^\mathrm{Pek}\|^2=\infty.\\
\end{align*}

\textbf{Outline.} The paper is structured as follows. In Section \ref{Sec:Fock_Space_Formalism} we introduce Fock space formalism, including all the objects appearing in the definition of $\mathbb H$, as well as the most important semi-classical objects. In the subsequent Section \ref{Sec:Construction_of_a_Trial_State} we will introduce an appropriate boost operator as a counterpart to the semi-classical function $B$ defined in Eq.~(\ref{Eq:Introduction_H_B_with_modifier}) and give a precise definition of the trial state $\Psi_{\alpha,p}$. Furthermore, we will isolate the main contribution in the energy $\left\langle \Psi_{\alpha,p},\mathbb H_p \Psi_{\alpha,p}\right\rangle_\mathcal{F}$ from various error terms $\mathcal{E}$. The proof of Theorem \ref{Th:Main} is then content of Section \ref{Sec:Proof_of_Theorem}, where we confirm the formal semi-classical predictions. In the rest of the paper we will derive technical results necessary for the proof of Theorem \ref{Th:Main}, starting with the analysis of the regularized median $m_q$, see Eq.~(\ref{Eq:Introduction_regularized_median}), in Section \ref{Sec:Analysis_of_the_generalized_Median}. In the following Section \ref{Sec:Spatial_concentration_of_Probability} we will discuss the spatial distribution of the ground state $\Psi_\alpha$ and establish that $\Psi_\alpha$ is mostly supported on the spectral subspace where the boost $|{\mathbb  B}|$ is small. In a similar spirit, we verify in Section \ref{Sec:Asymptotic_concentration_of_Probability} that the ground state $\Psi_\alpha$ is mostly supported on polarization fields $\varphi$ close to the semi-classical minimizer $\varphi^\mathrm{Pek}$. The analysis of the error terms $\mathcal{E}$ from Section \ref{Sec:Construction_of_a_Trial_State} is then content of Section \ref{Sec:Analysis_of_the_Error_Terms} and in Section \ref{Sec:Bose-Einstein_Condensation} we will derive improved results on the issue of Bose-Einstein condensation. Especially, we will verify that the empirical measures $\rho_{Y_n}$ are close to the semi-classical distribution $\rho^\mathrm{Pek}$ with respect to the regularized total variation. In the Appendix \ref{Appendix:Properties_of_the_Semi-Classical_minimizers} we are going to verify various decay properties of the semi-classical minimizers, in momentum space as well as position space.

\section{Definitions and Fock Space Formalism}
\label{Sec:Fock_Space_Formalism}

\subsection{Basic Notation}
\label{Subsec:Basic_Notation}
It is the content of this Subsection, to introduce the Fock space $\mathcal{F}$, as well as two important families of operators defined on $\mathcal{F}$. Furthermore, we are presenting a collection of  useful cut-off functions. 

In order to define $\mathcal{F}$, let $\underset{n\in \mathbb N}{\bigcup}\mathbb R^{3n}$ be the (disjoint) union of the topological spaces $\mathbb R^{3n}$, and let us denote with $L^2\! \left(\underset{n\in \mathbb N}{\bigcup}\mathbb R^{3n}\right)$ the space of all Borel measurable complex valued (equivalence classes of) functions 
\begin{align*}
 \Psi:\underset{n\in \mathbb N}{\bigcup}\mathbb R^{3n}\longrightarrow \mathbb C   ,
\end{align*}
such that $ \sum_{n=0}^\infty \int_{\mathbb R^{3n}}\left|\Psi(Y_n)\right|^2\mathrm{d}Y_n<\infty$, equipped with the inner product
\begin{align*}
    \Big\langle\Psi,\Phi\Big\rangle_{L^2\! \left(\underset{n\in \mathbb N}{\bigcup}\mathbb R^{3n}\right)}:=\sum_{n=0}^\infty \int_{\mathbb R^{3n}}\overline{\Psi(Y_n)}\Phi(Y_n)\mathrm{d}Y_n.
\end{align*}
For the readers convenience, we will indicate the dimension of an element $Y_n\in \underset{n\in \mathbb N}{\bigcup}\mathbb R^{3n}$ in our notation, i.e. $Y_n\in \mathbb R^{3n}$, and we will usually structure the vector $Y_n$ as 
\begin{align*}
    Y_n=(y_1,\dots ,y_n),
\end{align*}
with $y_i\in \mathbb R^3$. Clearly, the Hilbert space $L^2\! \left(\underset{n\in \mathbb N}{\bigcup}\mathbb R^{3n}\right)$ is isomorphic to the orthogonal sum
\begin{align*}
    L^2\! \left(\underset{n\in \mathbb N}{\bigcup}\mathbb R^{3n}\right)\cong \underset{n\in \mathbb N}{\bigoplus} L^2\! \left(\mathbb R^{3n}\right).
\end{align*}
We then define the Fock space $\mathcal{F}$ as the subspace of $L^2\! \left(\underset{n\in \mathbb N}{\bigcup}\mathbb R^{3n}\right)$ containing all permutation symmetric functions $\Psi$, i.e. $\Psi\in \mathcal{F}$ if and only if for all $Y_n\in \underset{n\in \mathbb N}{\bigcup}\mathbb R^{3n}$ and permutations $\sigma\in S_n$
\begin{align*}
    \Psi(y_{\sigma_1},\dots ,y_{\sigma_n})=\Psi(y_1,\dots ,y_n).
\end{align*}
In case we want to emphasise that the functions involved in an inner product are permutation symmetric, we will write $\Big\langle \Psi,\Phi\Big\rangle_{\mathcal{F}}$ instead of $\Big\langle \Psi,\Phi\Big\rangle_{L^2 \! \left(\underset{n\in \mathbb N}{\bigcup}\mathbb R^{3n}\right)}$. In the following it will be useful to introduce for $\alpha>0$ and $Y_n\in \underset{n\in \mathbb N}{\bigcup}\mathbb R^{3n}$ the empirical measure
\begin{align}
\label{Eq:Empirical_Measure}
    \rho_{Y_n}=\alpha^{-2}\sum_{j=1}^n \delta_{y_j},
\end{align}
where we suppress the $\alpha$-dependence of $ \rho_{Y_n}$ in our notation for the sake of readability and use the convention that $\rho_{Y_0}=0$. Then $\Psi\in L^2\! \left(\underset{n\in \mathbb N}{\bigcup}\mathbb R^{3n}\right)$ is an element of $\mathcal{F}$, if and only if it can be written as a function of the empirical measures, i.e. if and only if there exists a complex valued function $\Psi'$ taking measures on $\mathbb R^3$ as an argument, such that
\begin{align*}
    \Psi(Y_n)=\Psi'(\rho_{Y_n}).
\end{align*}
For a measurable subset $\Omega\subseteq \mathbb R^3$, let us furthermore introduce the Fock space $\mathcal{F}(\Omega)$ as the space of all $\Psi\in \mathcal{F}$ satisfying
\begin{align*}
    \mathrm{supp}(\Psi)\subseteq \underset{n\in \mathbb N}{\bigcup}\Omega^n .
\end{align*}

The first important family of operators on the spaces $L^2\! \left(\underset{n\in \mathbb N}{\bigcup}\mathbb R^{3n}\right)$ and $\mathcal{F}$ are multiplication operators by a given measurable function $F:\underset{n\in \mathbb N}{\bigcup}\mathbb R^{3n}\longrightarrow \mathbb R$. In the following we will use the notation $F(Y_n)$ for the multiplication operator by $F$, e.g. we will write
\begin{align*}
    \Big\langle \Psi,F(Y_n)\Psi\Big\rangle_{L^2 \! \left(\underset{n\in \mathbb N}{\bigcup}\mathbb R^{3n}\right)}=\sum_{n=0}^\infty \int_{\mathbb R^{3n}}F(Y_n)\left|\Psi(Y_n)\right|^2\mathrm{d}Y_n.
\end{align*}
Many of the functions considered in this manuscript will naturally depend only on the empirical measure $\rho_{Y_n}$, i.e. will be of the form $F(Y_n)=f(\rho_{Y_n})$, in which case we might write
\begin{align*}
    \Big\langle \Psi,f(\rho_{Y_n})\Psi\Big\rangle_{\mathcal{F}}= \Big\langle \Psi,f(\rho_{Y_n})\Psi\Big\rangle_{L^2 \! \left(\underset{n\in \mathbb N}{\bigcup}\mathbb R^{3n}\right)}=\sum_{n=0}^\infty \int_{\mathbb R^{3n}}f(\rho_{Y_n})\left|\Psi(Y_n)\right|^2\mathrm{d}Y_n.
\end{align*}
An especially important example of a multiplication operator is the (re-scaled) particle number operator $\mathcal{N}$ defined as
\begin{align}
\label{Eq:Definition_Particle_Number}
    \mathcal{N}:=\mathcal{N}(\rho_{Y_n}):=\int \mathrm{d}\rho_{Y_n}=\frac{n}{\alpha^2}.
\end{align}
The second essential family of operators on a Fock space are the creation and annihilation operators $a^*(f)$ and $a(f)$, parameterized by functions $f\in L^2\! \left(\mathbb R^3\right)$. In order to properly introduce $a^*(f)$, let us first define the operator $L(f)$ acting on $L^2 \! \left(\underset{n\in \mathbb N}{\bigcup}\mathbb R^{3n}\right)$ as
\begin{align}
\label{Eq:Definition_L_Operator}
    \Big(L(f)\Psi\Big)(Y_n):=\frac{\sqrt{n}}{\alpha}\Psi(Y_{n-1})f(y_n).
\end{align}
Here we have used the convention that $Y_n=(y_1,\dots ,y_n)$ and $Y_{n-1}=(y_1,\dots ,y_{n-1})$, a notion which we will use repeatedly in this article. We can now define on $\mathcal{F}$
\begin{align*}
    a^*(f): & =\pi_{\mathcal{F}} L(f)\iota_{\mathcal{F}},\\
    a(f) : & =\pi_{\mathcal{F}} L(f)^* \iota_{\mathcal{F}},
\end{align*}
where $A^*$ is the adjoint of an operator $A$, $\iota_{\mathcal{F}}$ is the embedding of $\mathcal{F}\subseteq L^2 \! \left(\underset{n\in \mathbb N}{\bigcup}\mathbb R^{3n}\right)$ and $\pi_\mathcal{F}=\iota_{\mathcal{F}}^*$ is the orthogonal projection onto $\mathcal{F}$.

The operators $a^*(f)$ and $a(f)$ will be a convenient tool to express many other operators acting on $\mathcal{F}$, e.g. for a bounded and measurable function $h:\mathbb R^3\longrightarrow \mathbb R$ and an orthonormal basis $\{\varphi_n:n\in \mathbb N\}$ we can express the multiplication operator $Y_n\mapsto \int h\mathrm{d}\rho_{Y_n}$ as
\begin{align}
\label{Eq:Multiplication_Operator_in_terms_of_a}
    \int h\mathrm{d}\rho_{Y_n}=\sum_{n,m=0}^\infty \braket{\varphi_m,h \varphi_n}a^*(\varphi_m)a(\varphi_n).
\end{align}
Therefore, we see that the definition of the particle number operator in Eq.~(\ref{Eq:Definition_Particle_Number}) and the definition given previously in Eq.~(\ref{Eq:Definition_Particle_Number_II}) coincide.\\

Before giving a proper introduction of the Fröhlich Hamiltonian in the following Subsection \ref{Subsec:Hamiltonian}, we shall introduce some useful notation. Following \cite{BS1,BS2}, we are going to introduce for some $\kappa\geq 0$ and $-\infty\leq a<b\leq \infty$, and a function $f:U\longrightarrow \mathbb R$ defined on some set $U$, cut-off functions 
\begin{align*}
\chi_\kappa(a\leq f\leq b):
    \begin{cases}
        U\longrightarrow \mathbb R, \\
        t\mapsto \chi_\kappa(a\leq f(t)\leq b),
    \end{cases}
\end{align*}
where $\kappa\geq 0$ is going to determine the sharpness of the cut-off. Usually we will take $U$ as the set $\mathbb R$, $\underset{n\in \mathbb N}{\bigcup}\mathbb R^{3n}$ or the set of all Borel measures on $\mathbb R^3$. In order to introduce $\chi_\kappa(a\leq f\leq b)$, let 
\begin{align*}
  \beta_1,\beta_2:\mathbb R\cup \{-\infty,\infty\}\longrightarrow [0,1]  
\end{align*}
be smooth functions on $\mathbb R$ such that $\beta_1^2+\beta_2^2=1$, $\mathrm{supp}\! \left(\beta_1\right)\subseteq [-\infty,1)$ and $\mathrm{supp}\! \left(\beta_2\right)\subseteq (-1,\infty]$. Then we define $\chi_\kappa(a\leq f\leq b)$ as
\begin{align}
\label{Eq:General_Cut-Off_Functions}
    \chi_\kappa(a\leq f(t)\leq b):=
        \begin{cases}
        \beta_1\!\left(\frac{f(t)-b}{\kappa}\right)\beta_2\!\left(\frac{f(t)-a}{\kappa}\right),\ \mathrm{ for } \ \kappa>0 \\
        \mathds{1}_{[a,b]}(f(t)), \ \mathrm{ for } \ \kappa=0.
    \end{cases}
\end{align}
Clearly,
\begin{align*}
\chi_\kappa(a\leq f(t)\leq b) & =1 \ \ \mbox{in case} \ \ f(t)\in [a+\kappa,b-\kappa],\\
    \chi_\kappa(a\leq f(t)\leq b) & =0 \ \ \mbox{in case} \ \ f(t)\notin [a-\kappa,b+\kappa].
\end{align*}
Furthermore, we define for a self-adjoint operator $T$ the contraction
\begin{align*}
  \chi_\kappa(a\leq T\leq b)  
\end{align*}
by the means of spectral calculus, and we write $\chi(a\leq f\leq b)$ in case $\kappa=0$ and $\chi_\kappa(a\leq f)$, respectively $\chi_\kappa(f\leq b)$, in case $b=\infty$, respectively $a=-\infty$. Note that for $\kappa>0$
\begin{align*}
  t\mapsto \chi_\kappa(a\leq t\leq b)  
\end{align*}
is a smooth function and in case the additional condition $a+\kappa<b-\kappa$ is satisfied we have that
\begin{align*}
    \sqrt{1-\chi_\kappa(a\leq t\leq b)^2}=\sqrt{\chi_\kappa(t\leq a)^2+\chi_\kappa(b\leq t)^2}=\chi_\kappa(t\leq a)+\chi_\kappa(b\leq t)
\end{align*}
is a smooth function as well.

\subsection{The Fröhlich Hamiltonian}
\label{Subsec:Hamiltonian}
With the definitions of the previous Subsection \ref{Subsec:Basic_Notation} at hand, we are in a position to properly introduce the polaron model. The Hilbert space under consideration is given by the tensor product
\begin{align*}
    L^2\! \left(\mathbb R^3\right)\otimes \mathcal{F}\subseteq L^2 \! \left(\mathbb R^3\times \underset{n\in \mathbb N}{\bigcup}\mathbb R^{3n}\right),
\end{align*}
which we can naturally identify as a subspace of $L^2 \! \left(\mathbb R^3\times \underset{n\in \mathbb N}{\bigcup}\mathbb R^{3n}\right)$. Here the first factor $\mathbb R^3$ is the space of the electron coordinate $x\in \mathbb R^3$ and $\underset{n\in \mathbb N}{\bigcup}\mathbb R^{3n}$ contains the coordinates $Y_n=(y_1,\dots ,y_n)$ of the elementary excitations of the polarization field. For a function $\Psi\in L^2 \! \left(\mathbb R^3\times \underset{n\in \mathbb N}{\bigcup}\mathbb R^{3n}\right)$, we will usually write $\Psi(x;Y_n)$ and $\nabla_x$ as well as $\Delta_x$ for the gradient, respectively the Laplace operator, in the direction of $x$. Similarly, we define
\begin{align}
\label{Eq:P_f_Operator}
    (\mathcal P\Psi)(x;Y_n):=\sum_{j=1}^n \frac{1}{i}\nabla_{y_j}\Psi(x;Y_n).
\end{align}
Note that with these notations at hand, the total momentum operator $\mathbb P$ reads
\begin{align*}
    \mathbb P=\frac{1}{i}\nabla_x+\mathcal P.
\end{align*}
With $\Delta_x,\mathcal{N}$, $a^*(f)$ and $a(f)$, we have all objects at hand in order to introduce 
\begin{align*}
     \mathbb H:=-\Delta_x+\mathcal{N}-a^*(v_x)-a(v_x),
\end{align*}
see Eq.~(\ref{Eq:First_Definition_Hamiltonian}). However, since $v_x\notin L^2\! \left(\mathbb R^3\right)$, due to an ultraviolet singularity, the definition in Eq.~(\ref{Eq:First_Definition_Hamiltonian}) is only formal. In order to properly define $\mathbb H$, let us first construct a quadratic form $\mathcal{Q}$, acting on the space of permutation symmetric compactly supported $C^\infty$ functions as
\begin{align*}
    \mathcal{Q}(\Psi):=\left\langle \Psi, \left(-\Delta_x+\mathcal{N}\right)\Psi\right\rangle-2\mathfrak{Re}\sum_{n=1}^\infty \frac{\sqrt{n}}{\alpha}\int_{\mathbb R^3}\int_{\mathbb R^{3n}} v(y_n-x)\overline{\Psi(x;Y_n)}\Psi(x;Y_{n-1})\mathrm{d}Y_n \mathrm{d}x.
\end{align*}
Since $\Psi$ is compactly supported, there exists a $n_*$ such that $\Psi(Y_n)=0$ for $n\geq n_*$ and $(x;Y_n)\mapsto \overline{\Psi(x;Y_n)}\Psi(x;Y_{n-1})$ is a bounded function with compact support for all $n\in \mathbb N$. Together with the fact that $v$ is locally $L^1$, we observe that $ \mathcal{Q}(\Psi)$ is well defined. Furthermore, it follows from \cite{LY}, that the quadratic form $\mathcal{Q}$ is bounded from below and closable, and therefore defines a self-adjoint operator $\mathbb H$. Similarly, we can give meaning to the formal definitions
\begin{align*}
      \mathbb H_{\frac{1}{i}\nabla_x} & :=\left( \frac{1}{i}\nabla_x- \mathcal P\right)^2+\mathcal{N}-a^*(v)-a(v), \\
    \mathbb H_p & :=\left( p - \mathcal P\right)^2+\mathcal{N}-a^*(v)-a(v).
\end{align*}
It is easy to see that $\mathbb H$ is unitary equivalent to $\mathbb H_{\frac{1}{i}\nabla_x}$ by centering the coordinate frame of the polarization field at the position of the electron $x$, i.e. we define the unitary map
\begin{align}
\label{Eq:Essentially_LLP}
    (\mathcal{T}\Psi)(x;Y_n):=\Psi(x;y_1-x,\dots ,y_n-x)
\end{align}
and observe that $\mathbb H_{\frac{1}{i}\nabla_x}=\mathcal{T}^* \mathbb H \mathcal{T}$. In this manuscript, we will use both representations $\mathbb H$ and $\mathbb H_{\frac{1}{i}\nabla_x}$, where we prefer $\mathbb H$ in case we want to establish a semi-classical correspondence and use $\mathbb H_{\frac{1}{i}\nabla_x}$ in case we want to work with a proper ground state $\Psi_\alpha$ of the fiber Hamiltonian $\mathbb H_0$.

Finally, let us introduce the modified versions $\mathbb H^K$ and $\mathbb H^K_{\frac{1}{i}\nabla_x}$, which include an ultraviolet regularization. For this purpose let us use for $K>0$ the projection $\chi(|\nabla|\leq K)$, see Eq.~(\ref{Eq:General_Cut-Off_Functions}), in order to define the regularized interaction
\begin{align*}
    v^K:=\chi(|\nabla|\leq K)v.
\end{align*}
With $v^K$ at hand we define the operators
\begin{align}
    \label{Eq:Cut-off_Hamiltonian}
    \mathbb H^K : & = -\Delta_x+\mathcal{N}-a^*(v^K_x)-a(v^K_x),\\
    \label{Eq:Cut-off_Hamiltonian_LLP}
     \mathbb H^K_{\frac{1}{i}\nabla_x}: & =\left( \frac{1}{i}\nabla_x- \mathcal P\right)^2+\mathcal{N}-a^*(v^K)-a(v^K),
\end{align}
which satisfy  $\mathbb H^K_{\frac{1}{i}\nabla_x}=\mathcal{T}^* \mathbb H^K \mathcal{T}$.

\subsection{Semi-Classical Objects}
\label{Subsec:Semi-Classical_Objects}
In this Subsection we want to introduce the most important objects appearing in the semi-classical correspondence principle. For this purpose let us introduce the vacuum state $\Xi\in \mathcal{F}$ as $\Xi(Y_n):=0$ for $n\geq 1$ and $\Xi(Y_0)=1$ for $Y_0\in \mathbb R^0=\{0\}$, and let us define for $f\in L^2\! \left(\mathbb R^3\right)$ the unitary Weyl transformation $W_f$ acting on $\mathcal{F}$ as
\begin{align}
\label{Eq:Def_Weyl_Operator}
    W_f:=e^{\alpha^2 a(f)-\alpha^2 a^*(f)}.
\end{align}
Note that, up to a phase, the Weyl transformations $W_f$ are characterized by the law
\begin{align*}
    W_g^* a(f) W_g=a(f)-\braket{f,g}.
\end{align*}
With $\Xi$ and $W_f$ at hand, we can introduce for $\psi,\varphi\in L^2\! \left(\mathbb R^3\right)$, where $\|\psi\|=1$, a family of semi-classical trial states
\begin{align}
\label{Eq:Product_with_coherent_state}
    \psi\otimes \Xi_\varphi:=\psi\otimes \left(W_\varphi^* \Xi\right)\in L^2\! \left(\mathbb R^3\right)\otimes \mathcal{F}.
\end{align}
Notably, the coherent states $\Xi_\varphi$ are joint eigenfunctions of the annihilation operators $a(f)$, i.e. 
\begin{align}
\label{Eq:Coherent_State_Eigenequation}
    a(f)\Xi_\varphi=\braket{f,\varphi}\Xi_\varphi,
\end{align}
a property which characterizes them up to a constant. Together with Eq.~(\ref{Eq:Multiplication_Operator_in_terms_of_a}), we observe the following correspondence between the measure $\mathrm{d}\rho(y):=|\varphi(y)|^2\mathrm{d}y$, a semi-classical object, and the empirical measure 
\begin{align}
\label{Eq:Demonstrating_rho_correspondence}
    \left\langle \Xi_\varphi, \int h\mathrm{d}\rho_{Y_n} \Xi_\varphi \right\rangle_{\mathcal{F}}=\int h\mathrm{d}\rho_\varphi.
\end{align}
Furthermore, by Eq.~(\ref{Eq:Coherent_State_Eigenequation}) and Eq.~(\ref{Eq:Demonstrating_rho_correspondence}), we obtain that the energy of $\psi\otimes \Xi_\varphi$ is given by
\begin{align*}
    \mathcal{E}^\mathrm{Pek}(\psi,\varphi):=\left\langle   \psi\otimes \Xi_\varphi ,\mathbb H  \psi\otimes \Xi_\varphi\right\rangle_{L^2\! \left(\mathbb R^3\right)\otimes \mathcal{F}}=\|\varphi\|^2+\left\langle \psi, \left(-\Delta_x-2 v*\mathfrak{Re}[\varphi]\right)\psi\right\rangle.
\end{align*}
Both $\mathbb H$ as well as $\Xi_\varphi=W_\varphi \Xi$ depend on the parameter $\alpha>0$, since the definition of $a^*(f)$ and $a(f)$, respectively $L(f)$ in Eq.~(\ref{Eq:Definition_L_Operator}), is $\alpha$ dependent, however the energy
\begin{align*}
 \mathcal{E}^\mathrm{Pek}(\psi,\varphi)=\|\varphi\|^2+\left\langle \psi, \left(-\Delta_x-2 v*\mathfrak{Re}[\varphi]\right)\psi\right\rangle,   
\end{align*}
is $\alpha$ independent. Minimizing over all possible electron functions $\psi$ yields the semi-classical Pekar functional
\begin{align*}
    \mathcal{F}^\mathrm{Pek}(\varphi):=\inf_{\psi\in L^2\! \left(\mathbb R^3\right):\|\psi\|=1} \mathcal{E}^\mathrm{Pek}(\psi,\varphi)=\|\varphi\|^2+\inf \sigma \! \left(-\Delta_x-2 v*\mathfrak{Re}[\varphi]\right).
\end{align*}
According to \cite{Li}, the functional $\mathcal{E}^\mathrm{Pek}$ has unique minimizers $(\psi^\mathrm{Pek},\varphi^\mathrm{Pek})$ such that $\psi^\mathrm{Pek}$ is non-negative and $\varphi^\mathrm{Pek}$ is radial, and all other minimizers are of the form $(e^{i\theta}\psi^\mathrm{Pek}_z,\varphi^\mathrm{Pek}_z)$ for $\theta\in [0,2\pi)$ and $z\in \mathbb R^3$. Consequently, $\varphi^\mathrm{Pek}$ is a minimizer of $\mathcal{F}$ as well, and we define the minimal Pekar energy $e^\mathrm{Pek}$ as
\begin{align}
\label{Eq:Variational_Definition}
    e^\mathrm{Pek}:=\mathcal{F}^\mathrm{Pek}(\varphi^\mathrm{Pek})=\inf_{\varphi \in L^2\! \left(\mathbb R^3\right)}\mathcal{F}^\mathrm{Pek}(\varphi)=\inf_{\varphi,\psi \in L^2\! \left(\mathbb R^3\right):\|\psi\|=1}\mathcal{E}^\mathrm{Pek}(\psi,\varphi).
\end{align}
Let us furthermore denote with $P^\mathrm{Pek}:L^2\! \left(\mathbb R^3\right)\longrightarrow L^2\! \left(\mathbb R^3\right)$ the orthogonal projection onto $\psi^\mathrm{Pek}$ and with $Q^\mathrm{Pek}:L^2\! \left(\mathbb R^3\right)\longrightarrow L^2\! \left(\mathbb R^3\right)$ the orthogonal projection onto $\{\psi^\mathrm{Pek}\}^\perp$, i.e. we define
\begin{align*}
    (P^\mathrm{Pek}\psi)(x): & =\int_{\mathbb R^3} \overline{\psi^\mathrm{Pek}(x')}\psi(x')\mathrm{d}x' \psi^\mathrm{Pek}(x), \\
    Q^\mathrm{Pek}: & =1-P^\mathrm{Pek},
\end{align*}
and let us define the Pekar measure $\rho^\mathrm{Pek}$ on $\mathbb R^3$ as $\mathrm{d}\rho^\mathrm{Pek}(y):=\left|\varphi^\mathrm{Pek}(y)\right|^2\mathrm{d}y$. 

The following Lemma \ref{Lem:Semiclassical_objects_properties} provides useful identities and estimates for the minimizers $(\psi^\mathrm{Pek},\varphi^\mathrm{Pek})$ of $\mathcal{E}^\mathrm{Pek}$ and the truncated interaction terms
\begin{align}
\label{Eq:Def_v_Lambda}
      v^\Lambda: & = \chi(|\nabla|\leq \Lambda)v, \\
\label{Eq:Def_w_Lambda}
    w^\Lambda : & = \frac{i\nabla}{\Delta}\chi(|\nabla|> \Lambda)v.
\end{align}
The proof of Lemma \ref{Lem:Semiclassical_objects_properties} will be the content of Appendix \ref{Appendix:Properties_of_the_Semi-Classical_minimizers}.
\begin{lem}
    \label{Lem:Semiclassical_objects_properties}
   We have the identities $|\psi^\mathrm{Pek}|^2*v=\varphi^\mathrm{Pek}$ and
   \begin{align}
   \label{Eq:Decomposition_in_Lemma_of_interaction}
        v=v^\Lambda+\frac{1}{i}\nabla\cdot w^\Lambda.
   \end{align}
   Furthermore, there exists a constant $C>0$ such that for all $\Lambda>0$ and $R>0$
   \begin{align}
   \label{Eq:In_Lemma_v_w_Lambda}
      & \left\|v^\Lambda\right\|\leq C\Lambda^{\frac{1}{2}}, \ \ \ \ \left\|w^\Lambda\right\|\leq C\Lambda^{-\frac{1}{2}}, \\
      \label{Eq:High_Momentum_varphi}
     & \ \   \left\|\chi\!\left(|\nabla|>\Lambda\right)\varphi^\mathrm{Pek}\right\|\leq C\Lambda^{-\frac{1}{2}},\\
     \label{Eq:High_position_varphi}
      &    \ \  \left\|\chi\! \left(|y|>R\right)\varphi^\mathrm{Pek}\right\|\leq C R^{-\frac{1}{2}},
   \end{align}
  and $\frac{\nabla \varphi^\mathrm{Pek}}{\varphi^\mathrm{Pek}}\in C_b^2(\mathbb R^3,\mathbb R^3)$.
\end{lem}

\subsection{Conventions and Definitions}
\label{Subsec:Conventions_and_Definitions}
In this Subsection we want to state some conventions and definitions that will be used repeatedly throughout this article. \\

\textbf{Convention.} Most importantly, $\alpha$ will always refer to a parameter $\alpha\geq 1$, which appears for example in the definition of the operator $L(f)$ in Eq.~(\ref{Eq:Definition_L_Operator}), and therefore in the definition of the creation and annihilation operators $a(f)$ and $a^*(f)$, and in the definition of the Fröhlich Hamiltonian $\mathbb H$. Furthermore, the empirical measures $\rho_{Y_n}$ defined in Eq.~(\ref{Eq:Empirical_Measure}) depends on $\alpha$, and therefore all functions of $\rho_{Y_n}$ do so as well. For the sake of readability, we will not always indicate the dependence on $\alpha$. Furthermore, if not indicated otherwise, all estimates are supposed to hold uniformly in $\alpha$ as $\alpha\rightarrow \infty$, e.g. we write 
\begin{align*}
  X\lesssim Y,  
\end{align*}
for $\alpha$-dependent objects $X$ and $Y$, in case there exists an $\alpha_0\geq 1$ and a $C>0$, such that
\begin{align*}
    X\leq CY,
\end{align*}
 for all $\alpha\geq \alpha_0$. Similarly, $p$ will always refer to an elements $p\in \mathbb R^3$ and, if not indicated otherwise, estimates are supposed to hold uniform in $p$. 
 
 The second convention we will make use of, is that multiplication operators by functions $Y_n\mapsto F(Y_n)$ will be denoted simply by $F(Y_n)$, and in case $F$ is in a natural way a function $f$ of the empirical measure we write $f(\rho_{Y_n})$. Occasionally we will further write $F(y_n)$ for multiplication operators depending only on the component $y_n$ in $Y_n=(y_1,\dots ,y_n)$. 
 
 Finally, we call an element $\Psi\in \mathcal{H}$ of a Hilbert space $\mathcal{H}$ a state, in case $\|\Psi\|=1$. \\

\textbf{Definition.} Besides the notation $\chi(a\leq f\leq b)$ introduced in Eq.~(\ref{Eq:General_Cut-Off_Functions}), we will denote with $g$ a (fixed) rotationally symmetric mollifier on $\mathbb R^3$, i.e. $g:\mathbb R^3\longrightarrow [0,1]$ will denote a smooth function with compact support, $\int_{\mathbb R^3}g(x)\mathrm{d}x=1$ and $g(x)=g(y)$ in case $|x|=|y|$. Furthermore, it will be useful to have mollifiers at hand that are convolutions them self, and therefore we introduce the function
\begin{align*}
    \chi:=g*g,
\end{align*}
which is clearly a rotationally symmetric mollifier again. It will also be convenient to have for $T>0$ the re-scaled (rotationally symmetric) molllifiers $g_T$ at hand
\begin{align*}
    g_T(y):=T^3 g(Ty).
\end{align*}
We will furthermore use a third family of smooth cut-off functions $\tau_{\eta}$ defined as
\begin{align*}
    \tau_{\eta}(y):=\eta^{-\frac{3}{2}}\tau\! \left(\eta^{-1}y\right),
\end{align*}
where $\tau:\mathbb R^3\longrightarrow [0,1]$ is smooth with support in $\left\{y:|y|\leq \frac{1}{3}\right\}$ and satisfies $\int\tau(y)^2\mathrm{d}y=1$.

In addition let us fix some constants, used throughout this paper. In the following let $\delta_*$ be such that $0<\delta_*<\frac{\|\varphi^\mathrm{Pek}\|^4}{2}$, and pick $\kappa,\sigma$ and $q$, such that $0<\kappa<\sigma$ and
\begin{align}
\label{Eq:rho_not_zero}
   & \ \ \ \    \ \     \sigma+\kappa <\|\varphi^\mathrm{Pek}\|^2, \\
    \nonumber
  &  \delta_*+\sigma+\kappa<\frac{(\|\varphi^\mathrm{Pek}\|^2-\sigma-\kappa)^2}{2}, \\
    \nonumber
 &  \  0<q<\frac{1}{2}-\frac{ \delta_*+\sigma+\kappa}{(\|\varphi^\mathrm{Pek}\|^2-\sigma-\kappa)^2}.
\end{align}
Moreover, since $\varphi^\mathrm{Pek}\in L^2\! \left(\mathbb R^3\right)$, $\varphi^\mathrm{Pek}>0$ and $\delta_*<\|\varphi^\mathrm{Pek}\|^4$, we find a $R_*>0$ such that
\begin{align*}
    \underset{|x-y|>R_*}{\int \int}|\varphi^\mathrm{Pek}(x)|^2|\varphi^\mathrm{Pek}(y)|^2\mathrm{d}x\mathrm{d}y= \delta_*.
\end{align*}
With $\delta_*$ and $R_*$ at hand, let us introduce for $\lambda>0$
\begin{align}
\label{Eq:Omega_def}
  \Omega_\lambda:=\bigcup_{n\in \mathbb N}\Omega_\lambda^{(n)}  
\end{align}
as the set of all $Y_n\in \underset{n\in \mathbb N}{\bigcup}\mathbb R^{3n}$, such that
\begin{align*}
 & \ \  \|\varphi^\mathrm{Pek}\|^2-\lambda\leq \int \mathrm{d}\rho_{Y_n}\leq  \|\varphi^\mathrm{Pek}\|^2+\lambda,\\
  &    \delta_*-\lambda\leq \underset{|x-y|>R_*}{\int \int}\mathrm{d}\rho_{Y_n}(x)\mathrm{d}\rho_{Y_n}(y)\leq \delta_*+\lambda.
\end{align*}
It will be the main result of Section \ref{Sec:Asymptotic_concentration_of_Probability} that the ground state concentrates on sets of the form $\Omega_\lambda$ for $\lambda>0$, as $\alpha\rightarrow \infty$. Furthermore, we demonstrate in Section \ref{Sec:Analysis_of_the_generalized_Median} that $\sigma$ and $\kappa$ are chosen in such a way, that the regularized median $m_q$ introduced in Eq.~(\ref{Eq:Introduction_regularized_median}) satisfies uniform continuity and differentiability properties on $\Omega_{\sigma+\kappa}$, see also the assumptions of \cite[Lemma 3.10]{BS1}.

\section{Construction of a Trial State}
\label{Sec:Construction_of_a_Trial_State}
It is the content of this Section, to construct a trial state $\Psi_{\alpha,p}\in \mathcal{F}$ satisfying
\begin{align*}
   \left\langle \Psi_{\alpha,p},\mathbb H_p \Psi_{\alpha,p}\right\rangle_\mathcal{F}\leq E_\alpha+\frac{|p|^2}{2\alpha^4 m_{\mathrm{LP}}}+C\alpha^{-(4+\epsilon)}|p|^2,
\end{align*}
as is shown in the subsequent Section \ref{Sec:Proof_of_Theorem}. Following the strategy proposed in Section \ref{Sec:Introduction_and_Main_Result}, right after Theorem \ref{Th:Main}, we will define $\Psi_{\alpha,p}$ as the ground state $\Psi_\alpha$ of the fiber Hamiltonian $\mathbb H_0$ boosted by the unitary group $e^{ip\cdot \mathbb B}$. Applying the (pseudo) boost $e^{ip\cdot \mathbb B}$ is a necessary procedure, as we want to have a cancellation of the $p$-term appearing in the fiber Hamiltonian
\begin{align*}
   \mathbb H_p=\left( p - \mathcal P\right)^2+\mathcal{N}-a^*(v)-a(v).
\end{align*}
Choosing the right (pseudo) boost, i.e. choosing the right self-adjoint operators $\mathbb B=(\mathbb B_1,\mathbb B_2,\mathbb B_3)$ satisfying $e^{-ip\cdot \mathbb B}\mathcal Pe^{ip\cdot \mathbb B}\approx \mathcal P+p$, is a delicate task, which we are going to elaborate on in the following Subsection \ref{Subsec:The_correct_Way_to_Boost}. The state $\Psi_{\alpha,p}$ is then introduced in Eq.~(\ref{Eq:Definition_proper_boosted_State}).

\subsection{The correct Way to Boost}
\label{Subsec:The_correct_Way_to_Boost}
On a semi-classical level, the right boost is generated by the function $B$, in the sense of Eq.~(\ref{Eq:Introduction_Hamiltonian_Diff._Eq.}), defined on the phase space $L^2\! \left(\mathbb R^3\right)$ as
\begin{align}
\label{Eq:Introduction_H_B_with_modifier_COPY}
     B(\varphi):  =m_q(\chi* \rho_\varphi)+\int_{\mathbb R^3} f\Big(y-m_q(\chi* \rho_\varphi)\Big)\mathrm{d}\rho_\varphi(y) 
\end{align}
for a suitable function $f:\mathbb R^3\longrightarrow \mathbb R^3$, see Eq.~(\ref{Eq:Introduction_H_B_with_modifier}) and the discussion below, which we assume to be in $C^2_b(\mathbb R^3,\mathbb R^3)$. It will be the goal of this Subsection to find a quantum counterpart ${\mathbb  B}=({\mathbb  B}_1,{\mathbb  B}_2,{\mathbb  B}_3)$ to $B$ in the form of self-adjoint commuting operators ${\mathbb  B}_j$ defined on $\mathcal{F}$. 

In advance we are giving a proper definition of the regularized median $m_q$ appearing in the definition of $B$ in Eq.~(\ref{Eq:Introduction_H_B_with_modifier_COPY}), see also \cite{BS0,BS1,BS2} for previous appearances of the statistical quantity $m_q$. Following \cite{BS1}, we define for $0\leq \lambda \leq 1$ the $\lambda$-quantile of a (finite, Borel) measure $\nu$ on $\mathbb R$ as
\begin{align}
\label{Eq:Def_Quantile}
    x^\lambda(\nu):=\sup \! \left\{t:\int_{-\infty}^t\mathrm{d}\nu\leq \lambda \int \mathrm{d}\nu\right\}.
\end{align}
Furthermore, let $\rho_j$ denote for $j\in \{1,2,3\}$ the $j$-th marginal distribution of a (finite, Borel) measure $\rho$ on $\mathbb R^3$, defined by
\begin{align}
\label{Eq:Def_Marginal_Measure}
    \int_{\mathbb R} h\mathrm{d}\rho_j:=\int_{\mathbb R^3} h(y_j)\mathrm{d}\rho(y).
\end{align}
With these notations at hand, and $K_q(\nu):=\left[x^{\frac{1}{2}-q}(\nu),x^{\frac{1}{2}+q}(\nu)\right]\subseteq \mathbb R$ for $q>0$, we define the regularized median $m_q(\rho)=(m_q(\rho)_1,m_q(\rho)_2,m_q(\rho)_3)$ in coordinates as
\begin{align*}
   m_q(\rho)_j:=\left(\int \mathds{1}_{K_q(\rho_j)}\mathrm{d}\rho_j\right)^{-1} \! \int \mathds{1}_{K_q(\rho_j)}(t)\, t\mathrm{d}\rho_j(t),
\end{align*}
with the convention that $m_q(0):=0$. Finally, recall for $\varphi\in L^2\! \left(\mathbb R^3\right)$ the definition of the measure $\rho_\varphi$ from Section \ref{Sec:Introduction_and_Main_Result}
\begin{align}
\label{Eq:Definition_rho_of_varphi}
    \mathrm{d}\rho_\varphi(y):=|\varphi(y)|^2\mathrm{d}y.
\end{align}
In order to find a quantum counterpart to the semi-classical object $B$ from Eq.~(\ref{Eq:Introduction_H_B_with_modifier_COPY}), observe the formal correspondence between the empirical measure $\rho_{Y_n}$ defined in Eq.~(\ref{Eq:Empirical_Measure}) and the measure $\rho_\varphi$ defined in Eq.~(\ref{Eq:Definition_rho_of_varphi}), see for example Eq.~(\ref{Eq:Demonstrating_rho_correspondence}). Therefore, a natural candidate for the boost $\mathbb B$ would be given by the multiplication operator $\mathcal{B}(\rho_{Y_n})$ with
\begin{align}
\label{Eq:Boost_Operator}
  \mathcal{B}(\rho):=m_q(\chi* \rho)+\int_{\mathbb R^3} f\Big(y-m_q(\chi* \rho)\Big)\mathrm{d}\rho(y).
\end{align}
Since we will require a multiplication operator that satisfies certain continuity and differentiability properties, we will use a slightly modified version $G$ instead of the function $\mathcal{B}$ from Eq.~(\ref{Eq:Boost_Operator}), defined as
\begin{align*}
    G(\rho):=F(\rho)\mathcal{B}(\rho),
\end{align*}
where the function 
\begin{align}
    \label{Eq:Def_F_cut_off_functional}
    F(\rho): & =\chi_{\kappa^2}\! \left(\left(\int \mathrm{d}\rho-\|\varphi^\mathrm{Pek}\|^2\right)^2 + \left(\underset{|x-y|>R_*}{\int \int}\mathrm{d}\rho\mathrm{d}\rho-\delta_*\right)^2 \leq \sigma^2\right)
\end{align}
makes sure that $G$ is only supported on regular enough measures $\rho$. The operator $\mathbb B$, which we will use as the generator of a unitary group, is then defined as
\begin{align}
\label{Eq:Definition_Almost_Boost_Operator}
    \mathbb B:=G(\rho_{Y_n})=F(\rho_{Y_n})\mathcal{B}(\rho_{Y_n}).
\end{align}
 The constants $\kappa, R_*,\delta_*$ and $\sigma$ appearing in the definition of $F$ as well as the constant $q$ and the mollifier $\chi$ appearing in the definition of $\mathcal{B}$ are chosen according to Subsection \ref{Subsec:Conventions_and_Definitions}.

Before we define the trial state $\Psi_{\alpha,p}$, let us make some observations about the spectrum of the fiber Hamiltonian $\mathbb H_0$ and analyse the operators $\mathbb B=(\mathbb B_1,\mathbb B_2,\mathbb B_3)$. First of all, it has been shown in \cite{DS}, respectively \cite{Miy}, that $E_\alpha$ is an isolated and non-degenerate point in the spectrum of $\mathbb H_0$. As an immediate consequence we obtain the following Lemma \ref{Lem:Unique_GS}, which establishes the existence of a unique reflection symmetric ground state $\Psi_\alpha$ of $\mathbb H_0$.
\begin{lem}
\label{Lem:Unique_GS}
    Let $\mathcal{S}:\mathcal{F}\longrightarrow \mathcal{F}$ be defined as $(\mathcal{S}\Psi)(Y_n):=\Psi(-Y_n)$ for all $Y_n\in \underset{n\in \mathbb N}{\bigcup} \mathbb{R}^{3 n}$. Then the operator $\mathbb H_0$ has, up to a phase, a unique ground state $\Psi_\alpha\in \mathcal{F}$ satisfying
    \begin{align*}
       \mathcal{S}\Psi_\alpha = \Psi_\alpha.
    \end{align*}
    Furthermore, the phase can be chosen such that $\Psi_\alpha(Y_n)\geq 0$ for all $Y_n$.
\end{lem}
\begin{proof}
    By \cite{DS}, respectively \cite{Miy}, we know that $E_\alpha$ is an isolated and non-degenerate point in the spectrum of $\mathbb H_0$, i.e. there exists, up to a phase, a unique ground state $\Psi_\alpha$. Since the unitary operator $\mathcal{S}$ commutes with $\mathbb H_0$, i.e. $\mathcal{S}^* \mathbb H_0 \mathcal{S}=\mathbb H_0$, we obtain 
    \begin{align*}
        \left\langle \mathcal{S}\Psi_\alpha, \mathbb H_0 \mathcal{S} \Psi_\alpha\right\rangle=\left\langle \Psi_\alpha, \mathbb H_0  \Psi_\alpha\right\rangle=E_\alpha.
    \end{align*}
    Consequently, $\mathcal{S}\Psi_\alpha$ is a ground state as well and by the uniqueness of ground states there exists a $\theta\in [0,2\pi)$ with $\mathcal{S}\Psi_\alpha=e^{i\theta}\Psi_\alpha$. A simple computation exhibits that $|\Psi_\alpha|\in \mathcal{F}$ defined as $|\Psi_\alpha|(Y_n):=\left|\Psi_\alpha(Y_n)\right|$ satisfies
    \begin{align*}
        \left\langle |\Psi_\alpha|,\mathbb H \, |\Psi_\alpha|\right\rangle\leq \left\langle \Psi_\alpha,\mathbb H \Psi_\alpha\right\rangle=E_\alpha,
    \end{align*}
    i.e. $|\Psi_\alpha|$ is a ground state again and by the uniqueness of ground states we can therefore assume w.l.o.g. that $\Psi_\alpha\geq 0$. Hence, $\theta=0$.
\end{proof}

The support of the function $G$ in $\mathbb B=G(\rho_{Y_n})$ is clearly contained in the support of $F$
\begin{align*}
    \mathrm{supp}\! \left(G\right)\subseteq \mathrm{supp}\! \left(F\right),
\end{align*}
and using the set $\Omega_\lambda$ defined in Eq.~(\ref{Eq:Omega_def}) we are going to localize the support of $F$ in the following Lemma \ref{Lem:Precise_Support_F}.
\begin{lem}
\label{Lem:Precise_Support_F}
    Let $F$ be as in Eq.~(\ref{Eq:Def_F_cut_off_functional}) and $\Omega_\lambda$ as in Eq.~(\ref{Eq:Omega_def}). Then,
    \begin{align}
    \label{Eq:Support_Chain}
         \Omega_{\frac{\sigma-\kappa}{\sqrt{2}}}\subseteq \left\{Y_n:F(\rho_{Y_n})=1\right\}\subseteq \mathrm{supp}(F)\subseteq \Omega_{\sigma+\kappa}.
    \end{align}
    Moreover, $\Omega_{\sigma+\kappa}\subseteq \left\{Y_n:\rho_{Y_n}\neq 0\right\}$.
\end{lem}
\begin{proof}
    The inclusion $ \left\{Y_n:F(\rho_{Y_n})=1\right\}\subseteq \mathrm{supp}(F)=\left\{Y_n:F(\rho_{Y_n})\neq 0\right\}$ is trivial. In order to see that $\mathrm{supp}(F)$ is contained in $\Omega_{\sigma+\kappa}$, note that by the definition of the cut-off function $\chi_{\kappa^2}(\cdot \leq \sigma^2)$ in Eq.~(\ref{Eq:General_Cut-Off_Functions}), $F(\rho_{Y_n})\neq 0$ implies
    \begin{align*}
       \left(\int \mathrm{d}\rho_{Y_n}-\|\varphi^\mathrm{Pek}\|^2\right)^2 + \left(\underset{|x-y|>R_*}{\int \int}\mathrm{d}\rho_{Y_n}\mathrm{d}\rho_{Y_n}-\delta_*\right)^2 \leq \sigma^2+\kappa^2 \leq (\sigma+\kappa)^2,
    \end{align*}
    and therefore $Y_n\in \Omega_{\sigma+\kappa}$, since we have both 
    \begin{align*}
        & \ \ \ \left|\int \mathrm{d}\rho_{Y_n}-\|\varphi^\mathrm{Pek}\|^2\right|\leq \sigma+\kappa,\\
        & \left| \underset{|x-y|>R_*}{\int \int}\mathrm{d}\rho_{Y_n}\mathrm{d}\rho_{Y_n}-\delta_*\right|\leq \sigma+\kappa.
    \end{align*}
    On the other hand, in case $Y_n\in \Omega_{\frac{\sigma-\kappa}{\sqrt{2}}}$, we clearly have
    \begin{align*}
        \left(\int \mathrm{d}\rho_{Y_n}-\|\varphi^\mathrm{Pek}\|^2\right)^2 + \left(\underset{|x-y|>R_*}{\int \int}\mathrm{d}\rho_{Y_n}\mathrm{d}\rho_{Y_n}-\delta_*\right)^2\leq (\sigma-\kappa)^2\leq \sigma^2-\kappa^2,
    \end{align*}
    where we have used $\kappa\leq 1$. Since $\chi_{\kappa^2}(x\leq \sigma^2)=1$ for $x\leq \sigma^2-\kappa^2$, this concludes the proof of Eq.~(\ref{Eq:Support_Chain}). Finally, recall that $\sigma+\kappa<\|\varphi^\mathrm{Pek}\|^2$, and therefore we have for $Y_n\in \Omega_{\sigma+\kappa}$
    \begin{align*}
        \int \mathrm{d}\rho_{Y_n}\geq \|\varphi^\mathrm{Pek}\|^2-\left|\int \mathrm{d}\rho_{Y_n}-\|\varphi^\mathrm{Pek}\|^2\right|\geq \|\varphi^\mathrm{Pek}\|^2-(\sigma+\kappa)>0.
    \end{align*}
\end{proof}
In the following we will consider $\Omega_{\sigma+\kappa}$ as a set of regular configurations, e.g. we have $\int \mathrm{d}\rho_{Y_n}> 0$ for all $Y_n\in \mathrm{supp}(F)$ by Lemma \ref{Lem:Precise_Support_F}. While it is clear that the components of $\mathbb B=(\mathbb B_1,\mathbb B_2,\mathbb B_3)$ commute, the following Lemma \ref{Lem:Pseudo_Boost_B_FB} demonstrates that both the multiplication operator $\mathcal{B}(\rho_{Y_n})$ as well as the operator $\mathbb B$ are pseudo boosts, see Eq.~(\ref{Eq:Introduction_Boost_Relations_For_P}), where we anticipate that the ground state $\Psi_\alpha$ is mostly supported on configurations $Y_n\in \Omega_{\frac{\sigma-\kappa}{\sqrt{2}}}$, see Section \ref{Sec:Asymptotic_concentration_of_Probability}, and hence $\Psi_\alpha$ is mostly supported on configurations satisfying both $\int \mathrm{d}\rho_{Y_n}>0$ and $F(\rho_{Y_n})=1$.
\begin{lem}
\label{Lem:Pseudo_Boost_B_FB}
    Let $\mathcal{B}$ be as in Eq.~(\ref{Eq:Boost_Operator}) and $\mathbb B$ as in Eq.~(\ref{Eq:Definition_Almost_Boost_Operator}). Then,
\begin{align}
\label{Eq:Pseudo_Boost_B}
    e^{-ip\cdot \mathcal{B}(\rho_{Y_n})} \mathcal P  e^{ip\cdot \mathcal{B}(\rho_{Y_n})} & =\mathcal P+\chi\! \left(\int \mathrm{d}\rho_{Y_n}>0\right)\! p,\\
    \label{Eq:Pseudo_Boost_FB}
        e^{-ip\cdot \mathbb B} \mathcal P  e^{ip\cdot \mathbb B} & =\mathcal P+F(\rho_{Y_n})p.
\end{align}
\end{lem}
\begin{proof}
    Let us define for $z\in \mathbb R^3$ the shift of a measure $\rho$ as
    \begin{align*}
        \int h\mathrm{d}\rho_z :=\int h(y+z)\mathrm{d}\rho(y),
    \end{align*}
    and compute for $\Psi\in \mathcal{F}$ using the definition of $\mathcal P$, see Eq.~(\ref{Eq:P_f_Operator}),
    \begin{align*}
      & \ \ \ \ \ \ \left( e^{-ip\cdot \mathbb B} \mathcal P  e^{ip\cdot \mathbb B}\Psi-\mathcal P\Psi\right)(Y_n)= \frac{1}{i}\sum_{j=1}^n e^{-ip\cdot G(\rho_{Y_n})}\nabla_{y_j}\left( e^{ip\cdot G(\rho_{Y_n})}\right)\Psi(Y_n)\\
        & = \frac{1}{i}e^{-ip\cdot G(\rho_{Y_n})}\nabla_{z}\Big|_{z=0} \! \left( e^{ip\cdot G \left((\rho_{Y_n})_z\right)}\right)\! \Psi(Y_n) \! = \! \frac{1}{i}e^{-ip\cdot G(\rho_{Y_n})}\nabla_{z}\Big|_{z=0} \! \left( e^{ip\cdot \mathcal{B} \left((\rho_{Y_n})_z\right)F \left((\rho_{Y_n})_z\right)}\right)\! \Psi(Y_n)\\
        & \ \ \ \ \ \ =\frac{1}{i}e^{-ip\cdot G(\rho_{Y_n})}\nabla_{z}\Big|_{z=0}\left( e^{ip\cdot \left(\mathcal{B}(\rho_{Y_n})+z\right)F(\rho_{Y_n})}\right)\Psi(Y_n)=F(\rho_{Y_n})p \, \Psi(Y_n).  
    \end{align*}
where we have used $\mathcal{B}((\rho_{Y_n})_z)=\mathcal{B}(\rho_{Y_n})+z$ in case $\int \mathrm{d}\rho_{Y_n}>0$ and $F(\rho_{Y_n})=0$ in case $\rho_{Y_n}=0$, see Lemma \ref{Lem:Precise_Support_F}. Eq.~(\ref{Eq:Pseudo_Boost_B}) can be verified analogously. 
\end{proof}

Having the (unique, symmetric) ground state $\Psi_\alpha\in \mathcal{F}$ of $\mathbb H_0$ at hand, see Lemma \ref{Lem:Unique_GS}, and the (multiplication) operator $\mathbb B$ defined in  Eq.~(\ref{Eq:Definition_Almost_Boost_Operator}), we introduce the trial state $\Psi_{\alpha,p}\in \mathcal{F}$ as
\begin{align}
\label{Eq:Definition_proper_boosted_State}
    \Psi_{\alpha,p}:=e^{ip\cdot \mathbb B}\Psi_{\alpha}=\left(Y_n\mapsto e^{ip\cdot G(\rho_{Y_n})}\Psi_{\alpha}(Y_n)\right)=\left(Y_n\mapsto e^{i p\cdot F(\rho_{Y_n})\mathcal{B}(\rho_{Y_n})}\Psi_{\alpha}(Y_n)\right),
\end{align}
which depends on the function $f:\mathbb R^3\longrightarrow \mathbb R^3$ in Eq.~(\ref{Eq:Boost_Operator}). Note that we optimize the choice of $f$ in Section \ref{Sec:Proof_of_Theorem}, however the constants $\kappa, R_*, \delta_*,\sigma$ and $q$, and the mollifier $\chi$, appearing in the definition of $F$ and $\mathcal{B}$, and therefore in the definition of $\mathbb B$, should be considered as being fixed according to Subsection \ref{Subsec:Conventions_and_Definitions}. It will be the main challenge of this manuscript, to show that the quantum energy 
   \begin{align*}
     \left\langle \Psi_{\alpha,p},\mathbb H_p  \Psi_{\alpha,p}  \right\rangle_\mathcal{F} -E_\alpha\leq \frac{|p|^2}{2\alpha^4 m_{\mathrm{LP}}}+o\! \left(\alpha^{-4}|p|^2\right).
   \end{align*}
   is to leading order bounded from above by the semi-classical prediction $\frac{|p|^2}{2\alpha^4 m_{\mathrm{LP}}}$ for an optimal choice of $f$.

\subsection{Isolating the essential Contribution in $\left\langle \Psi_{\alpha,p},\mathbb H_p  \Psi_{\alpha,p}  \right\rangle_\mathcal{F}$}
\label{Sec:Isolating_the_essential_Contribution}
In this Subsection we want to compute the energy $\left\langle \Psi_{\alpha,p},\mathbb H_p  \Psi_{\alpha,p}  \right\rangle_\mathcal{F}$, and isolate the leading order term from contributions that are small compared to $\alpha^{-4}|p|^2$. In a first step, we provide a rather explicit expression for the energy in the following Lemma \ref{Lem:Rather_Explicit_Expression}, using the function $\xi_p:\underset{n\in \mathbb N}{\bigcup}\mathbb R^{3n}\longrightarrow \mathbb C$ defined as
\begin{align}
\label{Eq:Definition_xi_Computation}
    \xi_p(Y_n):=e^{ip\cdot \left[G(\rho_{Y_{n-1}})-G(\rho_{Y_{n}})\right]}-1-ip\cdot \left[G(\rho_{Y_{n-1}})-G(\rho_{Y_{n}})\right].
\end{align}

\begin{lem}
\label{Lem:Rather_Explicit_Expression}
    Let $f\in C^2_b(\mathbb R^3,\mathbb R^3)$ be reflection anti-symmetric, i.e. $f(-y)=-f(y)$, and let $ \Psi_{\alpha,p}$ be the state defined in Eq.~(\ref{Eq:Definition_proper_boosted_State}). Then we have the identity
    \begin{align}
    \label{Eq:Computation_Trial_Energy_using_Symmetry}
      &  \left\langle \Psi_{\alpha,p},\mathbb H_p  \Psi_{\alpha,p}  \right\rangle_\mathcal{F}=E_\alpha -2\mathfrak{Re}\! \left[\big\langle \Psi_{\alpha},\xi_p(Y_n)L(v)\Psi_{\alpha}\big\rangle_{L^2\! \left(\underset{n\in \mathbb N}{\bigcup}\mathbb R^{3n}\right)}\right]\\
     \nonumber
        & \ \ \ \  \ \ \ \  \ \ \ \  \ \ \ \ +|p|^2 \left\langle \Psi_{\alpha}, \left[1-F(\rho_{Y_n})\right]^2 \Psi_{\alpha}\right\rangle_\mathcal{F}.
    \end{align}
    Furthermore, there exists a constant $C>0$ such that
    \begin{align}
    \label{Eq:NEW_T_1_Smallness}
       \mathcal{E}_0:= \left\langle \Psi_{\alpha}, \left[1-F(\rho_{Y_n})\right]^2 \Psi_{\alpha}\right\rangle_\mathcal{F}\leq C\alpha^{-5}.
    \end{align}
\end{lem}
\begin{proof}
    By Eq.~(\ref{Eq:Pseudo_Boost_FB}), we can write
    \begin{align*}
     & \ \ \ \   \ \ \ \  \left\langle \Psi_{\alpha,p},\left(p \! - \! \mathcal P\right)^2  \Psi_{\alpha,p}  \right\rangle_\mathcal{F} \! = \!  \left\langle \Psi_{\alpha},e^{-ip\cdot \mathbb B}\left(p \! - \! \mathcal P\right)^2 e^{ip\cdot \mathbb B} \Psi_{\alpha}  \right\rangle_\mathcal{F} \\
     &  \ \ \ \  =  \left\langle \Psi_{\alpha},\left(p \! - \! e^{-ip\cdot \mathbb B}\mathcal P e^{ip\cdot \mathbb B} \right)^2 \Psi_{\alpha}  \right\rangle_\mathcal{F} =\left\langle \Psi_{\alpha},\Big(\big[1-F(\rho_{Y_n})\big]p-\mathcal{P} \right)^2 \Psi_{\alpha}  \Big\rangle_\mathcal{F}\\
     & =\left\langle \Psi_{\alpha},\mathcal{P}^2 \Psi_{\alpha}  \right\rangle_\mathcal{F}+|p|^2\left\langle \Psi_{\alpha}, \left[1-F(\rho_{Y_n})\right]^2 \Psi_{\alpha}\right\rangle_\mathcal{F}+2p \mathfrak{Re}\! \left[\left\langle \Psi_{\alpha}, \left[1-F(\rho_{Y_n})\right]\mathcal{P} \Psi_{\alpha}  \right\rangle_\mathcal{F}\right].
    \end{align*}
    Note that $\left\langle \Psi_{\alpha}, \left[1-F(\rho_{Y_n})\right]\mathcal{P} \Psi_{\alpha}  \right\rangle_\mathcal{F}=0$, since $\Psi_{\alpha}$ is invariant under the reflection $\mathcal{S}$, see Lemma \ref{Lem:Unique_GS}, and $ \left[1-F(\rho_{Y_n})\right]\mathcal{P}$ is relection anti-symmetric, i.e.
    \begin{align*}
       \mathcal{S}^* \left[1-F(\rho_{Y_n})\right]\mathcal{P}\mathcal{S}=\left[1-F(\rho_{Y_n})\right]\mathcal{S}^* \mathcal{P}\mathcal{S}=-\left[1-F(\rho_{Y_n})\right]\mathcal{P}.
    \end{align*}
    Consequently,
    \begin{align*}
        \left\langle \Psi_{\alpha,p},\left(p \! - \! \mathcal P\right)^2  \Psi_{\alpha,p}  \right\rangle_\mathcal{F}=\left\langle \Psi_{\alpha},\mathcal{P}^2 \Psi_{\alpha}  \right\rangle_\mathcal{F}+|p|^2\left\langle \Psi_{\alpha}, \left[1-F(\rho_{Y_n})\right]^2 \Psi_{\alpha}\right\rangle_\mathcal{F}.
    \end{align*}
    We observe that both $e^{ip\cdot \mathbb B}$ and $\mathcal{N}$ are multiplication operators and therefore
    \begin{align*}
         \left\langle \Psi_{\alpha,p},\mathcal{N}  \Psi_{\alpha,p}  \right\rangle_\mathcal{F}=\left\langle \Psi_{\alpha},e^{-ip\cdot \mathbb B }\mathcal{N}  e^{ip\cdot \mathbb B }\Psi_{\alpha}  \right\rangle_\mathcal{F}=\left\langle \Psi_{\alpha},\mathcal{N}  e^{-ip\cdot \mathbb B }e^{ip\cdot \mathbb B }\Psi_{\alpha}  \right\rangle_\mathcal{F}=\left\langle \Psi_{\alpha},\mathcal{N}  \Psi_{\alpha}  \right\rangle_\mathcal{F}.
    \end{align*}
    Finally, recall the definition of $L(v)$ and $a^*(v)$ from Subsection \ref{Subsec:Basic_Notation} and compute
    \begin{align}
   \nonumber
\left\langle \Psi_{\alpha,p},a^*(v) \Psi_{\alpha,p}  \right\rangle_\mathcal{F}     \!   &  = \! \left\langle \Psi_{\alpha,p},L(v) \Psi_{\alpha,p}  \right\rangle_{L^2\! \left(\underset{n\in \mathbb N}{\bigcup}\mathbb R^{3n}\right)} \! \!  = \! \sum_{n=0}^\infty \frac{\sqrt{\alpha}}{n}\int_{\mathbb R^{3n}} \!  \! \overline{\Psi_{\alpha,p}(Y_n)}\Psi_{\alpha,p}(Y_{n-1})v(y_n)\mathrm{d}Y_n\\
\label{Eq:Computing_a_Star_Term}
      &  =\sum_{n=0}^\infty \frac{\sqrt{\alpha}}{n}\int_{\mathbb R^{3n}} \!  \! e^{i p\cdot \left[G(\rho_{Y_{n-1}})- G(\rho_{Y_n})\right]}\overline{\Psi_{\alpha}(Y_n)}\Psi_{\alpha}(Y_{n-1})v(y_n)\mathrm{d}Y_n.
    \end{align}
    Due to our assumption that $f$ is reflection anti-symmetric, we obtain that $G(\rho_{-Y_n})=-G(\rho_{Y_n})$. Using again that $\Psi_\alpha$ is reflection invariant, therefore yields
    \begin{align}
    \label{Eq:Reflection_anti-symmetry-_f}
        \sum_{n=0}^\infty \frac{\sqrt{\alpha}}{n}\int_{\mathbb R^{3n}} \!  \! \left[G(\rho_{Y_{n-1}})- G(\rho_{Y_n})\right]\overline{\Psi_{\alpha}(Y_n)}\Psi_{\alpha}(Y_{n-1})v(y_n)\mathrm{d}Y_n=0.
    \end{align}
 Combining Eq.~(\ref{Eq:Computing_a_Star_Term}), Eq.~(\ref{Eq:Reflection_anti-symmetry-_f}) and the definition of $\xi_p$ in Eq.~(\ref{Eq:Definition_xi_Computation}), we observe that
 \begin{align*}
    \left\langle \Psi_{\alpha,p},a^*(v) \Psi_{\alpha,p}  \right\rangle_\mathcal{F} & =\left\langle \Psi_{\alpha},a^*(v) \Psi_{\alpha}  \right\rangle_\mathcal{F} +\sum_{n=0}^\infty \frac{\sqrt{\alpha}}{n}\int_{\mathbb R^{3n}} \xi_p(Y_n)\overline{\Psi_{\alpha}(Y_n)}\Psi_{\alpha}(Y_{n-1})v(y_n)\mathrm{d}Y_n\\
    &= \left\langle \Psi_{\alpha},a^*(v) \Psi_{\alpha}  \right\rangle_\mathcal{F} +\Big\langle \Psi_{\alpha},\xi_p(Y_n)L(v) \Psi_{\alpha}  \Big\rangle_{L^2\! \left(\underset{n\in \mathbb N}{\bigcup}\mathbb R^{3n}\right)} .
 \end{align*}
 This concludes the proof of Eq.~(\ref{Eq:Computation_Trial_Energy_using_Symmetry}), since
 \begin{align*}
    &  \left\langle \Psi_{\alpha,p},\mathbb H_p  \Psi_{\alpha,p}  \right\rangle_\mathcal{F}=\left\langle \Psi_{\alpha,p},\left(p \! - \! \mathcal P\right)^2  \Psi_{\alpha,p}  \right\rangle_\mathcal{F}+ \left\langle \Psi_{\alpha,p},\mathcal{N}  \Psi_{\alpha,p}  \right\rangle_\mathcal{F}-2\mathfrak{Re}\! \left[ \left\langle \Psi_{\alpha,p},a^*(v) \Psi_{\alpha,p}  \right\rangle_\mathcal{F}\right]\\
      & \ \  =\left\langle \Psi_{\alpha}, \mathcal P  \Psi_{\alpha}  \right\rangle_\mathcal{F}+ \left\langle \Psi_{\alpha},\mathcal{N}  \Psi_{\alpha}  \right\rangle_\mathcal{F}-2\mathfrak{Re}\! \left[ \left\langle \Psi_{\alpha},a^*(v) \Psi_{\alpha}  \right\rangle_\mathcal{F}\right]+ |p|^2\left\langle \Psi_{\alpha}, \left[1-F(\rho_{Y_n})\right]^2 \Psi_{\alpha}\right\rangle_\mathcal{F}\\
      &  \ \  \ \  \ \  \ \  \ \ -2\mathfrak{Re}\! \left[\Big\langle \Psi_{\alpha},\xi_p(Y_n)L(v) \Psi_{\alpha}  \Big\rangle_{L^2\! \left(\underset{n\in \mathbb N}{\bigcup}\mathbb R^{3n}\right)}\right]\\
      & \ \ = E_\alpha -2\mathfrak{Re}\! \left[\Big\langle \Psi_{\alpha},\xi_p(Y_n)L(v) \Psi_{\alpha}  \Big\rangle_{L^2\! \left(\underset{n\in \mathbb N}{\bigcup}\mathbb R^{3n}\right)}\right]+|p|^2\left\langle \Psi_{\alpha}, \left[1-F(\rho_{Y_n})\right]^2 \Psi_{\alpha}\right\rangle_\mathcal{F}.
 \end{align*}
 Regarding Eq.~(\ref{Eq:NEW_T_1_Smallness}), note that $\mathrm{supp}\! \left(1-F(\rho_{Y_n})\right)\subseteq \underset{n\in \mathbb N}{\bigcup}\mathbb R^{3n}\setminus \Omega_{\frac{\sigma-\kappa}{\sqrt{2}}}$ by Eq.~(\ref{Eq:Support_Chain}), hence
 \begin{align*}
     \left\langle \Psi_{\alpha}, \left[1-F(\rho_{Y_n})\right]^2 \Psi_{\alpha}\right\rangle_\mathcal{F}\leq \sum_{n=0}^\infty \int_{\mathbb R^{3n}\setminus \Omega_{\frac{\sigma-\kappa}{\sqrt{2}}}^{(n)}}\left|\Psi_\alpha(Y_n)\right|^2\mathrm{d}Y_n=:P_\alpha\! \left(\frac{\sigma-\kappa}{\sqrt{2}}\right).
 \end{align*}
 In Lemma \ref{Lem:Most_Of_Mass} we are going to show that $P_\alpha(\lambda)\leq C_\lambda \alpha^{-5}$ for $\lambda>0$ and a suitable $C_\lambda>0$.
\end{proof}

In the remainder of this Section, we are going to isolate the main contribution of the term
\begin{align}
\label{Eq:Raw_Contribution_0}
  \Big\langle \Psi_{\alpha},\xi_p(Y_n)L(v)\Psi_{\alpha}\Big\rangle_{L^2\! \left(\underset{n\in \mathbb N}{\bigcup}\mathbb R^{3n}\right)}
\end{align}
appearing in Eq.~(\ref{Eq:Computation_Trial_Energy_using_Symmetry}). In order to simplify the analysis later, we are first going to address the ultraviolet singularity of the interaction term $v(y)=\pi^{-\frac{3}{2}}|y|^{-2}$. Using the decomposition
\begin{align*}
    v=v^\Lambda+\frac{1}{i}\nabla\cdot w^\Lambda,
\end{align*}
see Eq.~(\ref{Eq:Decomposition_in_Lemma_of_interaction}), we can split the term in Eq.~(\ref{Eq:Raw_Contribution_0}) according to
\begin{align}
\label{Eq:Raw_Contribution_1_PRE}
    \Big\langle \Psi_{\alpha},\xi_p(Y_n)L(v)\Psi_{\alpha}\Big\rangle_{L^2\! \left(\underset{n\in \mathbb N}{\bigcup}\mathbb R^{3n}\right)}=  \Big\langle \Psi_{\alpha},\xi_p(Y_n)L(v^\Lambda)\Psi_{\alpha}\Big\rangle_{L^2\! \left(\underset{n\in \mathbb N}{\bigcup}\mathbb R^{3n}\right)}+\mathcal{E}_1,
\end{align}
where $v^\Lambda$ is an ultraviolet regularized interaction satisfying $\|v^\Lambda\|\leq C\sqrt{\Lambda}$, see Eq.~(\ref{Eq:In_Lemma_v_w_Lambda}), and the residuum $\mathcal{E}_1$ is defined as
\begin{align}
\label{Eq:Def_E_1}
    \mathcal{E}_1 :=\frac{1}{i} \Big\langle \Psi_{\alpha},\xi_p(Y_n)L(\nabla\cdot w^\Lambda)\Psi_{\alpha}\Big\rangle_{L^2\! \left(\underset{n\in \mathbb N}{\bigcup}\mathbb R^{3n}\right)}.
\end{align}
It is content of Lemma \ref{Lem:E_1} to show that the error term is small compared to $\frac{ |p|^2}{\sqrt{\Lambda} \alpha^4}$, i.e. taking the cut-off parameter $\Lambda$ of the order $\alpha^\epsilon$ for $\epsilon>0$ will yield a term small compared to $\alpha^{-4}|p|^2$. In order to further analyse the first term on the right hand side of Eq.~(\ref{Eq:Raw_Contribution_1_PRE})
\begin{align}
\label{Eq:Raw_Contribution_1}
    \Big\langle \Psi_{\alpha},\xi_p(Y_n)L(v^\Lambda)\Psi_{\alpha}\Big\rangle_{L^2\! \left(\underset{n\in \mathbb N}{\bigcup}\mathbb R^{3n}\right)},
\end{align}
we are going to Taylor expand the exponential in the definition of $\xi_p$ in Eq.~(\ref{Eq:Definition_xi_Computation}), i.e. we write
\begin{align}
\label{Eq:Taylor_xi}
    \xi_p(Y_n)=-\frac{1}{2}\left[p \! \cdot \! G(\rho_{Y_{n}})-p \! \cdot \! G(\rho_{Y_{n-1}})\right]^2+\left(\xi_p(Y_n)+\frac{1}{2}\left[p \! \cdot \! G(\rho_{Y_{n}})-p \! \cdot \! G(\rho_{Y_{n-1}})\right]^2\right),
\end{align}
where the expression in $\big(\dots \big)$ is the corresponding Taylor residuum. Based on Eq.~(\ref{Eq:Taylor_xi}), we are going to decompose Eq.~(\ref{Eq:Raw_Contribution_1}) according to
\begin{align}
\label{Eq:Raw_Contribution_2_PRE}
    \Big\langle \!  \Psi_{\alpha},\xi_p(Y_n)L(v^\Lambda)\Psi_{\alpha} \! \Big\rangle_{L^2\! \left(\underset{n\in \mathbb N}{\bigcup}\mathbb R^{3n}\right)} \! = \! - \frac{1}{2} \Big\langle \!  \Psi_{\alpha}, \! \left[p \! \cdot \! G(\rho_{Y_{n}}) \! - \! p \! \cdot \! G(\rho_{Y_{n-1}})\right]^2  \! L(v^\Lambda)\Psi_{\alpha} \! \Big\rangle_{L^2\! \left(\underset{n\in \mathbb N}{\bigcup}\mathbb R^{3n}\right)} \! + \! \mathcal{E}_2,
\end{align}
where the error term $\mathcal{E}_2$ is defined in terms of the Taylor residuum as
\begin{align}
\label{Eq:Def_E_2}
    \mathcal{E}_2:= \Big\langle \!  \Psi_{\alpha},\left(\xi_p(Y_n)+\frac{1}{2}\left[p \! \cdot \! G(\rho_{Y_{n}})-p \! \cdot \! G(\rho_{Y_{n-1}})\right]^2\right)L(v^\Lambda)\Psi_{\alpha} \! \Big\rangle_{L^2\! \left(\underset{n\in \mathbb N}{\bigcup}\mathbb R^{3n}\right)} .
\end{align}
We are going to show in Lemma \ref{Lem:E_2,3} that $\mathcal{E}_2$ is small compared to $\frac{\sqrt{\Lambda} |p|^3}{ \alpha^6}$, i.e. as long as the cut-off parameter $\Lambda$ is chosen small compared to $\frac{\alpha^4}{|p|^2}$, $\mathcal{E}_2$ is of the order $o\! \left(\alpha^{-4}|p|^2\right)$. 

In order to provide an asymptotically correct computation of the first contribution on the right hand side of Eq.~(\ref{Eq:Raw_Contribution_2_PRE})
\begin{align}
\label{Eq:Raw_Contribution_2}
    - \frac{1}{2} \Big\langle \!  \Psi_{\alpha}, \! \left[p \! \cdot \! G(\rho_{Y_{n}}) \! - \! p \! \cdot \! G(\rho_{Y_{n-1}})\right]^2  \! L(v^\Lambda)\Psi_{\alpha} \! \Big\rangle_{L^2\! \left(\underset{n\in \mathbb N}{\bigcup}\mathbb R^{3n}\right)},
\end{align}
it will be essential to make use of the concept of Bose-Einstein condensation, see \cite{BS1,BS2} in the context of the Fröhlich polaron, which tells us that low energy states $\Phi\in L^2\! \left(\mathbb R^3\right)\otimes \mathcal{F}$ of the (cut-off) Hamiltonian $\mathbb H^K$, see Eq.~(\ref{Eq:Cut-off_Hamiltonian}), are close, in a suitable sense, to a superposition of the tensor products with coherent states $\psi^\mathrm{Pek}_z\otimes \Xi_{\varphi^\mathrm{Pek}_z}$ for $z\in \mathbb R^3$ defined in Eq.~(\ref{Eq:Product_with_coherent_state}). Furthermore, in case the (regularized) median is highly localized at the origin $0\in \mathbb R^3$, the low energy state $\Phi$ is even close to the single tensor product $\psi^\mathrm{Pek} \otimes \Xi_{\varphi^\mathrm{Pek}}$, see \cite{BS1}. Unfortunately, we cannot directly apply this result for $\Psi_\alpha$, since $\Psi_\alpha \in \mathcal{F}$ is only an element of the fiber Hilbert space $\mathcal{F}$. In the following we want to lift $\Psi_\alpha$ to a state $\Phi_\alpha\in L^2\! \left(\mathbb R^3\right)\otimes \mathcal{F}$ on the full Hilbert space, in such a way that the (regularized) median is localized around the origin. Given the parameter
\begin{align}
\label{Eq:Definition_eta_from_beta}
    \eta:=\alpha^{-\beta},
\end{align}
where $\frac{2}{29}<\beta<\frac{3}{4}$ is fixed, representing the sharpness of the localization, let us recall the definition of
\begin{align*}
    \tau_{\eta}(y)=\eta^{-\frac{3}{2}}\tau\! \left(\eta^{-1}y\right),
\end{align*}
from Subsection \ref{Subsec:Conventions_and_Definitions}, where $\tau:\mathbb R^3\longrightarrow [0,1]$ is smooth with support in $\left\{y:|y|\leq \frac{1}{3}\right\}$ satisfying $\int\tau(y)^2\mathrm{d}y=1$. Then we define the state $\overline{\Psi}_\alpha\in  L^2\! \left(\mathbb R^3\right)\otimes \mathcal{F}$
\begin{align}
    \label{Eq:Localized_State_LLP}
    \overline{\Psi}_\alpha(x;Y_n):=\mu_\alpha^{-1}\tau_{\eta}\! \left(m_\eta (\rho_{Y_{n}})+x\right)F(\rho_{Y_{n}})\Psi_\alpha(Y_n),
\end{align}
where $\mu_\alpha\in (0,1]$ is chosen such that $\| \overline{\Psi}_\alpha\|=1$, as well as the state $\Phi_\alpha\in  L^2\! \left(\mathbb R^3\right)\otimes \mathcal{F}$ via the map $\mathcal{T}$ from Eq.~(\ref{Eq:Essentially_LLP}) as
\begin{align}
\label{Eq:Def_Phi_State}
    \Phi_\alpha(x;Y_n):=\left(\mathcal{T}  \overline{\Psi}_\alpha\right)\! (x;Y_n)=\mu_\alpha^{-1}\tau_{\eta}\! \left(m_\eta (\rho_{Y_{n}})\right)F(\rho_{Y_{n}})\Psi_\alpha(y_1-x_1,\dots ,y_n-x).
\end{align}
It will be the content of Lemma \ref{Lem:Energy_Estimates_tilde} to show that $\overline{\Psi}_\alpha$ is a low energy state of the (cut-off) Hamiltonian $\mathbb H^K_{\frac{1}{i}\nabla}$ defined in Eq.~(\ref{Eq:Cut-off_Hamiltonian_LLP}), or equivalently $\Phi_\alpha$ is a low energy state of the (cut-off) Hamiltonian $\mathbb H^K$ defined in Eq.~(\ref{Eq:Cut-off_Hamiltonian}). Furthermore, we clearly have for $\Phi_\alpha$ that the regularized median $m_\eta$ is localized around the origin on a length scale of the order $\eta$. Using the fact that we have $\int_{\mathbb R^3} \tau_\eta (m_{\eta}(\rho_{Y_n})+x)^2\mathrm{d}x=1$ for all $Y_n$, we can express the term in Eq.~(\ref{Eq:Raw_Contribution_2}) as
\begin{align}
\nonumber
   & \ \ \ \  -\frac{1}{2} \Big\langle \!  \Psi_{\alpha}, \! \left[p \! \cdot \! G(\rho_{Y_{n}}) \! - \! p \! \cdot \! G(\rho_{Y_{n-1}})\right]^2  \! L(v^\Lambda)\Psi_{\alpha} \! \Big\rangle_{L^2\! \left(\underset{n\in \mathbb N}{\bigcup}\mathbb R^{3n}\right)}\\
   \nonumber
    & = -\frac{1}{2} \left\langle \!  \Psi_{\alpha}, \! \left[p \! \cdot \! G(\rho_{Y_{n}}) \! - \! p \! \cdot \! G(\rho_{Y_{n-1}})\right]^2  \! L(v^\Lambda)\int_{\mathbb R^3} \! \! \tau_\eta(m_{\eta}(\rho_{Y_n})+x)^2 \! \mathrm{d}x\, \Psi_{\alpha} \! \right\rangle_{L^2\! \left(\underset{n\in \mathbb N}{\bigcup}\mathbb R^{3n}\right)}\\
    \nonumber
    & =  \! -\frac{1}{2} \! \int_{\mathbb R^3}   \! \! \! \left\langle \!  \tau_\eta(m_{\eta}(\rho_{Y_{n-1}}) \! + \! x)\Psi_{\alpha}, \! \left[p \! \cdot \! G(\rho_{Y_{n}}) \! - \! p \! \cdot \! G(\rho_{Y_{n-1}})\right]^2  \! L(v^\Lambda) \tau_\eta(m_{\eta}(\rho_{Y_n}) \! + \! x) \, \Psi_{\alpha} \! \right\rangle_{ \!  \!  \! L^2\! \left(\underset{n\in \mathbb N}{\bigcup}\mathbb R^{3n}\right)}\!  \! \mathrm{d}x\\
     \label{Eq:Raw_Contribution_3_Pre}
     & =  \! -\frac{1}{2} \! \int_{\mathbb R^3}  \! \! \left\langle \!  \tau_\eta(m_{\eta}(\rho_{Y_{n}}) \! + \! x)\Psi_{\alpha}, \! \left[p \! \cdot \! G(\rho_{Y_{n}}) \! - \! p \! \cdot \! G(\rho_{Y_{n-1}})\right]^2  \! L(v^\Lambda) \tau_\eta(m_{\eta}(\rho_{Y_n}) \! + \! x) \, \Psi_{\alpha} \! \right\rangle_{ \!  \!  \! L^2\! \left(\underset{n\in \mathbb N}{\bigcup}\mathbb R^{3n}\right)}\! \mathrm{d}x\\
     \nonumber
     & \ \ \ \ \ \  \ \ \ \ \ \  -\frac{1}{2} \Big\langle \!  \Psi_{\alpha}, \! \varphi_\eta(Y_n)\left[p \! \cdot \! G(\rho_{Y_{n}}) \! - \! p \! \cdot \! G(\rho_{Y_{n-1}})\right]^2  \! L(v^\Lambda)\Psi_{\alpha} \! \Big\rangle_{L^2\! \left(\underset{n\in \mathbb N}{\bigcup}\mathbb R^{3n}\right)},
\end{align}
where we have used the commutator relation $L(v^\Lambda)\tau_\eta(m_{\eta}(\rho_{Y_{n}}) \! + \! x)=\tau_\eta(m_{\eta}(\rho_{Y_{n-1}}) \! + \! x)L(v^\Lambda)$ and defined $\varphi_\eta$ as 
\begin{align}
    \label{Eq:Def_F_with_subscript}
    \varphi_\eta(Y_n): &=\int_{\mathbb{R}^3}\left[\tau_{\eta}\! \left(m_\eta(\rho_{Y_{n-1}})+x\right)-\tau_{\eta}\! \left(m_\eta(\rho_{Y_{n}})+x\right)\right]\tau_{\eta}\! \left(m_\eta(\rho_{Y_{n-1}})+x\right)\mathrm{d}x.
\end{align}
By Lemma \ref{Lem:E_2,3}, the residual term $\mathcal{E}_3$, defined as
\begin{align}
\label{Eq:Def_E_3}
     \mathcal{E}_3 :  =-\frac{1}{2}\Big\langle \!  \Psi_{\alpha}, \! \varphi_\eta(Y_n)\left[p \! \cdot \! G(\rho_{Y_{n}}) \! - \! p \! \cdot \! G(\rho_{Y_{n-1}})\right]^2  \! L(v^\Lambda)\Psi_{\alpha} \! \Big\rangle_{L^2\! \left(\underset{n\in \mathbb N}{\bigcup}\mathbb R^{3n}\right)},
\end{align}
is of the order $\frac{\sqrt{\Lambda} |p|^2}{\eta^4 \alpha^8}$, and therefore $o\! \left(\alpha^{-4}|p|^2\right)$ as long as $\eta$ is large compared to $\alpha^{-1}\Lambda^{\frac{1}{8}}$. Finally, in order to extract the essential contribution from the first term on the right hand side of Eq.~(\ref{Eq:Raw_Contribution_3_Pre}), recall that $\Psi_\alpha(Y_n)$ is mostly supported on configurations $Y_n$ such that $F(\rho_{Y_n})=1$, for which we have
\begin{align*}
  \tau_\eta(m_{\eta}(\rho_{Y_{n}}) \! + \! x)\Psi_{\alpha} & =\mu_\alpha\overline{\Psi}_\alpha(x;Y_n),\\ 
 p \! \cdot \! G(\rho_{Y_{n}}) \! - \! p \! \cdot \! G(\rho_{Y_{n-1}}) & =p \! \cdot \! \mathcal{B}(\rho_{Y_{n}}) \! - \! p \! \cdot \! \mathcal{B}(\rho_{Y_{n-1}}).
\end{align*}
Based on this observation we are going to decompose the first term in Eq.~(\ref{Eq:Raw_Contribution_3_Pre}) as
\begin{align*}
   & -\frac{1}{2} \! \int_{\mathbb R^3}  \! \! \left\langle \!  \tau_\eta(m_{\eta}(\rho_{Y_{n}}) \! + \! x)\Psi_{\alpha}, \! \left[p \! \cdot \! G(\rho_{Y_{n}}) \! - \! p \! \cdot \! G(\rho_{Y_{n-1}})\right]^2  \! L(v^\Lambda) \tau_\eta(m_{\eta}(\rho_{Y_n}) \! + \! x) \, \Psi_{\alpha} \! \right\rangle_{L^2\! \left(\underset{n\in \mathbb N}{\bigcup}\mathbb R^{3n}\right)}\! \mathrm{d}x\\
    & \ \ \ \  =-\frac{\mu_\alpha^2}{2} \int_{\mathbb R^3}  \! \! \left\langle   \overline{\Psi}_\alpha(x;\cdot ), \! \left[p \! \cdot \! \mathcal{B}(\rho_{Y_{n}}) \! - \! p \! \cdot \! \mathcal{B}(\rho_{Y_{n-1}})\right]^2  \! L(v^\Lambda) \overline{\Psi}_\alpha(x;\cdot )  \right\rangle_{L^2\! \left(\underset{n\in \mathbb N}{\bigcup}\mathbb R^{3n}\right)}\! \mathrm{d}x+\mathcal{E}_4\\
    &  \ \ \ \  =-\frac{\mu_\alpha^2}{2} \left\langle   \overline{\Psi}_\alpha, \! \left[p \! \cdot \! \mathcal{B}(\rho_{Y_{n}}) \! - \! p \! \cdot \! \mathcal{B}(\rho_{Y_{n-1}})\right]^2  \! L(v^\Lambda) \overline{\Psi}_\alpha  \right\rangle_{L^2\! \left(\mathbb R^3\times \underset{n\in \mathbb N}{\bigcup}\mathbb R^{3n}\right)}+\mathcal{E}_4\\
     &  \ \ \ \  =-\frac{\mu_\alpha^2}{2} \left\langle   \Phi_\alpha, \! \left[p \! \cdot \! \mathcal{B}(\rho_{Y_{n}}) \! - \! p \! \cdot \! \mathcal{B}(\rho_{Y_{n-1}})\right]^2  \! L(v^\Lambda_x) \Phi_\alpha  \right\rangle_{L^2\! \left(\mathbb R^3\times \underset{n\in \mathbb N}{\bigcup}\mathbb R^{3n}\right)}+\mathcal{E}_4,
\end{align*}
where we consider $\overline{\Psi}_\alpha(x;\cdot )$ as a state of $\mathcal{F}$ for all $x\in \mathbb R^3$ and define
\begin{align}
\nonumber
& \mathcal{E}_4:=\frac{1}{2} \int_{\mathbb R^3}  \! \! \Big\langle  \tau_{\eta}\! \left(m_\eta (\rho_{Y_{n}})+x\right)\Psi_\alpha,  \Big(F(\rho_{Y_{n}})F(\rho_{Y_{n-1}})\left[p \! \cdot \! \mathcal{B}(\rho_{Y_{n}}) \! - \! p \! \cdot \! \mathcal{B}(\rho_{Y_{n-1}})\right]^2\\
\nonumber
& -\left[p \! \cdot \! G(\rho_{Y_{n}}) \! - \! p \! \cdot \! G(\rho_{Y_{n-1}})\right]^2\Big)  \! L(v^\Lambda) \tau_{\eta}\! \left(m_\eta (\rho_{Y_{n}})+x\right)\Psi_\alpha(Y_n)\Big\rangle_{L^2\! \left(\underset{n\in \mathbb N}{\bigcup}\mathbb R^{3n}\right)}\! \mathrm{d}x\\
\label{Eq:Def_E_4}
&  \ \ \ \ =-\frac{1}{2} \! \int_{\mathbb R^3}  \! \! \Big\langle \!  \tau_\eta(m_{\eta}(\rho_{Y_{n}}) \! + \! x)\Psi_{\alpha}, \! \left(F(\rho_{Y_{n}}) \! - \! F(\rho_{Y_{n-1}})\right)\\
\nonumber
 &\times \left(p \! \cdot \! G(\rho_{Y_{n}})p \! \cdot \! \mathcal{B}(\rho_{Y_{n}}) \! - \! p \! \cdot \! G(\rho_{Y_{n-1}})p \! \cdot \! \mathcal{B}(\rho_{Y_{n-1}})\right)  \! L(v^\Lambda) \tau_\eta(m_{\eta}(\rho_{Y_n}) \! + \! x) \, \Psi_{\alpha} \! \Big\rangle_{L^2\! \left(\underset{n\in \mathbb N}{\bigcup}\mathbb R^{3n}\right)}\! \mathrm{d}x.
\end{align}
We going to confirm in Lemma \ref{Lem:E_4} that the term $\mathcal{E}_4$ is of the order $\frac{\sqrt{\Lambda} |p|^2}{ \alpha^{6}}$, and therefore small compared to $\alpha^{-4}|p|^2$ for $\Lambda$ small compared to $\alpha^4$.

Summarizing the results of this Subsection, we obtain the following Corollary \ref{Cor:Summarizing_Subsec_3.2}, which extracts the subleading residual terms $\mathcal{E}_0,\dots ,\mathcal{E}_4$ from the quantum energy.
\begin{cor}
\label{Cor:Summarizing_Subsec_3.2}
     Let $f\in C^2_b(\mathbb R^3,\mathbb R^3)$ be reflection anti-symmetric, $\Phi_\alpha\in L^2\! \left(\mathbb R^3\times \underset{n\in \mathbb N}{\bigcup}\mathbb R^{3n}\right)$ the state defined in Eq.~(\ref{Eq:Def_Phi_State}), and let $\mathcal{E}_0,\dots ,\mathcal{E}_4$ be defined in Eq.~(\ref{Eq:NEW_T_1_Smallness}), Eq.~(\ref{Eq:Def_E_1}), Eq.~(\ref{Eq:Def_E_2}), Eq.~(\ref{Eq:Def_E_3}) and Eq.~(\ref{Eq:Def_E_4}), as well as $\mu_\alpha$ below Eq.~(\ref{Eq:Localized_State_LLP}). Then
    \begin{align*}
    \left\langle \Psi_{\alpha,p},\mathbb H_p  \Psi_{\alpha,p}  \right\rangle_\mathcal{F} =E_\alpha & +\mu_\alpha^2\left\langle   \Phi_\alpha, \! \left[p \! \cdot \! \mathcal{B}(\rho_{Y_{n}}) \! - \! p \! \cdot \! \mathcal{B}(\rho_{Y_{n-1}})\right]^2  \! L(v^\Lambda_x) \Phi_\alpha  \right\rangle_{L^2\! \left(\mathbb R^3\times \underset{n\in \mathbb N}{\bigcup}\mathbb R^{3n}\right)}\\
 &  +|p|^2\mathcal{E}_0-2\sum_{j=1}^4 \mathfrak{Re}\! \left[\mathcal{E}_j\right].
\end{align*}
\end{cor}
Note that we have removed the projection onto the real part $\mathfrak{Re}$ in the essential contribution
\begin{align}
\label{Eq:Final_Essential_Contribution}
    \left\langle   \Phi_\alpha, \! \left[p \! \cdot \! \mathcal{B}(\rho_{Y_{n}}) \! - \! p \! \cdot \! \mathcal{B}(\rho_{Y_{n-1}})\right]^2  \! L(v^\Lambda_x) \Phi_\alpha  \right\rangle_{L^2\! \left(\mathbb R^3\times \underset{n\in \mathbb N}{\bigcup}\mathbb R^{3n}\right)},
\end{align}
since $\Phi_\alpha\geq 0$, see Lemma \ref{Lem:Unique_GS}, and $v^\Lambda$ is real-valued, and therefore the term in Eq.~(\ref{Eq:Final_Essential_Contribution}) is an element of $\mathbb R$. Furthermore, it is notable that the expression in Eq.~(\ref{Eq:Final_Essential_Contribution}) is an expectation value with respect to the state $\Phi_\alpha$, which is a low energy state of $\mathbb H^K$ having a localized median, and therefore expected to satisfy Bose-Einstein condensation, see Lemma \ref{Lem:BEC_Imported}. Making use of Bose-Einstein condensation, we are going to evaluate the expression in Eq.~(\ref{Eq:Final_Essential_Contribution}) in the following Section \ref{Sec:Proof_of_Theorem}.

\section{Proof of Theorem \ref{Th:Main}}
In order to verify the main Theorem \ref{Th:Main}, we will make use of the results on the regularized median $m_q$ obtained in Section \ref{Sec:Analysis_of_the_generalized_Median} and the fact that $\Phi_\alpha$ satisfies Bose-Einstein condensation, which we are going to verify in Section \ref{Sec:Bose-Einstein_Condensation}, in order to evaluate the term in Eq.~(\ref{Eq:Final_Essential_Contribution}). As a first step, we are going to identify the leading term in the squared increment in Lemma \ref{Lem:Square_B_Analysis}
\begin{align*}
    \left[p \! \cdot \! \mathcal{B}(\rho_{Y_{n}}) \! - \! p \! \cdot \! \mathcal{B}(\rho_{Y_{n-1}})\right]^2.
\end{align*}
For this purpose, we are going to use Lemma \ref{Lem:Eq:Total_Variation_Upper_Bound} from Section \ref{Sec:Analysis_of_the_generalized_Median}, which tells us that for admissible measures $\rho$ and $\rho'$
\begin{align}
\label{Eq:Differential_median_APPLIED}
       & \ \ \ \ \ \ \ \ \ \left|m_q(\chi*\rho)  +  \int H\mathrm{d}(\rho' \! - \! \rho)  -  m_q(\chi*\rho')\right|\\
\nonumber
        & \leq C \|g \! * \! (\rho' \! - \! \rho)\|_{TV}\! \left(\|g \! * \! (\rho' \! - \! \rho)\|_{TV}  +  \|g \! * \! (\rho \! - \! \rho^\mathrm{Pek})\|_{TV}\right),
\end{align}
where $\rho^\mathrm{Pek}=\rho_{\varphi^\mathrm{Pek}}$ is the Pekar measure and the function $H:=H_{q,\rho^\mathrm{Pek}}^{\chi }:\mathbb R^3\longrightarrow \mathbb R^3$ is defined in components $H=(H_1,H_2,H_3)$ according to Eq.~(\ref{Eq:H_Function}) as
\begin{align*}
      &  \ \ \ \ \   \ \ \ \ \  H_j(y): =\frac{1}{2q \int \mathrm{d}\rho^\mathrm{Pek}} \Bigg(\left[\mathds{1}_{[x^{-}_{j,q}(\rho^\mathrm{Pek}),x^{+}_{j,q}(\rho^\mathrm{Pek})]}(y_j)\, y_j\right]*\chi\\
    & \! +  \! x^+_{j,q}(\rho^\mathrm{Pek})   \left( \! \frac{1}{2} \! + \! q \! - \! f_{j,x^+_{j,q}(\rho^\mathrm{Pek})}(y) \!  \! \right) \! - \! x^-_{j,q}(\rho^\mathrm{Pek}) \left(\frac{1}{2} \! - \! q \! - \! f_{j,x^-_{j,q}(\rho^\mathrm{Pek})}(y) \!  \! \right) \!  \!  \Bigg) .
\end{align*}
Note that the bounded function $f_{j,y}$ is defined in Eq.~(\ref{Eq:Def_f_j_t}), the regularized quantiles $x^-_{j,q}$ and $x^+_{j,q}$ are introduced in Eq.~(\ref{Eq:pm_Quantile}) and the total variation between $\rho$ and $\rho'$ is defined as
\begin{align*}
   \|\rho'-\rho\|_{\mathrm{TV}}:=\sup_{\|f\|_\infty=1}\left|\int f\mathrm{d}\rho'-\int f\mathrm{d}\rho\right|.
\end{align*}
Furthermore, we have used in Eq.~(\ref{Eq:Differential_median_APPLIED}) that $m_q\! \left(\chi* \rho^\mathrm{Pek}\right)=0$, which follows from the fact that the Pekar measure $\rho^\mathrm{Pek}$ is radial. In the following recall that we have chosen the mollifier $\chi$ as a convolution $\chi=g*g$ of another mollifier $g$ and defined $g_T$ as a rescaled version of $g$ in Subsection \ref{Subsec:Conventions_and_Definitions}.

\begin{lem}
\label{Lem:Square_B_Analysis}
    Given a function $f_0\in C_\mathrm{b}^2(\mathbb R^3 , \mathbb R^3)$ and $T>0$, let us make the choice $f:=g_T * f_0$ in Eq.~(\ref{Eq:Boost_Operator}). Then there exists a constant $C>0$ such that for all $T> 0$, and $\rho$ and $\rho'$ satisfying that $(\rho^{\mathrm{Pek}},\rho,\rho')$ is an admissible triple in the sense of Lemma \ref{Lem:Eq:Total_Variation_Upper_Bound}
    \begin{align}
    \nonumber
      &  \left|\Big[ p \! \cdot \! \mathcal{B}(\rho')- p \! \cdot \! \mathcal{B}(\rho) \Big]^2-\left[\int p \! \cdot \! (H+f)\mathrm{d}(\rho'-\rho)-\sum_{j=1}^3\int \partial_{y_j} (p \! \cdot \! f)\mathrm{d}\rho^\mathrm{Pek} \! \!  \int H_j \mathrm{d}(\rho'-\rho)\right]^2\right|\\
      \label{Eq:B_diff_squared_estimate}
      & \ \ \ \ \leq C \|\rho' \! - \! \rho\|^2_{TV}\! \left(\|\rho' \! - \! \rho\|_{TV} + \|g \! * \! (\rho \! - \! \rho^\mathrm{Pek})\|_{TV} +  \|g_T \! * \! (\rho \! - \! \rho^\mathrm{Pek})\|_{TV}\right)|p|^2.
    \end{align}
\end{lem}
\begin{proof}
   Let us first write $\mathcal{B}(\rho')-  \mathcal{B}(\rho)$ as
   \begin{align}
   \nonumber
      & \mathcal{B}(\rho')-   \mathcal{B}(\rho)=m_q(\chi*\rho') -  m_q(\chi*\rho)+\int f(y-m_q(\chi*\rho))\mathrm{d}(\rho'-\rho)\\
      \label{Eq:B_split}
       & \ \ \ \  +\int \big[f(y-m_q(\chi*\rho'))-f(y-m_q(\chi*\rho))\big]\mathrm{d}\rho'.
   \end{align}
   Furthermore, we introduce for the sake of convenience the Landau notation
   \begin{align*}
       O_*:=O\! \left(\|\rho' \! - \! \rho\|_{TV}\! \left(\|\rho' \! - \! \rho\|_{TV}  +  \|g \! * \! (\rho \! - \! \rho^\mathrm{Pek})\|_{TV} + \|g_T\! * \! (\rho \! - \! \rho^\mathrm{Pek})\|_{TV}\right)\right).
   \end{align*}
   By Eq.~(\ref{Eq:Differential_median_APPLIED}), and the fact that taking the convolution with $g$ is a contraction of the total variation, i.e. $\|g*(\nu_1-\nu_2)\|_{\mathrm{TV}}\leq \|\nu_1-\nu_2\|_{\mathrm{TV}}$, we immediately obtain
   \begin{align}
   \label{Eq:Repeat_M_diff}
       m_q(\chi*\rho') -   m_q(\chi*\rho)=\int H \mathrm{d}(\rho'-\rho)+O_*.
   \end{align}
   Since any measure $\nu$ in an admissible triple satisfies $\|\nu\|_{\mathrm{TV}}=\int \mathrm{d}\nu\leq D$, and $m_q\! \left(\chi*\rho^\mathrm{Pek}\right)=0$, we furthermore have $\left|m_q(\chi*\rho)\right|\leq C\|g \! * \! \left(\rho \! - \! \rho^\mathrm{Pek}  \right) \! \|_{\mathrm{TV}}$ for a suitable constant $C>0$. Hence
   \begin{align*}
     & \ \  \ \  \left| \int f(y-m_q(\chi*\rho))\mathrm{d}(\rho'-\rho)-\int f\mathrm{d}(\rho'-\rho)\right|\\
      & =  \left|\int_0^1 \int_{\mathbb R^3}m_q(\chi*\rho) \! \cdot \! \nabla f(y-sm_q(\chi*\rho))\mathrm{d}(\rho'-\rho)(y)\mathrm{d}s\right|\\
      &\leq C\|g \! * \! \left(\rho \! - \! \rho^\mathrm{Pek}  \right) \! \|_{\mathrm{TV}} \|\nabla f\|_{\infty} \|\rho'-\rho\|_{\mathrm{TV}}=O_*.
   \end{align*}
   Regarding the final term in Eq.~(\ref{Eq:B_split}) let us use that for a suitable constant $C>0$
   \begin{align*}
     \left|m_q(\chi*\rho)\right|+\left|m_q(\chi*\rho')\right|\leq C\left(\|\rho'-\rho\|_\mathrm{TV}+\|g \! * \! \left(\rho \! - \! \rho^\mathrm{Pek}  \right) \! \|_{\mathrm{TV}}\right),  
   \end{align*}
   as well as Eq.~(\ref{Eq:Repeat_M_diff}) in order to obtain
   \begin{align*}
      & \ \ \ \ \int \big[f(y-m_q(\chi*\rho'))-f(y-m_q(\chi*\rho))\big]\mathrm{d}\rho'\\
        & =\int \big[f_0(y-m_q(\chi*\rho'))-f_0(y-m_q(\chi*\rho))\big]\mathrm{d}(g_T*\rho')\\
      & =-\sum_{j=1}^3\left(m_q(\chi*\rho')-m_q(\chi*\rho)\right)_j\int_0^1 \partial_{y_j} f_0(y-sm_q(\chi*\rho')-(1-s)m_q(\chi*\rho))\mathrm{d}(g_T* \rho')\\
      & = -\sum_{j=1}^3\int H_j \mathrm{d}(\rho'-\rho)\int_0^1 \partial_{y_j} f_0(y-sm_q(\chi*\rho')-(1-s)m_q(\chi*\rho))\mathrm{d}(g_T* \rho')+O_*\\
       & = -\sum_{j=1}^3\int H_j \mathrm{d}(\rho'-\rho)\int_0^1 \partial_{y_j} f_0(y-sm_q(\chi*\rho')-(1-s)m_q(\chi*\rho))\mathrm{d}(g_T* \rho^\mathrm{Pek})+O_*\\
        & = -\sum_{j=1}^3\int H_j \mathrm{d}(\rho'-\rho)\int_0^1 \partial_{y_j} f_0(y)\mathrm{d}(g_T* \rho^\mathrm{Pek})+O_*\\
        & = -\sum_{j=1}^3\int H_j \mathrm{d}(\rho'-\rho)\int_0^1 \partial_{y_j} f(y)\mathrm{d} \rho^\mathrm{Pek}+O_*.
   \end{align*}
   Summarizing what we have so far yields
   \begin{align}
   \label{Eq:B_diff_no_square}
       \mathcal{B}(\rho')-   \mathcal{B}(\rho)=\int  (H+f)\mathrm{d}(\rho'-\rho)-\sum_{j=1}^3\int \partial_{y_j}  f\mathrm{d}\rho^\mathrm{Pek} \! \!  \int H_j \mathrm{d}(\rho'-\rho)+O_*.
   \end{align}
   Since the terms $\int  (H+f)\mathrm{d}(\rho'-\rho)$ and $\int \partial_{y_j}  f\mathrm{d}\rho^\mathrm{Pek} \! \!  \int H_j \mathrm{d}(\rho'-\rho)$ are bounded by $C\|\rho'-\rho\|_{\mathrm{TV}}$, the left hand side of Eq.~(\ref{Eq:B_diff_squared_estimate}) is bounded for a suitable constant $C>0$ by
   \begin{align*}
       & \ \ \ C\|\rho' \! - \! \rho\|_{TV}^2 \! \left(\|\rho' \! - \! \rho\|_{TV}   \! +   \|g \! * \! (\rho \! - \! \rho^\mathrm{Pek})\|_{TV}+\!  \|g_T \! * \! (\rho \! - \! \rho^\mathrm{Pek})\|_{TV}\right) \\
       &  \ \ \ \ \times \left(\|\rho' \! - \! \rho\|_{TV}   \! +  \|g \! * \! (\rho \! - \! \rho^\mathrm{Pek})\|_{TV} + \!  \|g_T \! * \! (\rho \! - \! \rho^\mathrm{Pek})\|_{TV} \! + \! 1\right) \!  |p|^2\\
       &  \leq C(4D+1)\|\rho' \! - \! \rho\|_{TV}^2\left(\|\rho' \! - \! \rho\|_{TV}   \! +  \|g \! * \! (\rho \! - \! \rho^\mathrm{Pek})\|_{TV} +  \!  \|g_T \! * \! (\rho \! - \! \rho^\mathrm{Pek})\|_{TV}\right)|p|^2,
   \end{align*}
   where we have used that $\|\nu\|_{\mathrm{TV}}=\|g_T* \nu\|_{\mathrm{TV}}=\|g* \nu\|_{\mathrm{TV}}\leq D$ for $\nu\in \{\rho_0,\rho,\rho'\}$.
\end{proof}

In order to successfully apply Lemma \ref{Lem:Square_B_Analysis} for the analysis of the term
\begin{align}
\label{Eq:Main_Contribution_COPY_IN_IV}
    \left\langle   \Phi_\alpha, \! \left[p \! \cdot \! \mathcal{B}(\rho_{Y_{n}}) \! - \! p \! \cdot \! \mathcal{B}(\rho_{Y_{n-1}})\right]^2  \! L(v^\Lambda_x) \Phi_\alpha  \right\rangle_{L^2\! \left(\mathbb R^3\times \underset{n\in \mathbb N}{\bigcup}\mathbb R^{3n}\right)},
\end{align}
it is imperative to control the expectation value of the total variation
\begin{align}
\label{Eq:BEC_in_the_Sense_of_TV}
    \|g_T*(\rho_{Y_n}-\rho^\mathrm{Pek})\|
\end{align}
with respect to the state $\Phi_\alpha$. Estimates for the expression in Eq.~(\ref{Eq:BEC_in_the_Sense_of_TV}), and other results concerning the issue of Bose-Einstein condensation and the closeness of $\Phi_\alpha$ to a product state $\psi^\mathrm{Pek}\otimes \Xi_\varphi$, see Eq.~(\ref{Eq:Product_with_coherent_state}), are being verified in Section \ref{Sec:Bose-Einstein_Condensation} and used in the subsequent Theorem \ref{Sec:Proof_of_Theorem}, which identifies the leading order behaviour of the expression in Eq.~(\ref{Eq:Main_Contribution_COPY_IN_IV}). For the convenience of the reader, we are going to state the main results of Section \ref{Sec:Bose-Einstein_Condensation}, which have been verified in Lemma \ref{Lem:BEC_Imported}, Lemma \ref{Lem:Estimate_Convoluted_Total_Variation} and Lemma \ref{Lem:Electron_Wave_Function} respectively
\begin{align}
   \label{Eq:BEC_Instance_I}
   &   \ \ \ \ \   \ \ \ \ \  \left\langle \Phi_\alpha, W^{-1}_{\varphi^{\mathrm{Pek}}}\mathcal{N}W_{\varphi^{\mathrm{Pek}}} \, \Phi_\alpha\right\rangle_{L^2(\mathbb R^3)\otimes \mathcal{F}}\leq C\alpha^{-\frac{2}{29}},\\
   \label{Eq:BEC_Instance_III}
  & \sum_{n=0}^\infty \int_{\mathbb R^{3n}}\|g_T \! * \! (\rho_{Y_n}-\rho^{\mathrm{Pek}})\|_{\mathrm{TV}} \! \left(\int_{\mathbb R^3} |\Phi_\alpha(x;Y_n)|^2 \mathrm{d}x\right)\! \mathrm{d}Y_n\leq \alpha^{-\epsilon}T^{\frac{3}{2}},\\
        \label{Eq:BEC_Instance_II}
   & \ \ \ \ \   \ \ \ \ \  \ \ \ \ \   \ \ \ \ \  \ \ \ \   \|Q^\mathrm{Pek}_*\Phi_\alpha \|\leq C\alpha^{-\epsilon} ,
\end{align}
where $Q^\mathrm{Pek}_*:=Q^\mathrm{Pek}\otimes 1_\mathcal{F}$, $T>0$ and $C,\epsilon>0$ are suitable constants. \\

\label{Sec:Proof_of_Theorem}
\begin{thm}
\label{Th:Quadratic_Upper_Bound}
 Let $H$ be as in Eq.~(\ref{Eq:Differential_median_APPLIED}), and define on $L^2\! \left(\mathbb R^3\right)$ the linear operator 
    \begin{align*}
        \mathcal{L}(X):=X+2\sum_{j=1}^3 \left\langle \partial_{y_j} \varphi^\mathrm{Pek},X\right\rangle \varphi^\mathrm{Pek}H_j.
    \end{align*}
    Given a bounded function $f_0\in C_\mathrm{b}^2(\mathbb R^3 , \mathbb R^3)$, let us furthermore choose $f:=g_T * f_0$ in Eq.~(\ref{Eq:Boost_Operator}). Then there exist constants $\epsilon>0 $ and $C>0$ such that for $T\geq 1$ and $\Lambda\geq 1$
    \begin{align*}
      &  \ \   \ \     \ \  \ \ \left\langle \Phi_\alpha, \left[p \! \cdot \! \mathcal{B}(\rho_{Y_n})-p \! \cdot \! \mathcal{B}(\rho_{Y_{n-1}})\right]^2 L(v^\Lambda_{x}) \Phi_\alpha \right\rangle_{L^2\! \left(\mathbb{R}^3\times \underset{n\in \mathbb N}{\bigcup}  \mathbb{R}^{3 n}\right)}\\
         &\leq \alpha^{-4}\left\|\sum_{i=1}^3 p_i\mathcal{L}\! \left(\varphi^\mathrm{Pek}f_i\right)+\sum_{i=1}^3 p_i\varphi^\mathrm{Pek} H_i\right\|^2+C\left(\Lambda^{\frac{1}{2}}T^{\frac{3}{2}}\alpha^{-\epsilon}+\Lambda^{-\frac{1}{2}}\right)\alpha^{-4}|p|^2.
    \end{align*}
\end{thm}
\begin{proof}
    Let us define the auxiliary function $\zeta:\mathbb R^3\longrightarrow \mathbb R$
    \begin{align}
    \label{Eq:Definition_zeta_in_proof}
        \zeta(y): & = \left[p \! \cdot \! H(y)+p \! \cdot \! f(y)-\sum_{j=1}^3\int \partial_{y_j} (p \! \cdot \! f)\mathrm{d}\rho^\mathrm{Pek}  H_j(y)\right]^2, \\
        \nonumber
         \theta(Y_n): & =\left[p \! \cdot \! \mathcal{B}(\rho_{Y_n})-p \! \cdot \! \mathcal{B}(\rho_{Y_{n-1}})\right]^2-\alpha^{-4}\zeta(y_n),
    \end{align}
       which allow us to write
    \begin{align}
    \label{Eq:Trivial_Decompostion_zeta_and_theta}
    & \ \  \ \  \ \ \left\langle \Phi_\alpha, \left[p \! \cdot \! \mathcal{B}(\rho_{Y_n})-p \! \cdot \! \mathcal{B}(\rho_{Y_{n-1}})\right]^2 L(v^\Lambda_{x}) \Phi_\alpha \right\rangle_{L^2\! \left(\mathbb{R}^3\times \underset{n\in \mathbb N}{\bigcup}  \mathbb{R}^{3 n}\right)}\\
    \nonumber
        &  =\left\langle \Phi_\alpha, \zeta(y_n) L(v^\Lambda_{x}) \Phi_\alpha \right\rangle_{L^2\! \left(\mathbb{R}^3\times \underset{n\in \mathbb N}{\bigcup}  \mathbb{R}^{3 n}\right)}+\left\langle \Phi_\alpha, \theta(Y_n) L(v^\Lambda_{x}) \Phi_\alpha \right\rangle_{L^2\! \left(\mathbb{R}^3\times \underset{n\in \mathbb N}{\bigcup}  \mathbb{R}^{3 n}\right)}.
    \end{align}
    Since the differences of the empirical measures satisfies $\rho_{Y_n}-\rho_{Y_{n-1}}=\alpha^{-2}\delta_{y_n}$, we clearly have
    \begin{align*}
        \left[\int p \! \cdot \! (H+f)\mathrm{d}(\rho_{Y_n}-\rho_{Y_{n-1}})-\sum_{j=1}^3\int \partial_{y_j} (p \! \cdot \! f)\mathrm{d}\rho^\mathrm{Pek} \! \!  \int H_j \mathrm{d}(\rho_{Y_n}-\rho_{Y_{n-1}})\right]^2=\alpha^{-4}\zeta(y_n).
    \end{align*}
  Furthermore, $\|\rho_{Y_n}-\rho_{Y_{n-1}}\|_{\mathrm{TV}}=\alpha^{-2}$. Consequently, we obtain by Lemma \ref{Lem:Square_B_Analysis}, that there exists a constant $C>0$ such that 
    \begin{align}
    \nonumber
      &    \left|\theta(Y_n)\right|\leq C\alpha^{-4}|p|^2\! \left(\|g \! * \! \left(\rho_{Y_n}-\rho^{\mathrm{Pek}}\right) \! \|_{\mathrm{TV}}+\|g_T \! * \! \left(\rho_{Y_n}-\rho^{\mathrm{Pek}}\right) \! \|_{\mathrm{TV}}+\alpha^{-2}\right)\\
      \nonumber
      & \ \leq C\alpha^{-4}|p|^2 \sqrt{\left(\|g \! * \! \left(\rho_{Y_{n-1}}-\rho^{\mathrm{Pek}}\right) \! \|_{\mathrm{TV}}+\|g_T \! * \! \left(\rho_{Y_{n-1}}-\rho^{\mathrm{Pek}}\right) \! \|_{\mathrm{TV}}+2\alpha^{-2}\right)}\\
      \label{Eq:Estimate_for_admissible_Triples}
        & \ \ \ \    \times \sqrt{\left(\|g \! * \! \left(\rho_{Y_n}-\rho^{\mathrm{Pek}}\right) \! \|_{\mathrm{TV}}+\|g_T \! * \! \left(\rho_{Y_n}-\rho^{\mathrm{Pek}}\right) \! \|_{\mathrm{TV}}+2\alpha^{-2}\right)},
    \end{align}
    for all $Y_n\in \mathbb R^{3n}$ satisfying that $(\rho^\mathrm{Pek},\rho_{Y_{n-1}},\rho_{Y_n})$ is an admissible triple in the sense of Lemma \ref{Lem:Eq:Total_Variation_Upper_Bound}. In the following we want to show that $(\rho^\mathrm{Pek},\rho_{Y_{n-1}},\rho_{Y_n})$ is an admissible triple for all $Y_n\in A_n$, where $A_n$ is the set of all $Y_n\in \mathbb R^{3n}$ such that there exists a $x\in \mathbb R^3$ with
    \begin{align}
    \label{Eq:Two_Sided_Support}
        (x,Y_{n-1}),(x,Y_n)\in \mathrm{supp}(\Phi_\alpha),
    \end{align}
 with the concrete choices
 \begin{align*}
     \tilde \beta: & =\underset{|y_j-x^+_{j,q}(\rho^\mathrm{Pek})|\leq \ell}{\int}\mathrm{d}\rho^\mathrm{Pek}(y)=\underset{|y_j-x^-_{j,q}(\rho^\mathrm{Pek})|\leq \ell}{\int}\mathrm{d}\rho^\mathrm{Pek}(y),\\
     \tilde d: & = \|\varphi^\mathrm{Pek}\|^2-\sigma-\kappa,\\
     \tilde D: & =\|\varphi^\mathrm{Pek}\|^2+\sigma+\kappa,\\
     \tilde \delta: & =\delta_* +\sigma+\kappa,\\
        \tilde q:  & =q,\ \ \ \tilde R:  =R_*,\ \ \  \tilde R':  =3,
 \end{align*}
 for the constants in Lemma \ref{Lem:Eq:Total_Variation_Upper_Bound}, where $\kappa$, $R_*$, $\delta_*$, $\sigma$ and $q$ are chosen according to Subsection \ref{Subsec:Conventions_and_Definitions}. By our definition of $\tilde \beta$ and the fact that $\sigma,\kappa>0$ it is clear that $\nu:=\rho^\mathrm{Pek}$ satisfies 
    \begin{align}
    \label{Eq:COPY_Mass_Conditions}
        \tilde d \leq \int\mathrm{d}\nu\leq \tilde D,\ \ \  \underset{|x-y|>\tilde R}{\int \int}\mathrm{d}\nu\mathrm{d}\nu\leq \tilde \delta,
    \end{align}
and $\underset{|y_j-x^\pm_{j,q}(\rho^\mathrm{Pek})|\leq \ell}{\int}\mathrm{d}\rho^{\mathrm{Pek}}(y)\geq \tilde \beta$. By the support properties of $\Phi_\alpha$ we furthermore have $Y_n\in \Omega_{\sigma+\kappa}$ and $Y_{n-1}\in \Omega_{\sigma+\kappa}$ for $Y_{n}\in A_n$, and therefore Eq.~(\ref{Eq:COPY_Mass_Conditions}) holds for $\nu\in \{\rho_{Y_n},\rho_{Y_{n-1}}\}$. We also note that $\tau_\eta(m_\eta(\nu))\neq 0$ for $\nu\in \{\rho_{Y_n},\rho_{Y_{n-1}}\}$ by the support properties of $\Phi_\alpha$, and hence $|m_\eta(\nu)|\leq \eta\leq 1$, see the definition of $\tau_\eta(\cdot )$ below Eq.~(\ref{Eq:Localized_State_LLP}). Since $m_q\! \left(\rho^\mathrm{Pek}\right)=0$ we obtain
\begin{align*}
    \left|m_\eta\! \left(\rho_{Y_n}\right) - m_\eta\! \left(\rho_{Y_{n-1}}\right)\right|+\left| m_\eta\! \left(\rho_{Y_{n-1}}\right)-m_\eta\! \left(\rho^{\mathrm{Pek}}\right)\right|\leq 3=\tilde R',
\end{align*}
and conclude that $(\rho^\mathrm{Pek},\rho_{Y_{n-1}},\rho_{Y_n})$ is an admissible triple for $\alpha$ large enough such that
\begin{align*}
 \eta=\alpha^{-\beta}\leq \tilde q=q.   
\end{align*}
Especially, Eq.~(\ref{Eq:Estimate_for_admissible_Triples}) holds for all $Y_n\in A_n$. Let us furthermore define
    \begin{align*}
        \widetilde \Phi_\alpha(Y_n):=\sqrt{\left(\|g \! * \! \left(\rho_{Y_n}-\rho^{\mathrm{Pek}}\right) \! \|_{\mathrm{TV}}+\|g_T \! * \! \left(\rho_{Y_n}-\rho^{\mathrm{Pek}}\right) \! \|_{\mathrm{TV}}+2\alpha^{-2}\right)}\left|\Phi_\alpha(Y_n)\right|.
    \end{align*}
    By the support properties of $\Phi_\alpha$ we have $\Phi_\alpha(x,Y_n)=0$ in case $n>(\|\varphi^\mathrm{Pek}\|^2+\sigma+\kappa)\alpha^2:=D_\alpha$ and hence we obtain together with the fact that Eq.~(\ref{Eq:Estimate_for_admissible_Triples}) holds for all $Y_n\in A_n$
    \begin{align}
    \label{Eq:Estimate_theta_1}
      & \ \ \ \ \   \ \ \ \ \    \ \ \ \ \    \ \ \ \ \    \left|  \left\langle \Phi_\alpha,\theta(Y_n) L(v^\Lambda_{x}) \Phi_\alpha \right\rangle_{L^2\! \left(\mathbb{R}^3\times \underset{n\in \mathbb N}{\bigcup}  \mathbb{R}^{3 n}\right)}  \right| \\
      \nonumber
      &   \ \ \ \ \    \ \ \ \ \  \leq  \!  \!  \! \sum_{n=1}^{D_\alpha} \!  \! \frac{\sqrt{n}}{\alpha} \!  \! \int_{\mathbb R^3}\int_{A_n}  \!  \! \left|\Phi_\alpha(x,Y_n)\right| \left|\Phi_\alpha(x,Y_{n-1})\right| \left|\theta(Y_n)\right| \left|v^\Lambda_{x}(y_n)\right| \mathrm{d}Y_n\mathrm{d}x\\
      \nonumber
        &  \lesssim  \alpha^{-4}|p|^2 \!  \sum_{n=1}^{D_\alpha}         \int_{\mathbb R^{3n}}  \!  \!  \widetilde \Phi_\alpha(Y_n)  \widetilde \Phi_\alpha(Y_{n-1}) \! \left|v^\Lambda_{x}(y_n)\right|\mathrm{d}Y_n  \! \leq  \! \sqrt{D+\kappa}\, C \alpha^{-4}|p|^2 \|v^\Lambda\| \| \widetilde \Phi_\alpha\|^2 \! \! .
    \end{align}
    Note that $\|v^\Lambda\|\lesssim \Lambda^{\frac{1}{2}}$, see Lemma \ref{Lem:Semiclassical_objects_properties}, and by Eq.~(\ref{Eq:BEC_Instance_III}) there exists an $\epsilon>0$ such that
    \begin{align}
    \nonumber
         \| \widetilde \Phi_\alpha\|^2 & = \left\langle \Phi_\alpha,\left(\|g \! * \! \left(\rho_{Y_n}-\rho^{\mathrm{Pek}}\right) \! \|_{\mathrm{TV}}+\|g_T \! * \! \left(\rho_{Y_n}-\rho^{\mathrm{Pek}}\right) \! \|_{\mathrm{TV}}+2\alpha^{-2}\right)\Phi_{\alpha}\right\rangle_{L^2\! \left(\mathbb{R}^3\times \underset{n\in \mathbb N}{\bigcup}  \mathbb{R}^{3 n}\right)}\\
            \label{Eq:Estimate_theta_2}
         & \lesssim \alpha^{-\epsilon} + \alpha^{-\epsilon} T^{\frac{3}{2}}+\alpha^{-2}\lesssim  \alpha^{-\epsilon} T^{\frac{3}{2}}.
    \end{align}
Combining Eq.~(\ref{Eq:Trivial_Decompostion_zeta_and_theta}) with the estimates in Eq.~(\ref{Eq:Estimate_theta_1}) and Eq.~(\ref{Eq:Estimate_theta_2}) yields for a suitable $C>0$
    \begin{align}
    \label{Eq:Summarize_zeta_1}
           &   \  \   \ \ \left\langle \Phi_\alpha, \left[p \! \cdot \! \mathcal{B}(\rho_{Y_n})-p \! \cdot \! \mathcal{B}(\rho_{Y_{n-1}})\right]^2 L(v^\Lambda_{x}) \Phi_\alpha \right\rangle_{L^2\! \left(\mathbb{R}^3\times \underset{n\in \mathbb N}{\bigcup}  \mathbb{R}^{3 n}\right)}\\
           \nonumber
         &  \leq \alpha^{-4}\left\langle \Phi_\alpha,\zeta(y_n)L(v^\Lambda_{x}) \Phi_\alpha \right\rangle_{L^2\! \left(\mathbb{R}^3\times \underset{n\in \mathbb N}{\bigcup}  \mathbb{R}^{3 n}\right)}+C\Lambda^{\frac{1}{2}}T^{\frac{3}{2}}\alpha^{-(4+\epsilon)}|p|^2.
    \end{align}
    Recalling the definition of the operators $L$ and $a^*$, we compute explicitly using the permutation symmetry of $\Phi_\alpha$
    \begin{align}
    \nonumber 
       & \left\langle \Phi_\alpha,\zeta(y_n)L(v^\Lambda_{x}) \Phi_\alpha \right\rangle_{L^2\! \left(\mathbb{R}^3\times \underset{n\in \mathbb N}{\bigcup}  \mathbb{R}^{3 n}\right)}=\sum_{n=1}^\infty \frac{\sqrt{n}}{\alpha}\int_{\mathbb{R}^{3 n}}\overline{\Phi_\alpha(Y_n)}\Phi_\alpha(Y_{n-1})\zeta(y_n)v^\Lambda_{x}(y_n)\mathrm{d}Y_n\\
           \nonumber 
       & \ =\left\langle \Phi_\alpha,L(\zeta v^\Lambda_{x}) \Phi_\alpha \right\rangle_{L^2\! \left(\mathbb{R}^3\times \underset{n\in \mathbb N}{\bigcup}  \mathbb{R}^{3 n}\right)}=\left\langle \Phi_\alpha,a^*(\zeta v^\Lambda_{x}) \Phi_\alpha \right\rangle_{L^2(\mathbb R^3)\otimes \mathcal{F}}\\
             \label{Eq:Only_x_Exp_and_BEC}
       & \ = \left\langle \Phi_\alpha, \left\langle\varphi^{\mathrm{Pek}},\zeta v^\Lambda_{x}\right\rangle \Phi_\alpha \right\rangle_{L^2(\mathbb R^3)\otimes \mathcal{F}}+\left\langle \Phi_\alpha,\Big(a^*(\zeta v^\Lambda_{x})-\left\langle\varphi^{\mathrm{Pek}},\zeta v^\Lambda_{x}\right\rangle\Big) \Phi_\alpha \right\rangle_{L^2(\mathbb R^3)\otimes \mathcal{F}}.
    \end{align}
    Regarding the second term in Eq.~(\ref{Eq:Only_x_Exp_and_BEC}), let us define $a^*_j:=a^*(u_j)$ for an orthonormal basis $\{u_j:j\in \mathbb N\}$ of $L^2\! \left(\mathbb R^3\right)$ and estimate for $\lambda>0$
    \begin{align}
    \nonumber
        &  \left\langle \Phi_\alpha,\Big( \! a^*(\zeta v^\Lambda_{x}) \! - \! \left\langle\varphi^{\mathrm{Pek}},\zeta v^\Lambda_{x}\right\rangle \!  \Big) \Phi_\alpha \right\rangle_{ \!  \!  \! L^2(\mathbb R^3)\otimes \mathcal{F}} \!  \! = \!\left\langle  \! \Phi_\alpha,  \sum_j \braket{u_j,\zeta v^\Lambda_{x}}\left(a_j^* \! - \! \left\langle\varphi^{\mathrm{Pek}},u_j\right\rangle\right)  \! \Phi_\alpha  \! \right\rangle_{ \!  \!  \!  \! L^2(\mathbb R^3)\otimes \mathcal{F}}\\
          \label{Eq:zeta_v_Estimate}
        &  \! \leq  \! \lambda  \! \left\langle  \! \!  \Phi_\alpha,  \!  \sum_j  \!  | \! \braket{u_j,\zeta v^\Lambda_{x}} \! |^2 \Phi_\alpha \!  \! \right\rangle_{ \!  \!  \!  \! L^2(\mathbb R^3)\otimes \mathcal{F}} \!  \!     \!  \!  \!  \!  + \lambda^{-1} \!  \! \left\langle  \!  \! \Phi_\alpha ,  \! \sum_j  \! \left(a_j^* \! - \! \left\langle\varphi^{\mathrm{Pek}},u_j\right\rangle \right)  \! \left(a_j \! - \! \left\langle \varphi^{\mathrm{Pek}},u_j\right\rangle\right) \! \Phi_\alpha   \! \! \right\rangle_{ \!  \!  \!  \! L^2(\mathbb R^3)\otimes \mathcal{F}} \!  \!  \!  \!  \! .
    \end{align}
  In order to analyse the term in Eq.~(\ref{Eq:zeta_v_Estimate}), we use the Weyl operator $W_{\varphi^\mathrm{Pek}}$ introduced in Eq.~(\ref{Eq:Def_Weyl_Operator}), in order to write
    \begin{align*}
  \sum_j  \! \left(a_j^* \! - \! \left\langle\varphi^{\mathrm{Pek}},u_j\right\rangle \right)  \! \left(a_j \! - \! \left\langle \varphi^{\mathrm{Pek}},u_j\right\rangle\right)=  W^{-1}_{\varphi^{\mathrm{Pek}}}\mathcal{N}W_{\varphi^{\mathrm{Pek}}},
\end{align*}
 and note that by Eq.~(\ref{Eq:BEC_Instance_I}) there exists a $C>0$ such that
\begin{align*}
     \left\langle \Phi_\alpha, W^{-1}_{\varphi^{\mathrm{Pek}}}\mathcal{N}W_{\varphi^{\mathrm{Pek}}} \, \Phi_\alpha\right\rangle_{L^2(\mathbb R^3)\otimes \mathcal{F}}\leq C\alpha^{-\frac{2}{29}}.
\end{align*}
Since we have
\begin{align*}
 \sum_j |\braket{u_j,\zeta v^\Lambda_{x}}|^2=\|\zeta v^\Lambda_{x}\|^2\lesssim |p|^4\|  v^\Lambda_{x}\|^2=|p|^4\|  v^\Lambda\|^2\lesssim \Lambda |p|^4, 
\end{align*}
we obtain for the choice $\lambda:=\alpha^{-\frac{1}{29}}\Lambda^{-\frac{1}{2}}|p|^{-2}$ and a suitable constant $C>0$ the estimate
\begin{align}
 \label{Eq:Summarize_zeta_2}
    \left\langle \Phi_\alpha,\Big(a^*(\zeta v^\Lambda_{x})-\left\langle\varphi^{\mathrm{Pek}},\zeta v^\Lambda_{x}\right\rangle\Big) \Phi_\alpha \right\rangle_{L^2(\mathbb R^3)\otimes \mathcal{F}}\leq C\Lambda^{\frac{1}{2}}\alpha^{-\frac{1}{29}}|p|^2.
\end{align}
In order to analyse the first term in Eq.~(\ref{Eq:Only_x_Exp_and_BEC}), let us define the multiplication operator
\begin{align*}
    \mathcal{G}:= \left\langle\varphi^{\mathrm{Pek}},\zeta v^\Lambda_{x}\right\rangle
\end{align*}
acting on $L^2(\mathbb R^3)$ and let us denote with $\mathcal{G}\otimes 1_{\mathcal{F}}$ the corresponding operator on $L^2(\mathbb R^3)\otimes \mathcal{F}$, which allows us to write
\begin{align}
 \label{Eq:Summarize_zeta_extra}
    \left\langle \Phi_\alpha, \left\langle\varphi^{\mathrm{Pek}},\zeta v^\Lambda_{x}\right\rangle \Phi_\alpha \right\rangle_{L^2(\mathbb R^3)\otimes \mathcal{F}}=\left\langle  \Phi_\alpha,  \mathcal{G}\otimes 1_{\mathcal{F}}\Phi_\alpha\right\rangle.
\end{align}
Recalling the definition of $\psi^\mathrm{Pek}$, $P^\mathrm{Pek}$ and $Q^\mathrm{Pek}$ in Subsection \ref{Subsec:Semi-Classical_Objects}, let us define $P^\mathrm{Pek}_*:=P^\mathrm{Pek}\otimes 1_{\mathcal{F}}$ and $Q^\mathrm{Pek}_*:=Q^\mathrm{Pek}\otimes 1_{\mathcal{F}}$, and compute
\begin{align*}
& \ \  \ \left\langle \Phi_\alpha,  \mathcal{G} \otimes 1_{\mathcal{F}}\, \Phi_\alpha \right\rangle_{L^2(\mathbb R^3)\otimes \mathcal{F}}=\left\langle \psi^\mathrm{Pek},\mathcal{G}\psi^\mathrm{Pek}\right\rangle - \|Q^\mathrm{Pek}_*\Phi_\alpha\|^2\left\langle \psi^\mathrm{Pek},\mathcal{G} \psi^\mathrm{Pek}\right\rangle\\
&  +2\mathfrak{Re}\left\langle P^\mathrm{Pek}_*\Phi_\alpha, \mathcal{G}\otimes 1_{\mathcal{F}}\, Q^\mathrm{Pek}_*\Phi_\alpha \right\rangle_{L^2(\mathbb R^3)\otimes \mathcal{F}}+\left\langle Q^\mathrm{Pek}_*\Phi_\alpha, \mathcal{G}\otimes 1_{\mathcal{F}}\, Q^\mathrm{Pek}_*\Phi_\alpha \right\rangle_{L^2(\mathbb R^3)\otimes \mathcal{F}},
\end{align*}
Together with the fact that the operator norm satisfies for a suitable $C>0$
\begin{align*}
    \|\mathcal{G}\|_{\mathrm{op}}\leq \|\zeta\|_\infty \|\varphi^\mathrm{Pek}\|\, \|v^\Lambda\|\leq C|p|^2 \Lambda^{\frac{1}{2}},
\end{align*}
see Eq.~(\ref{Eq:In_Lemma_v_w_Lambda}), we obtain by Eq.~(\ref{Eq:BEC_Instance_II}) that there exists an $\epsilon>0$ and a $D>0$ such that
\begin{align}
 \label{Eq:Summarize_zeta_3}
   \left| \left\langle \Phi_\alpha, \mathcal{G} \otimes 1_{\mathcal{F}}\, \Phi_\alpha \right\rangle_{L^2(\mathbb R^3)\otimes \mathcal{F}} - \left\langle \psi^\mathrm{Pek},\mathcal{G}\psi^\mathrm{Pek}\right\rangle\right|\leq 4C|p|^2\Lambda^{\frac{1}{2}}\|Q^\mathrm{Pek}_*\Phi_\alpha\|\leq D\Lambda^{\frac{1}{2}}\alpha^{-\epsilon}|p|^2.
\end{align}
Combining the upper bound in Eq.~(\ref{Eq:Summarize_zeta_1}) with the identities in Eq.~(\ref{Eq:Only_x_Exp_and_BEC}) and Eq.~(\ref{Eq:Summarize_zeta_extra}), and the estimates in Eq.~(\ref{Eq:Summarize_zeta_2}) and Eq.~(\ref{Eq:Summarize_zeta_3}), yields for a suitable $C>0$
\begin{align*}
       &   \left\langle \Phi_\alpha, \left[p \! \cdot \! \mathcal{B}(\rho_{Y_n})-p \! \cdot \! \mathcal{B}(\rho_{Y_{n-1}})\right]^2 L(v^\Lambda_{x}) \Phi_\alpha \right\rangle_{L^2\! \left(\mathbb{R}^3\times \underset{n\in \mathbb N}{\bigcup}  \mathbb{R}^{3 n}\right)} \\
       & \ \  \ \  \    \leq \alpha^{-4}\left\langle \psi^\mathrm{Pek},\mathcal{G}\psi^\mathrm{Pek}\right\rangle + C\Lambda^{\frac{1}{2}}T^{\frac{3}{2}}\alpha^{-(4+\epsilon)}|p|^2.
\end{align*}
We use $\left|\psi^{\mathrm{Pek}}\right|^2 * v=\varphi^\mathrm{Pek}$, see Lemma \ref{Lem:Semiclassical_objects_properties}, which tells us furthermore that
\begin{align*}
  & \left|\psi^{\mathrm{Pek}}\right|^2 * v^\Lambda=\chi(|\nabla|\leq \Lambda)\! \left(\left|\psi^{\mathrm{Pek}}\right|^2 * v\right)=\chi(|\nabla|\leq \Lambda)\varphi^\mathrm{Pek} ,\\
    & \ \ \ \left\|\varphi^\mathrm{Pek}-\left|\psi^{\mathrm{Pek}}\right|^2 * v^\Lambda\right\|=\|\chi(|\nabla|>\Lambda)\varphi^\mathrm{Pek}\|\leq \Lambda^{-\frac{1}{2}},
\end{align*}
and therefore $\|\zeta\|_\infty\lesssim |p|^2$ yields for a suitable $C>0$
\begin{align*}
     \left\langle \psi^\mathrm{Pek},\mathcal{G}\psi^\mathrm{Pek}\right\rangle=\left\langle\varphi^{\mathrm{Pek}},\zeta \left|\psi^{\mathrm{Pek}}\right|^2 * v^\Lambda\right\rangle \leq \left\langle\varphi^{\mathrm{Pek}},\zeta \varphi^{\mathrm{Pek}}\right\rangle+C \Lambda^{-\frac{1}{2}}|p|^2.
\end{align*}
Finally, using $-\int \partial_{y_j}  f\mathrm{d}\rho^\mathrm{Pek}=-\int \partial_{y_j}  f \! \left|\varphi^{\mathrm{Pek}}\right|^2\!\mathrm{d}x=2\int  \varphi^{\mathrm{Pek}}f\partial_{y_j} \varphi^{\mathrm{Pek}}\mathrm{d}x$, we compute
\begin{align*}
   \left\langle\varphi^{\mathrm{Pek}},\zeta \varphi^{\mathrm{Pek}}\right\rangle & = \left\|p\! \cdot \! \left(H+f+2\sum_{j=1}^3\int  \varphi^{\mathrm{Pek}}f\partial_{y_j} \varphi^{\mathrm{Pek}}\mathrm{d}x H_j\right)\! \varphi^{\mathrm{Pek}}\right\|^2\\
   & =\left\|\sum_{i=1}^3 p_i\mathcal{L}\! \left(\varphi^\mathrm{Pek}f_i\right)+\sum_{i=1}^3 p_i\varphi^\mathrm{Pek} H_i\right\|^2.
\end{align*}
\end{proof}

\begin{rem}
\label{Rem:Optimizing}
   In the following we want to choose the function $f=(f_1,f_2,f_3):\mathbb R^3\longrightarrow \mathbb R^3$ such that it minimizes the following expression appearing in Theorem \ref{Th:Quadratic_Upper_Bound}
   \begin{align*}
        \left\|\sum_{i=1}^3 p_i\mathcal{L}\! \left(\varphi^\mathrm{Pek}f_i\right)+\sum_{i=1}^3 p_i\varphi^\mathrm{Pek} H_i\right\|^2.
   \end{align*}
   Restricted to functions $\varphi^\mathrm{Pek}f_i\in L^2(\mathbb R^3)$, this minimization problem is clearly equivalent to the variational problem
   \begin{align}
   \label{Eq:Optimization}
       \inf_{f:\varphi^\mathrm{Pek}f\in L^2(\mathbb R^3,\mathbb R^3)}\left\|\sum_{i=1}^3 p_i\mathcal{L}\! \left(\varphi^\mathrm{Pek}f_i\right) \! + \! \sum_{i=1}^3 p_i\varphi^\mathrm{Pek} H_i\right\|^2 \! = \! \inf_{X\in L^2(\mathbb R^3)}\left\| \mathcal{L}\! \left(X\right)+Y\right\|^2 \! = \! \|\pi Y\|^2 \! ,
   \end{align}
   with $Y:=\sum_{i=1}^3 p_i\varphi^\mathrm{Pek} H_i$ and $\pi$ being the orthogonal projection onto $\mathcal{L}\! \left(L^2(\mathbb R^3)\right)^\perp$. In order to identify $\mathcal{L}\! \left(L^2(\mathbb R^3)\right)$, note that $\mathcal{L}(X)=X$ in case $X\perp \partial_j \varphi^\mathrm{Pek}$ for all $j\in \{1,2,3\}$, hence
   \begin{align}
   \label{Eq:Inclusion}
      \left\{\partial_{y_1} \varphi^\mathrm{Pek},\partial_{y_2} \varphi^\mathrm{Pek},\partial_{y_3} \varphi^\mathrm{Pek}\right\}^\perp\subseteq \mathcal{L}\! \left(L^2(\mathbb R^3)\right).
   \end{align}
It turns out that the inclusion in Eq.~(\ref{Eq:Inclusion}) is even an identity. To verify this, recall that $\varphi\mapsto m_q(\chi*\rho_\varphi)$ is a boost in the sense of Eq.~(\ref{Eq:Introduction_Translation_Covariant}), which especially means that
\begin{align*}
    m_q\! \left(\chi*\rho^{\mathrm{Pek}}_z\right)=m_q\! \left(\chi*\rho_{\varphi^\mathrm{Pek}_z}\right)=m_q\! \left(\chi*\rho_{\varphi^\mathrm{Pek}}\right)+z=z,
\end{align*}
where $\mathrm{d}\rho^{\mathrm{Pek}}_z(x):=\left|\varphi^\mathrm{Pek}(x - z)\right|^2\mathrm{d}x$, and therefore as a consequence of Lemma \ref{Lem:Eq:Total_Variation_Upper_Bound}
\begin{align}
\label{Eq:Inner_Product_T-Inv}
    \delta_{i,j}=\partial_{z_i} z_j\Big|_{z=0}=\partial_{z_i} m_q\! \left(\chi*\rho^{\mathrm{Pek}}_z\right)_j\Big|_{z=0}=-\int H_j \partial_{y_i} \left|\varphi^\mathrm{Pek}\right|^2\mathrm{d}x=-2\braket{\varphi^\mathrm{Pek}H_j,\partial_{y_i} \varphi^\mathrm{Pek}}.
\end{align}
With this at hand we compute for $i\in \{1,2,3\}$ and $X\in L^2\! \left(\mathbb R^3\right)$
\begin{align*}
   \left\langle \partial_{y_i} \varphi^\mathrm{Pek},\mathcal{L}(X)\right\rangle=\left\langle \partial_{y_i} \varphi^\mathrm{Pek},X\right\rangle+2\sum_{j=1}^3 \left\langle \partial_{y_j} \varphi^\mathrm{Pek},X\right\rangle  \left\langle \partial_{y_i} \varphi^\mathrm{Pek}, \varphi^\mathrm{Pek}H_j\right\rangle=0.
\end{align*}
Hence, $\mathcal{L}\! \left(L^2(\mathbb R^3)\right)= \left\{\partial_{y_1} \varphi^\mathrm{Pek},\partial_{y_2} \varphi^\mathrm{Pek},\partial_{y_3} \varphi^\mathrm{Pek}\right\}^\perp$, or equivalently $\pi$ is the orthogonal projection onto $ \mathrm{span} \! \left\{\partial_{y_1} \varphi^\mathrm{Pek},\partial_{y_2} \varphi^\mathrm{Pek},\partial_{y_3} \varphi^\mathrm{Pek}\right\}$. Using again Eq.~(\ref{Eq:Inner_Product_T-Inv}), we can therefore express the right hand side of Eq.~(\ref{Eq:Optimization}) as
\begin{align*}
    \|\pi Y\|^2=\sum_{j=1}^3 \frac{\left|\braket{\partial_{y_j} \varphi^\mathrm{Pek},Y}\right|^2}{\|\partial_{y_j} \varphi^\mathrm{Pek}\|^2}=\sum_{j=1}^3 \frac{\left|\sum_{i=1}^3 p_i\braket{\partial_{y_j} \varphi^\mathrm{Pek},\varphi^\mathrm{Pek} H_i}\right|^2}{\|\partial_{y_j} \varphi^\mathrm{Pek}\|^2}=\sum_{j=1}^3 \frac{p_j^2}{4\|\partial_{y_j} \varphi^\mathrm{Pek}\|^2}.
\end{align*}
Due to the rotational symmetry of $\varphi^{\mathrm{Pek}}$ we have $\|\partial_{y_j} \varphi^\mathrm{Pek}\|^2=\frac{1}{3}\|\nabla \varphi^\mathrm{Pek}\|^2$ for $j\in \{1,2,3\}$, and therefore we obtain by Eq.~(\ref{Eq:Optimization})
\begin{align}
\label{Eq:Variational_LP_Constant}
     \inf_{f:\varphi^\mathrm{Pek}f\in L^2(\mathbb R^3,\mathbb R^3)}\left\|\sum_{i=1}^3 p_i\mathcal{L}\! \left(\varphi^\mathrm{Pek}f_i\right) \! + \! \sum_{i=1}^3 p_i\varphi^\mathrm{Pek} H_i\right\|^2=\frac{3|p|^2}{4\|\nabla \varphi^\mathrm{Pek}\|^2}=\frac{|p|^2}{2m_{\mathrm{LP}}},
\end{align}
with the Landau-Pekar constant $m_{\mathrm{LP}}=\frac{2}{3}\|\nabla \varphi^\mathrm{Pek}\|^2$ introduced in Eq.~(\ref{Eq:Effective_Mass_Def}). Finally, we note that a ($p$-independent) optimizer of Eq.~(\ref{Eq:Variational_LP_Constant}) is given by
\begin{align*}
    f^\diamond_i:=\frac{(\pi-1) \! \left(\varphi^\mathrm{Pek} H_i\right)}{\varphi^\mathrm{Pek}} =-\frac{\partial_i \varphi^{\mathrm{Pek}}}{m_{\mathrm{LP}}\varphi^{\mathrm{Pek}}}-H_i.
\end{align*}
\end{rem}

\begin{proof}[Proof of Theorem \ref{Th:Main}]
  Let us define the function $f:\mathbb R^3 \longrightarrow \mathbb R^3$ in Eq.~(\ref{Eq:Boost_Operator}) as
   \begin{align*}
       f:=g_T * f^\diamond=-g_T*\! \left(\frac{\nabla \varphi^{\mathrm{Pek}}}{m_{\mathrm{LP}}\varphi^{\mathrm{Pek}}}+H\right),
   \end{align*}
   and note that $f(-y)=-f(y)$, and therefore we have by Corollary \ref{Cor:Summarizing_Subsec_3.2}
    \begin{align*}
    \left\langle \Psi_{\alpha,p},\mathbb H_p  \Psi_{\alpha,p}  \right\rangle_\mathcal{F} \leq E_\alpha & +\mu_\alpha^2 \left\langle   \Phi_\alpha, \! \left[p \! \cdot \! \mathcal{B}(\rho_{Y_{n}}) \! - \! p \! \cdot \! \mathcal{B}(\rho_{Y_{n-1}})\right]^2  \! L(v^\Lambda_x) \Phi_\alpha  \right\rangle_{L^2\! \left(\mathbb R^3\times \underset{n\in \mathbb N}{\bigcup}\mathbb R^{3n}\right)}\\
 &  +|p|^2\mathcal{E}_0-2\sum_{j=1}^4 \mathfrak{Re}\! \left[\mathcal{E}_j\right],
\end{align*}
   where $\mu_\alpha\in (0,1]$ is defined below Eq.~(\ref{Eq:Localized_State_LLP}). Since $f^\diamond\in C^2_b(\mathbb R^3,\mathbb R^3)$, see Lemma \ref{Lem:Semiclassical_objects_properties}, we have
   \begin{align*}
    \|f\|_\infty+\|\nabla f\|_\infty\leq  \|f^\diamond\|_\infty+\|\nabla f^\diamond\|_\infty< C   
   \end{align*}
for all $T>0$, and consequently we obtain using the estimates on $\mathcal{E}_0,\dots ,\mathcal{E}_4$ from Lemma \ref{Lem:Rather_Explicit_Expression}, Lemma \ref{Lem:E_1}, Lemma \ref{Lem:E_2,3} and Lemma \ref{Lem:E_4} that there exists a constant $C>0$ such that 
   \begin{align*}
 &     E_\alpha(p)\leq \left\langle \Psi_{\alpha,p},\mathbb H_p \Psi_{\alpha,p}\right\rangle_{\mathcal{F}}\leq E_\alpha \! + \! \mu_\alpha^2 \left\langle \Phi_\alpha, \left[p \! \cdot \! \mathcal{B}(\rho_{Y_n}) \! - \! p \! \cdot \! \mathcal{B}(\rho_{Y_{n-1}})\right]^2 L(v^\Lambda_{x}) \Phi_\alpha \right\rangle_{L^2\! \left(\mathbb{R}^3\times \underset{n\in \mathbb N}{\bigcup}  \mathbb{R}^{3 n}\right)}\\
       & \ \ \ \ \  \ \ \ \ \  \ \ \ \ \  \ \ \ \ \  \ \ \ \ \ +C|p|^2\alpha^{-4}\left\{\alpha^{-1}\Lambda^{\frac{1}{2}}+\Lambda^{-\frac{1}{2}}+\alpha^{-2}\Lambda^{\frac{1}{2}}|p|\right\}.
   \end{align*}
 In combination with Theorem \ref{Th:Quadratic_Upper_Bound} and the fact that $0<\mu_\alpha\leq 1$ we obtain for suitable constants $\epsilon',\kappa,C>0$, and the concrete choice $\Lambda:=\alpha^{\kappa}$, for all $|p|\leq \alpha^{2-\epsilon'}$
   \begin{align*}
        E_\alpha(p)\leq E_\alpha+\alpha^{-4}\left\|\sum_{i=1}^3 p_i\mathcal{L}\! \left(\varphi^\mathrm{Pek}f_i\right)+\sum_{i=1}^3 p_i\varphi^\mathrm{Pek} H_i\right\|^2+C|p|^2\alpha^{-(4+\epsilon')}T^{\frac{3}{2}}.
   \end{align*}
   In order to compare this with the right hand side of Eq.~(\ref{Eq:Variational_LP_Constant}), we use that $f^{\diamond}$ is an element of $C_\mathrm{b}^1(\mathbb R^3, \mathbb R^3)$ and therefore 
   \begin{align*}
       \|f-f^\diamond\|_\infty\leq \|\nabla f\|_\infty \int_{\mathbb R^3}|y|T^3g(Ty)\mathrm{d}y\leq \frac{C}{T},
   \end{align*}
for a suitable constant $C>0$. Since the operator $\mathcal{L}$ is bounded and $\varphi^{\mathrm{Pek}}\in L^2\! \left(\mathbb R^3\right)$, we obtain for a suitable constant $C>0$
\begin{align*}
   & \left\|\sum_{i=1}^3 p_i\mathcal{L}\! \left(\varphi^\mathrm{Pek}f_i\right) \! + \!  \! \sum_{i=1}^3 p_i\varphi^\mathrm{Pek} H_i\right\|^2 \!  \!  \!   \! \leq   \!  \! \left\|\sum_{i=1}^3 p_i\mathcal{L}\! \left(\varphi^\mathrm{Pek}f^\diamond_i\right) \! + \!  \! \sum_{i=1}^3 p_i\varphi^\mathrm{Pek} H_i\right\|^2  \!  \!  \!  \! + \! \frac{C|p|^2}{T}   \! = \! \frac{|p|^2}{2m_{\mathrm{LP}}} \! + \! \frac{C|p|^2}{T}   \!    ,
\end{align*}
where we have used Eq.~(\ref{Eq:Variational_LP_Constant}). Choosing $T:=\alpha^{\xi}$ with $0<\frac{3}{2}\xi<\epsilon'$ and defining $\epsilon:=\min\left\{\xi,\epsilon' -\frac{3}{2}\xi\right\}$ yields for a suitable constant $C>0$
\begin{align*}
    E_\alpha(p)\leq E_\alpha+\frac{|p|^2}{2 \alpha^4 m_{\mathrm{LP}}}+C\alpha^{-(4+\epsilon)}|p|^2.
\end{align*}
Recalling the definition of the effective mass $m_\mathrm{eff}(\alpha)$ in Eq.~(\ref{Eq:Effective_Mass_Def}), we therefore have 
\begin{align*}
    m_\mathrm{eff}(\alpha)=\lim_{p\rightarrow 0}\frac{|p|^2}{2\left(E_\alpha(p)-E_\alpha\right)}\geq \left(1+\frac{2C\alpha^{-\epsilon}}{m_{\mathrm{LP}}}\right)^{-1} \! \!  \alpha^4 m_{\mathrm{LP}}\geq \alpha^4 m_{\mathrm{LP}}-2C\alpha^{4-\epsilon}.
\end{align*}
\end{proof}

\section{Analysis of the generalized Median}
\label{Sec:Analysis_of_the_generalized_Median}
It is the goal of this Section to analyse the regularized median, and especially to identify the leading order term of the increment
\begin{align*}
    m_q(\chi*\rho')-m_q(\chi*\rho),
\end{align*}
as $\rho'\longrightarrow \rho$ in a suitable topology, see Lemma \ref{Lem:Eq:Total_Variation_Upper_Bound}. Since the regularized median $m_q$ involves the quantiles $x^\lambda$ of the marginal measures $\rho_j$, we are first going to investigate in Lemma \ref{Lem:Comparison_of_quantiles} the increments 
\begin{align*}
    x^{\lambda}\! \left((\chi* \rho')_j\right)- x^{\lambda}\! \left((\chi* \rho)_j\right).\\
\end{align*}

In the following, it will be convenient to define for $t\in \mathbb R$ and a permutation symmetric mollifier $\chi$, the function 
\begin{align}
\label{Eq:Def_f_j_t}
   f_{j,t}^\chi:\begin{cases}
       \mathbb R^3\longrightarrow [0,1],\\
  y\mapsto  (\chi*\mathds{1}_{(-\infty,t]}(\cdot_j))(y)=\int_{\mathbb R^3}\chi(z_j)\mathds{1}_{(-\infty,t]}(y_j-z_j)\mathrm{d}z ,\end{cases}
\end{align}
where we keep track of the function $\chi$ used for the convolution in our notation. Clearly $\partial_t f_{j,t}^\chi=-\partial_{y_j}f_{j,t}^\chi\geq 0$ and there exist constants $\ell,h>0$ such that
\begin{align}
\label{Eq:Lower_Bound_partial_f}
  -\partial_{y_j}f_{j,t}^\chi(y_1,y_2,y_3)\geq h\mathds{1}_{(t-2\ell,t+2\ell)}(y_j)  .
\end{align}
Let us recall the definition of the quantile $x^\lambda(\nu)$ of a measure $\nu$ on $\mathbb R$ in Eq.~(\ref{Eq:Def_Quantile}) as well as the definition of the marginal measure $\rho_j$ in Eq.~(\ref{Eq:Def_Marginal_Measure}). Using the functions $f_{j,t}$ we can express the $\lambda$-quantile of $(\chi* \rho)_j$ as
\begin{align*}
    x^{\lambda}\! \left((\chi* \rho)_j\right)=\sup \Big\{t:\int f_{j,t}^\chi\mathrm{d}\rho\leq \lambda \int \mathrm{d}\rho \Big\},
\end{align*}
Since the marginals distributions of the convoluted measures $(\chi * \rho)_j$ are atom-free, the supremum in the definition of $x^{\lambda}\! \left((\chi* \rho)_j\right)$ is a maximum for all $0<\lambda<1$ and we have
\begin{align*}
    \int \Big(f_{j,x^{\lambda} \left((\chi* \rho)_j\right)}^\chi-\lambda\Big) \mathrm{d}\rho=\int_{-\infty}^{x^{\lambda} \left((\chi* \rho)_j\right)}\mathrm{d}(\chi*\rho)_j-\lambda \int \mathrm{d}(\chi*\rho)_j=0.
\end{align*}

\begin{lem}
\label{Lem:Comparison_of_quantiles}
   Let $\chi:\mathbb R^3 \longrightarrow [0,1]$ be a $C^\infty$ function with compact support and $\int \chi\, \mathrm{d}x=1$. Given $\beta>0$, there exist constants $\gamma,C>0$ such that for all $0< \lambda< 1$ and measures $\rho$ and $\rho'$ satisfying $\underset{|y_j-x^{\lambda} \left((\chi* \rho)_j\right)|\leq \ell}{\int}\mathrm{d}\rho(y)\geq \beta$ and $\|\rho'-\rho\|_{\mathrm{TV}}\leq \gamma$
\begin{align}
\label{Eq:Comparison_of_quantiles}
\left|x^{\lambda}\! \left((\chi* \rho')_j\right)- x^{\lambda}\! \left((\chi* \rho)_j\right)\right|\leq C \|\rho'-\rho\|_{\mathrm{TV}}.
\end{align}
\end{lem}
\begin{proof}
In order to keep the notation light let us write $x:=x^{\lambda}\! \left((\chi* \rho)_j\right)$ and $x':=x^{\lambda}\! \left((\chi* \rho')_j\right)$. We first verify that there exists a constant $c>0$ such that for $y\in (x-\ell,x+\ell)$
    \begin{align}
    \label{Eq:Coercivity_Quantiles}
      c|y-x| \leq  \left|\int (f_{j,y}^\chi-\lambda)\mathrm{d}\rho\right|.
    \end{align}
Making use of the assumption $\underset{|y_1-x^{\lambda} \left((\chi* \rho)_j\right)|\leq \ell}{\int}\mathrm{d}\rho\geq \beta$ and Eq.~(\ref{Eq:Lower_Bound_partial_f}) we obtain
\begin{align*}
  &  \left|\int (f_{j,y}^\chi-\lambda)\mathrm{d}\rho\right|=\left|\int (f_{j,y}^\chi-\lambda)\mathrm{d}\rho-\int (f_{j,x}^\chi-\lambda)\mathrm{d}\rho\right| = |y \! - \! x| \int_0^1 \int (-\partial_{y_j}f_{j,x+s(y-x)}^\chi)\mathrm{d}\rho \mathrm{d}s\\
  & \ \  \ \  \ \ \geq  |y-x| h\int_0^1 \underset{|y_j-x-s(y-x)| \leq 2\ell}{\int}\mathrm{d}\rho(y) \mathrm{d}s \geq  |y-x| h\underset{|y_j-x|\leq \ell}{\int}\mathrm{d}\rho(y)\geq (h\beta)  |y-x| ,
    \end{align*}
    which concludes the proof of Eq.~(\ref{Eq:Coercivity_Quantiles}). In order to prove Eq.~(\ref{Eq:Comparison_of_quantiles}), we are going to distinguish between the cases $x'\geq x$ and $x'\leq x$. In the first case, we define $y:=\min\{x',x+\ell\}\geq x$. Using $y\leq x'$ and Eq.~(\ref{Eq:Coercivity_Quantiles}), and the fact that $|f_{j,y}^\chi-\lambda|\leq 1$, we obtain
    \begin{align*}
        0\geq & \int (f_{j,y}^\chi-\lambda)\mathrm{d}\rho'\geq \int (f_{j,y}^\chi-\lambda)\mathrm{d}\rho- \|\rho'-\rho\|_{\mathrm{TV}}= \left|\int (f_{j,y}^\chi-\lambda)\mathrm{d}\rho\right|-\|\rho'-\rho\|_{\mathrm{TV}}\\
        & \geq c|y-x|-\|\rho'-\rho\|_{\mathrm{TV}}=c(y-x)-\|\rho'-\rho\|_{\mathrm{TV}},
    \end{align*}
    and consequently 
    \begin{align}
    \label{Eq:Quantile_with_min}
        \min\{|x'-x|,\ell\}=y-x\leq \frac{1}{c}\|\rho'-\rho\|_{\mathrm{TV}}.
    \end{align}
    Choosing $\gamma<c\ell$ we have $\frac{1}{c}\|\rho'-\rho\|_{\mathrm{TV}}<\ell$, and therefore $|x-x'|\leq \frac{1}{c}\|\rho'-\rho\|_{\mathrm{TV}}$ by Eq.~(\ref{Eq:Quantile_with_min}). Regarding the other case $x'\leq x$ we define $y:=\max\{x',x-\ell\}$ and proceed similarly
    \begin{align*}
        0\leq \int (f_{j,y}^\chi-\lambda)\mathrm{d}\rho'\leq \int (f_{j,y}^\chi-\lambda)\mathrm{d}\rho+\|\rho'-\rho\|_{\mathrm{TV}}\leq c(y-x)+\|\rho'-\rho\|_{\mathrm{TV}}.
    \end{align*}
    This again yields Eq.~(\ref{Eq:Quantile_with_min}), which concludes the proof.
\end{proof}

In the following we want to quantify the increment $m_q(\chi*\rho')-m_q(\chi*\rho)$ of the regularized median. For this purpose let us introduce the abbreviation
\begin{align}
\label{Eq:pm_Quantile}
    x^{\pm}_{j,q}(\rho):  =x^{\frac{1}{2}\pm q}((\chi*\rho)_j),
\end{align}
and define for $j\in \{1,2,3\}$ the functions $H_{j,q,\rho}^\chi:\mathbb R^3 \longrightarrow \mathbb R$ as
\begin{align*}
   &  \ \ \ \ \   \ \ \ \ \   H_{j,q,\rho}^\chi(y): =\frac{1}{2q \int \mathrm{d}\rho} \Bigg(\left[\mathds{1}_{[x^{-}_{j,q}(\rho),x^{+}_{j,q}(\rho)]}(y_j)\big(y_j-m_q((\chi*\rho)_j)\big)\right]*\chi\\
    & \! +  \! \big(x^+_{j,q}(\rho) \! - \! m_q((\chi*\rho)_j)\big) \! \left( \! \frac{1}{2} \! + \! q \! - \! f_{j,x^+_{j,q}(\rho)}(y) \!  \! \right) \! + \! \big(m_q((\chi*\rho)_j) \! - \! x^-_{j,q}(\rho)\big) \! \left(\frac{1}{2} \! - \! q \! - \! f_{j,x^-_{j,q}(\rho)}(y) \!  \! \right) \!  \!  \!  \! \Bigg) ,
\end{align*}
where we keep track of $q,\rho$ and $\chi$ in our notation. Furthermore, we collect all the scalar valued functions $H_{j,q,\rho}^\chi$ into a single vector valued function $H^\chi_{q,\rho}:\mathbb R^3\longrightarrow \mathbb R^3$, i.e. we define
\begin{align}
\label{Eq:H_Function}
   H_{q,\rho}^\chi(y):=(H_{1,q,\rho}^\chi(y),H_{2,q,\rho}^\chi(y),H_{3,q,\rho}^\chi(y)).
\end{align}
Following the convention introduced in Subsection \ref{Subsec:Conventions_and_Definitions}, let $\chi$ be a convolution $\chi=g*g$, where $g$ is a rotational symmetric mollifier. Then, Lemma \ref{Lem:Eq:Total_Variation_Upper_Bound} identifies the leading order of $m_q(\chi*\rho')   -  m_q(\chi*\rho)$, up to an error depending on the $g$-mollified total variation.

\begin{lem}
\label{Lem:Eq:Total_Variation_Upper_Bound}
   Let $\ell$ be as in Eq.~(\ref{Eq:Lower_Bound_partial_f}) for $f_{j,t}^g$. Given $\beta,q,d,D,\delta,R,R'>0$, we call a triple $(\rho_0,\rho,\rho')$ of (finite, Borel) measures admissible, in case there exists a $0<\tilde{q}\leq q$ such that 
    \begin{align}
    \label{Eq:Comparison_Median_Assumptions}
     &   d \leq \int\mathrm{d}\nu\leq D,\ \ \  \underset{|x-y|>R}{\int \int}\mathrm{d}\nu\mathrm{d}\nu\leq \delta, \ \ \ |m_{\tilde{q}}(\rho')-m_{\tilde{q}}(\rho)|+|m_{\tilde{q}}(\rho)-m_{\tilde{q}}(\rho_0)|\leq R',
    \end{align}
 for any $\nu\in \{\rho_0,\rho,\rho'\}$, and $\underset{|y_j-x^\pm_{j,q}(\rho_0)|\leq \ell}{\int}\mathrm{d}\rho_0(y)\geq \beta$. Let $0<\delta<\frac{d^2}{2}$ and $0<q\leq \frac{1}{2}-\frac{\delta}{d^2}$. Then there exists a constant $C>0$, such that for all admissible triples 
    \begin{align}
               \label{Eq:Total_Variation_Upper_Bound}
       & \ \ \ \ \ \ \ \   \left|m_q(\chi*\rho)  +  \int H_{q,\rho_0}^{\chi }\mathrm{d}(\rho' \! - \! \rho)  -  m_q(\chi*\rho')\right|\\
\nonumber
        & \leq C \|g \! * \! (\rho' \! - \! \rho)\|_{TV}\! \left(\|g \! * \! (\rho' \! - \! \rho)\|_{TV}  +  \|g \! * \! (\rho \! - \! \rho_0)\|_{TV}\right).
    \end{align}
\end{lem}
\begin{proof}
In a first step, we are going to verify for general convolution functions $\chi$ satisfying Eq.~(\ref{Eq:Lower_Bound_partial_f}) for some $h>0$, that
\begin{align}
 \label{Eq:Total_Variation_Upper_Bound_no_g}
    \left|m_q(\chi*\rho)  +  \int H_{q,\rho_0}^{\chi }\mathrm{d}(\rho' \! - \! \rho)  -  m_q(\chi*\rho')\right|\leq C \|\rho' \! - \! \rho\|_{TV}\! \left(\|\rho' \! - \! \rho\|_{TV}  +  \|\rho \! - \! \rho_0\|_{TV}\right).
\end{align}
Let us first identify $m_q((\chi*\rho')_j)-m_q((\chi*\rho)_j)$ as
    \begin{align*}
     &     m_q((\chi*\rho')_j)-m_q((\chi*\rho)_j)=\frac{1}{\int_{x^{-}_{j,q}(\rho')}^{x^{+}_{j,q}(\rho')}\mathrm{d}(\chi*\rho')_j}\int_{x^{-}_{j,q}(\rho')}^{x^{+}_{j,q}(\rho')}y\mathrm{d}(\chi*\rho')_j(y)- m_q((\chi*\rho)_j)\\
        & \ \ \ \  \  \ \ \ \  \  \ \ \ \  \ = \frac{1}{\int_{x^{-}_{j,q}(\rho')}^{x^{+}_{j,q}(\rho')}\mathrm{d}(\chi*\rho')_j}\int_{x^{-}_{j,q}(\rho')}^{x^{+}_{j,q}(\rho')}\Big[y-m_q((\chi*\rho)_j)\Big]\mathrm{d}(\chi*\rho')_j(y)\\
        & \ \ \ \  \  \ \ \ \  \  \ \ \ \  \ =  \frac{1}{\int_{x^{-}_{j,q}(\rho')}^{x^{+}_{j,q}(\rho')}\mathrm{d}(\chi*\rho')_j}\left(S_1+S_2-S_3\right),
    \end{align*}
    where we define
    \begin{align*}
        S_1: & =\int_{x^{-}_{j,q}(\rho)}^{x^{+}_{j,q}(\rho)}\Big[y-m_q((\chi*\rho)_j)\Big]\mathrm{d}(\chi*(\rho'-\rho))_j(y),\\
        S_2: & = \int_{x^{+}_{j,q}(\rho)}^{x^{+}_{j,q}(\rho')}\Big[y-m_q((\chi*\rho)_j)\Big]\mathrm{d}(\chi*\rho')_j(y),\\
         S_2: & = \int_{x^{-}_{j,q}(\rho)}^{x^{-}_{j,q}(\rho')}\Big[y-m_q((\chi*\rho)_j)\Big]\mathrm{d}(\chi*\rho')_j(y).
    \end{align*}
    Let us start with the observation that
    \begin{align}
    \label{Eq:Quantile_Dif_Integral_I}
      &  \int_{x^{+}_{j,q}(\rho)}^{x^{+}_{j,q}(\rho')}\mathrm{d}(\chi*\rho')_j=\int \left(\frac{1}{2}+q -f_{j,x^+_{j,q}(\rho)}^{\chi}\right) \mathrm{d}(\rho'-\rho),\\
          \label{Eq:Quantile_Dif_Integral_II}
       & \int_{x^{-}_{j,q}(\rho)}^{x^{-}_{j,q}(\rho')}\mathrm{d}(\chi*\rho')_j=\int \left(\frac{1}{2}-q -f_{j,x^-_{j,q}(\rho)}^{\chi}\right) \mathrm{d}(\rho'-\rho).
    \end{align}
    We can therefore write
    \begin{align}
    \nonumber
      & m_q((\chi*\rho')_j)-m_q((\chi*\rho)_j)=  \int H_{q,\rho}^{\chi}\mathrm{d}(\rho'  -  \rho)+\left(\frac{1}{2q\int \mathrm{d}\rho'}-\frac{1}{2q\int \mathrm{d}\rho}\right) \! \big(S_1+S_2-S_3\big)\\
      \label{Eq:diff_median_multi-expression}
        &   \ \ +  \frac{1}{2q\int \mathrm{d}\rho }\left(\int_{x^{+}_{j,q}(\rho)}^{x^{+}_{j,q}(\rho')}\Big[y-x^{+}_{j,q}(\rho)\Big]\mathrm{d}(\chi*\rho')_j(y)-\int_{x^{-}_{j,q}(\rho)}^{x^{-}_{j,q}(\rho')}\Big[y-x^{-}_{j,q}(\rho)\Big]\mathrm{d}(\chi*\rho')_j(y)\right).
    \end{align}
      In the following let $r>0$ be large enough such that $\mathrm{supp}(\chi)\subseteq B_r(0)$, and note that we have for measures $\nu\neq 0$ 
    \begin{align*}
        \left|x^{\lambda}\! \left((\chi* \nu)_j\right)-x^{\lambda}\! \left( \nu\right)\right|\leq r.
    \end{align*}
 Furthermore, using \cite[Lemma 3.9]{BS1} together with our assumptions on $\nu \in \{\rho_0,\rho,\rho'\}$ yields
    \begin{align*}
        \left|x^{\frac{1}{2}+q}\! \left( \nu\right)-x^{\frac{1}{2}-q}\! \left( \nu\right)\right|\leq 2R.
    \end{align*}
    In combination with the assumption $|m_{\tilde{q}}(\rho')-m_{\tilde{q}}(\rho)|+|m_{\tilde{q}}(\rho)-m_{\tilde{q}}(\rho_0)|\leq R'$, and the fact that $x^{\frac{1}{2}-q}\! \left( \nu\right)\leq x^{\frac{1}{2}-\tilde{q}}\! \left( \nu\right)\leq m_{\tilde{q}}(\nu)\leq x^{\frac{1}{2}+\tilde{q}}\! \left( \nu\right)\leq x^{\frac{1}{2}+q}\! \left( \nu\right)$, we obtain 
    \begin{align}
    \label{Eq:x_y_diff}
        |x-y|\leq \widehat{R} : = R'+2r+4R, \ \mathrm{for} \ x,y\in \bigcup_{\nu\in \{\rho_0,\rho,\rho'\}}\{x^{-}_{j,q}(\nu), m_q((\chi*\nu)_j),x^{+}_{j,q}(\nu)\}.
    \end{align}
    As a consequence we can estimate $\left|S_1\right|\leq \widehat{R}\|\rho'-\rho\|_{\mathrm{TV}}$ and for $q\in \{2,3\}$ 
    \begin{align*}
        \left|S_q\right|\leq \widehat{R}\int_{x^{\pm}_{j,q}(\rho)}^{x^{\pm}_{j,q}(\rho')}\mathrm{d}(\chi*\rho')_j\leq \widehat{R} \left\|\frac{1}{2}\pm q -f_{j,x^+_{j,q}}^{\chi}\right\|_\infty \|\rho'-\rho\|_{\mathrm{TV}}\leq  \widehat{R} \|\rho'-\rho\|_{\mathrm{TV}},
    \end{align*}
    where we have used Eq.~(\ref{Eq:Quantile_Dif_Integral_I}) and Eq.~(\ref{Eq:Quantile_Dif_Integral_II}). Consequently
    \begin{align}
    \label{Eq:Weighted_S_combined}
        \left|\left(\frac{1}{2q\int \mathrm{d}\rho'}-\frac{1}{2q\int \mathrm{d}\rho}\right) \! \big(S_1+S_2-S_3\big)\right|\leq \frac{3\widehat{R} \|\rho'-\rho\|_{\mathrm{TV}}}{4q^2}\frac{\left|\int \mathrm{d}(\rho'-\rho)\right|}{\int \mathrm{d}\rho \int\mathrm{d}\rho'}\leq  \frac{3\widehat{R} \|\rho'-\rho\|^2_{\mathrm{TV}}}{4q^2 d^2}.
    \end{align}
    Similarly, we can use Eq.~(\ref{Eq:Quantile_Dif_Integral_I}) and Eq.~(\ref{Eq:Quantile_Dif_Integral_I}) in order to estimate
    \begin{align}
    \nonumber
      & \frac{1}{2q\int \mathrm{d}\rho }\left(\int_{x^{+}_{j,q}(\rho)}^{x^{+}_{j,q}(\rho')}\Big[y-x^{+}_{j,q}(\rho)\Big]\mathrm{d}(\chi*\rho')_j(y)-\int_{x^{-}_{j,q}(\rho)}^{x^{-}_{j,q}(\rho')}\Big[y-x^{-}_{j,q}(\rho)\Big]\mathrm{d}(\chi*\rho')_j(y)\right) \\
      \label{Eq:median_modified_integrand}
       & \ \ \ \ \leq \frac{\left|x^{+}_{j,q}(\rho') - x^{+}_{j,q}(\rho)\right|+\left|x^{-}_{j,q}(\rho') - x^{-}_{j,q}(\rho)\right|}{2q d}\|\rho'-\rho\|_{\mathrm{TV}}\leq \frac{\widehat{R}}{q d}\|\rho'-\rho\|_{\mathrm{TV}}.
    \end{align}
    By Eq.~(\ref{Eq:x_y_diff}) we have $\|H_{q,\rho}\|_\infty\leq \frac{3\widehat{R}}{2qd}$. Combining Eq.~(\ref{Eq:diff_median_multi-expression}), Eq.~(\ref{Eq:Weighted_S_combined}) and Eq.~(\ref{Eq:median_modified_integrand}) yields 
    \begin{align}
    \nonumber
     & \left| m_q((\chi*\rho')_j)-m_q((\chi*\rho)_j)\right|\leq \left(\|H_{q,\rho}\|_\infty+\frac{3\widehat{R} \|\rho'-\rho\|_{\mathrm{TV}}}{4q^2 d^2}+\frac{\widehat{R}}{q d}\right)\|\rho'-\rho\|_{\mathrm{TV}}\\
     \label{Eq:Comparison_Median_linear}
      & \ \ \  \ \ \  \ \ \ \leq \left(\frac{3\widehat{R}}{2qd}+\frac{3\widehat{R}D}{4q^2 d^2}+\frac{\widehat{R}}{q d}\right)\|\rho'-\rho\|_{\mathrm{TV}}.
    \end{align}
    Consequently, it is clearly enough to verify Eq.~(\ref{Eq:Total_Variation_Upper_Bound_no_g}) in the case that $\|\rho' \! - \! \rho\|_{TV}  \leq \gamma$ and $\|\rho \! - \! \rho_0\|_{TV}  \leq \gamma$, where we choose $\gamma>0$ as in Lemma \ref{Lem:Comparison_of_quantiles}. Assuming
    \begin{align*}
    &  \|\rho' \! - \! \rho\|_{TV}  \leq \gamma ,\\
    & \|\rho \! - \! \rho_0\|_{TV}  \leq \gamma,
    \end{align*}
   we have by Lemma \ref{Lem:Comparison_of_quantiles} that $\left|x^{\pm }_{j,q}(\rho') - x^{\pm}_{j,q}(\rho)\right|\leq C\|\rho'-\rho\|_{\mathrm{TV}}$, which yields in combination with Eq.~(\ref{Eq:median_modified_integrand})
    \begin{align}
    \label{Eq:median_modified_integrand_II}
           \frac{1}{2q \! \int \mathrm{d}\rho } \! \left( \! \int_{x^{+}_{j,q}(\rho)}^{x^{+}_{j,q}(\rho')} \!  \! \Big[y \! - \! x^{+}_{j,q}(\rho) \! \Big] \! \mathrm{d}(\chi*\rho')_j(y) \! - \! \int_{x^{-}_{j,q}(\rho)}^{x^{-}_{j,q}(\rho')} \!  \! \Big[  y \! - \! x^{-}_{j,q}(\rho) \! \Big] \! \mathrm{d}(\chi*\rho')_j(y)\right) \! \leq  \! \frac{C}{q d}\|\rho' \! - \! \rho\|^2_{\mathrm{TV}} .
    \end{align}
    Finally, we need to analyse the term 
    \begin{align}
    \label{Eq:Elementary_H_estimate}
        \left|\int H_{q,\rho}\mathrm{d}(\rho'  -  \rho)-\int H_{q,\rho_0}\mathrm{d}(\rho'  -  \rho)\right|\leq \|H_{q,\rho}-H_{q,\rho_0}\|_\infty \|\rho' \! - \! \rho\|_{\mathrm{TV}}.
    \end{align}
 By our assumption $\|\rho \! - \! \rho_0\|_{TV}  \leq \gamma$, we can apply Lemma \ref{Lem:Comparison_of_quantiles}, which yields
  \begin{align}
    \label{Eq:New_convolution_median_diff_I}
        \left|x^\pm_{j,q}(\rho)-x^\pm_{j,q}(\rho_0)\right|\leq C \|\rho \! - \! \rho_0\|_{TV}.
   \end{align}
    As an immediate consequence of Eq.~(\ref{Eq:Comparison_Median_linear}) and Eq.~(\ref{Eq:New_convolution_median_diff_I}), we obtain for a suitable $C>0$
    \begin{align*}
        \left|H^{(j)}_{q,\rho}(y)-H^{(j)}_{q,\rho_0}(y)\right|\leq C\|\rho-\rho_0\|_\mathrm{TV}+\frac{1}{2q \int \mathrm{d}\rho_0}\big(|B_1|+|B_2|+|B_3|+|B_4|\big),
    \end{align*}
    with
    \begin{align*}
        B_1: & =\left[\mathds{1}_{[x^{+}_{j,q}(\rho_0), x^{+}_{j,q}(\rho)]}(y_j)\big(y_j-m_q((\chi*\rho_0)_j)\big)\right]*\chi,\\
               B_2: & =\left[\mathds{1}_{[x^{-}_{j,q}(\rho_0), x^{-}_{j,q}(\rho)]}(y_j)\big(y_j-m_q((\chi*\rho_0)_j)\big)\right]*\chi,\\
        B_3 : & = \big( x^+_{j,q}( \rho_0)  -  m_q((\chi*\rho_0)_j)\big) \! \left(  f_{j, x^+_{j,q}( \rho_0)}^{\chi}(y) - f_{j, x^+_{j,q}( \rho)}^\chi (y)    \right) ,\\
               B_4 : & = \big( x^-_{j,q}( \rho_0)  -  m_q((\chi*\rho_0)_j)\big) \! \left(  f_{j, x^-_{j,q}( \rho_0)}^\chi (y) -   f_{j, x^-_{j,q}( \rho)}^\chi (y)    \right) .       
    \end{align*}
    Let us analyse the term $B_1$ in detail, the other terms can be controlled similarly. By Eq.~(\ref{Eq:x_y_diff}), we have for all $y$ such that $ x^{-}_{j,q}(\rho_0)\leq y_j\leq  x^{+}_{j,q}(\rho)$
    \begin{align*}
        |y_j-m_q((\chi*\rho_0)_j)|\leq  x^{+}_{j,q}(\rho)- x^{-}_{j,q}(\rho_0)\leq \widehat{R},
    \end{align*}
and therefore we obtain with the aid of Eq.~(\ref{Eq:New_convolution_median_diff_I}) for a suitable constant $C>0$
\begin{align*}
    & \ \  |B_1|\leq \int_{\mathbb R^3}\chi(y-z)\mathds{1}_{[ x^{+}_{j,q}(\rho_0), x^{+}_{j,q}(\rho)]}(z_j)\left|z_j-m_q((\chi*\rho_0)_j)\right|\mathrm{d}z\\
    & \leq \widehat{R} \left| x^{+}_{j,q}(\rho)- x^{-}_{j,q}(\rho_0)\right|\sup_{t}\int_{\mathbb R^2}\chi(t,u,v)\mathrm{d}u\mathrm{d}v\leq C \|\rho \! - \! \rho_0\|_{TV}.
\end{align*}
Similarly we obtain $|B_j|\leq C\|\rho - \rho_0\|_{TV}$ for $j\in \{2,3,4\}$, and therefore
\begin{align*}
    \|H_{q,\rho}-H_{q,\rho_0}\|_\infty\leq 4C\|\rho  - \rho_0\|_{TV},
\end{align*}
which concludes the proof of Eq.~(\ref{Eq:Total_Variation_Upper_Bound_no_g}) by Eq.~(\ref{Eq:Elementary_H_estimate}). 

In order to verify Eq.~(\ref{Eq:Total_Variation_Upper_Bound}) for the specific convolution function $\chi=g* g$ introduced in Subsection \ref{Subsec:Conventions_and_Definitions}, note that the measures $ \widetilde \rho_0:=g*\rho_0$, $\widetilde \rho:= g*\rho$ and $\widetilde \rho':=g* \rho$ are an admissible triple with $\widetilde R:=R+2\widetilde r$ and $\widetilde R':=R'+2\widetilde r$, where $g(x)=0$ for $|x|\geq \widetilde r$, and
\begin{align*}
    \widetilde \beta:=\beta \inf_{x:|x_j|\leq \ell}\int_{|y_j|\leq \ell}g(y-x)\mathrm{d}y>0.
\end{align*}
Consequently, we obtain by Eq.~(\ref{Eq:Total_Variation_Upper_Bound_no_g}), used for the convolution function $\widetilde \chi:=g$,
\begin{align*}
    & \left|m_q(\chi*\rho)  +  \int H_{q,\rho_0}^{\chi }\mathrm{d}(\rho' \! - \! \rho)  -  m_q(\chi*\rho')\right|= \left|m_q(g*\widetilde \rho)  +  \int H_{q,\widetilde \rho_0}^{g }\mathrm{d}(\widetilde \rho' \! - \! \widetilde \rho)  -  m_q(g*\widetilde \rho')\right|\\
        & \ \ \ \  \ \ \ \  \leq C \|\widetilde \rho' \! - \! \widetilde \rho\|_{TV}\! \left(\|\widetilde \rho' \! - \! \widetilde \rho\|_{TV}  +  \|\widetilde \rho \! - \! \widetilde \rho_0\|_{TV}\right)\\
        & \ \ \ \  \ \ \ \   = C \|g \! * \! (\rho' \! - \! \rho)\|_{TV}\! \left(\|g \! * \! (\rho' \! - \! \rho)\|_{TV}  +  \|g \! * \! (\rho \! - \! \rho_0)\|_{TV}\right).
\end{align*}
\end{proof}

\section{Spatial concentration of Probability}
\label{Sec:Spatial_concentration_of_Probability}

It is the objective of this Section to derive decay properties for the probability measure $|\Psi_\alpha(Y_n)|^2\mathrm{d}Y_n$. To be more precise, we show in Theorem \ref{Th:Powers_of_Median} that moments of the regularized median $m_q(\chi* \rho_{Y_{n}})$ are bounded (uniformly in $\alpha$) on the set of non-degenerate configurations $Y_n$ satisfying $\int \mathrm{d}\rho_{Y_n}\geq d$. The fact that the statistical quantity $m_q(\chi* \rho_{Y_{n}})$ has bounded moments on $\int \mathrm{d}\rho_{Y_n}\geq d$ will be a central tool in Section \ref{Sec:Analysis_of_the_Error_Terms}, where we analyse the various residual terms $\mathcal{E}_1,\dots ,\mathcal{E}_4$ that have been separated from the energy $\left\langle \Psi_{\alpha,p},\mathbb H_p \Psi_{\alpha,p}\right\rangle_{\mathcal{F}}$ in Subsection \ref{Sec:Isolating_the_essential_Contribution}.

\begin{thm}
\label{Th:Powers_of_Median}
   Let $0<q<\frac{1}{2}$, $d>0$ and $m\in \mathbb{N}$. Then there exists a $C<\infty$ such that
    \begin{align*}
        \left\langle \Psi_\alpha,\left|m_q(\chi* \rho_{Y_{n}})\right|^m \chi\! \left(\int \mathrm{d}\rho_{Y_{n}}\geq d\right)\Psi_\alpha\right\rangle_{\mathcal{F}}\leq C.
    \end{align*}
\end{thm}

Before we come to the proof of Theorem \ref{Th:Powers_of_Median}, we need the auxiliary Lemma \ref{Lem:Bad_quantile_support_Bad_Energy}, Lemma \ref{Lem:IMS_Mass_Away_From_Electron} and Lemma \ref{Lem:Quantile_Mass_Estimate_by_IMS}, which will provide us with strong bounds on the probabilities
\begin{align}
\label{Eq:Probability_t_lambda}
    P_\alpha(s,\lambda):=  \sum_{n=0}^\infty \underset{\mathbb R^{3n}\setminus \Omega_{s,\lambda}^{(n)}}{\int}|\Psi_\alpha(Y_n)|^2\, \mathrm{d}Y_n,
\end{align}
where $\Omega_{s,\lambda}^{(n)}$ is the set of all $Y_n\in \mathbb R^{3n}$ such that
\begin{align*}
  \int\chi(|y|>s)\mathrm{d}\rho_{Y_n}< \lambda.  
\end{align*}
Note that $\Omega_{s,\lambda}^{(n)}$ depends on $\alpha$ as well via the empirical measure $\rho_{Y_n}$ defined in Eq.~(\ref{Eq:Empirical_Measure}). Let us furthermore introduce $\Omega_{s,\lambda}:=\underset{n\in \mathbb N}{\bigcup} \Omega_{s,\lambda}^{(n)}$ and the shifted version $\Omega^{(n)}_{s,\lambda}(x)$ as the set of all $Y_n\in \mathbb R^{3n}$ satisfying 
\begin{align*}
 \int\chi(|y-x|>s)\mathrm{d}\rho_{Y_n}< \lambda.   
\end{align*}

\begin{lem}
    \label{Lem:Bad_quantile_support_Bad_Energy}
    For any $\lambda>0$, there exists a constant $s_0$ such that for all $s\geq s_0$ and all states $\Phi\in L^2\! \left(\mathbb{R}^3\right)\otimes \mathcal{F}$ with $\mathrm{supp}(\Phi)\subseteq \underset{n\in \mathbb N,x\in \mathbb R^3}{\bigcup}\{x\}\times \! \left(\mathbb R^{3n}\setminus \Omega_{s,\lambda}^{(n)}(x)\right)$
    \begin{align*}
        \left\langle \Phi,  \mathbb H\Phi\right\rangle_{L^2\! \left(\mathbb{R}^3\right)\otimes \mathcal{F}}\geq E_\alpha+\frac{\lambda}{4}.
    \end{align*}
\end{lem}
\begin{proof}
    Following the approach in \cite{LY}, respectively \cite[Theorem 2.5]{BS1}, let us localize the electron coordinate $x$ by introducing the states 
    \begin{align*}
       \Phi_{x'}(x;Y_n):=\tau_s(x-x')\Phi(x;Y_n),
    \end{align*}
    where $\tau_s(y):=s^{-\frac{3}{2}}\tau(s^{-1}y)$ and $\tau$ is a smooth $[0,1]$-valued function with $\tau(y)=0$ in case $|y|>\frac{1}{3}$ and $\int \tau(y)^2\mathrm{d}y=1$, see Subsection \ref{Subsec:Conventions_and_Definitions}. It is elementary to check that
    \begin{align}
    \label{Eq:IMS_in_x_Space}
         \left\langle \Phi,  \mathbb H\Phi\right\rangle_{L^2\! \left(\mathbb{R}^3\right)\otimes \mathcal{F}} & =\int_{\mathbb R^3} \left\langle \Phi_{x'},  \mathbb H \Phi_{x'}\right\rangle_{L^2\! \left(\mathbb{R}^3\right)\otimes \mathcal{F}}\mathrm{d}x'-s^{-2}\left\|\nabla \tau\right\|^2,
    \end{align}
  see for example \cite[Theorem 2.5]{BS1}. In order to analyse $\left\langle \Phi_{x'},  \mathbb H \Phi_{x'}\right\rangle_{L^2\! \left(\mathbb{R}^3\right)\otimes \mathcal{F}}$, let us introduce
    \begin{align*}
        v^{\leq }_{x,x',s}(y): & =\chi\! \left(|y-x'|\leq \frac{s}{2}\right)v_x(y),\\
        v^{> }_{x,x',s}(y): & =v_x(y)-v^{\leq }_{x,x',s}(y)=\chi\! \left(|y-x'|> \frac{s}{2}\right)v_x(y),
    \end{align*}
    as well as the multiplication operators $\mathcal{N}^{\leq }_{x',s}=\mathcal{N}^{\leq }_{x',s}(\rho_{Y_n})$ and $\mathcal{N}^{>}_{x',s}=\mathcal{N}^{>}_{x',s}(\rho_{Y_n})$, and the operator $\mathbb H_{x',s}$ as 
    \begin{align*}
        \mathcal{N}^{\leq }_{x',s}(\rho_{Y_n}): & =\int \chi \! \left(|y-x'|\leq \frac{s}{2}\right)\mathrm{d}\rho_{Y_n}, \\ 
        \mathcal{N}^{> }_{x',s}(\rho_{Y_n}): & =\int \chi\! \left(|y-x'|> \frac{s}{2}\right)\mathrm{d}\rho_{Y_n},\\
        \mathbb H_{x',s} := - \Delta_x & +\mathcal{N}^{\leq }_{x',s}-a\! \left(v^{\leq }_{x,x',s}\right)^*-a\! \left(v^{\leq }_{x,x',s}\right).
    \end{align*}
    Using the set $A_{x'}:=\left\{x\in \mathbb R^3:|x-x'|\leq \frac{s}{2}\right\}$ we can split the space $L^2\! \left(\mathbb{R}^3\right)\otimes \mathcal{F}$ into the tensor factors
    \begin{align}
    \label{Eq:Tensor_Split}
        L^2\! \left(\mathbb{R}^3\right)\otimes \mathcal{F}\cong L^2\! \left(\mathbb{R}^3\right)\otimes \mathcal{F}\! \left(A_{x'}\right)\otimes  \mathcal{F}\! \left(\mathbb R^3\setminus A_{x'}\right),
    \end{align}
    and we observe that the operator $\mathbb H_{x',s}$ acts as the identity on the last factor in the decomposition in Eq.~(\ref{Eq:Tensor_Split}). Denoting with $\Omega_{x'}$ the vacuum in $\mathcal{F}\! \left(\mathbb R^3\setminus A_{x'}\right)$, we furthermore have for any state $\Phi'\in L^2\! \left(\mathbb{R}^3\right)\otimes \mathcal{F}\! \left(A_{x'}\right)$
    \begin{align*}
        \left\langle\Phi'\otimes \Omega_{x'},\mathbb H_{x',s}\Phi'\otimes \Omega_{x'}\right\rangle_{  L^2\! \left(\mathbb{R}^3\right)\otimes \mathcal{F}} =  \left\langle\Phi'\otimes \Omega_{x'},\mathbb H\Phi'\otimes \Omega_{x'}\right\rangle_{  L^2\! \left(\mathbb{R}^3\right)\otimes \mathcal{F}}\geq E_\alpha,
    \end{align*}
    and therefore we have as an operator inequality
    \begin{align}
    \label{Eq:GS_Energy_Component}
        \mathbb H_{x',s}\geq E_\alpha.
    \end{align}
    Making use of the decomposition
    \begin{align*}
        \mathbb H=\mathbb H_{x',s}+ \mathcal{N}^{> }_{x',s}-a\! \left(v^{> }_{x,x',s}\right)^*-a\! \left(v^{> }_{x,x',s}\right)
    \end{align*}
    we obtain by Eq.~(\ref{Eq:GS_Energy_Component})
    \begin{align}
    \label{Eq:Applied_GS_Energy_Component}
        \left\langle \Phi_{x'},  \mathbb H \Phi_{x'}\right\rangle_{L^2\! \left(\mathbb{R}^3\right)\otimes \mathcal{F}}\geq E_\alpha\|\Phi_{x'}\|^2+ \left\langle \Phi_{x'}, \left(\mathcal{N}^{> }_{x',s}-a\! \left(v^{> }_{x,x',s}\right)^*-a\! \left(v^{> }_{x,x',s}\right)\right)\Phi_{x'}\right\rangle_{L^2\! \left(\mathbb{R}^3\right)\otimes \mathcal{F}}.
    \end{align}
    Note that $\Phi_{x'}(x;Y_n)=0$ for $|x-x'|> \frac{s}{3}$ and that for $|x-x'|\leq \frac{s}{3}$
    \begin{align*}
        \|v^{> }_{x,x',s}\|^2\lesssim \int_{\mathbb R^3}\frac{\chi \! \left(|y-x'|>\frac{s}{2}\right)}{|y-x|^4}\mathrm{d}y\leq \int_{\mathbb R^3}\frac{\chi \! \left(|y-x|>\frac{s}{6}\right)}{|y-x|^4}\mathrm{d}y\lesssim s^{-1},
    \end{align*}
    and therefore there exists a constant $C>0$ such that we have for all $(x;Y_n)\in \mathrm{supp}(\Phi_{x'})$ the operator inequality on $\mathcal{F}$
    \begin{align*}
        \mathcal{N}^{> }_{x',s}-a\! \left(v^{> }_{x,x',s}\right)^*-a\! \left(v^{> }_{x,x',s}\right)\geq \frac{1}{2}\mathcal{N}^{> }_{x',s}-Cs^{-1}.
    \end{align*}
    Together with Eq.~(\ref{Eq:Applied_GS_Energy_Component}) we therefore obtain
    \begin{align}
    \label{Eq:Phi_x_prime_Estimate}
          \left\langle \Phi_{x'},  \mathbb H \Phi_{x'}\right\rangle_{L^2\! \left(\mathbb{R}^3\right)\otimes \mathcal{F}}\geq \left(E_\alpha-Cs^{-1}\right)\|\Phi_{x'}\|^2+\frac{1}{2} \left\langle \Phi_{x'},  \mathcal{N}^{> }_{x',s} \Phi_{x'}\right\rangle_{L^2\! \left(\mathbb{R}^3\right)\otimes \mathcal{F}}.
    \end{align}
    We observe that $\Phi_{x'}(x;Y_n)\neq 0$ implies $|x-x'|\leq \frac{s}{3}$ as well as $Y_n\notin \Omega^{(n)}_{s,\lambda}(x)$ and hence
    \begin{align*}
        \int \chi\! \left(|y-x'|>\frac{s}{2}\right)\mathrm{d}\rho_{Y_n}\geq \int \chi\! \left(|y-x|>s\right)\mathrm{d}\rho_{Y_n}\geq \lambda.
    \end{align*}
    Consequently, $\left\langle \Phi_{x'},  \mathcal{N}^{> }_{x',s} \Phi_{x'}\right\rangle_{L^2\! \left(\mathbb{R}^3\right)\otimes \mathcal{F}}\geq \lambda\|\Phi_{x'}\|^2$, which yields together with Eq.~(\ref{Eq:IMS_in_x_Space}) and Eq.~(\ref{Eq:Phi_x_prime_Estimate}) for a suitable constant $C>0$ and $s$ large enough
    \begin{align*}
         \left\langle \Phi,  \mathbb H\Phi\right\rangle_{L^2\! \left(\mathbb{R}^3\right)\otimes \mathcal{F}} & \geq \left(E_\alpha+\frac{\lambda}{2}-Cs^{-1}\right)\int_{\mathbb R^{3}}\|\Phi_{x'}\|^2\mathrm{d}x'-s^{-2}\|\nabla \tau\|^2\\
         & =E_\alpha+\frac{\lambda}{2}-Cs^{-1}-s^{-2}\|\nabla \tau\|^2\geq E_\alpha+\frac{\lambda}{4}.
    \end{align*}
\end{proof}

Let us define $g(y):=\chi_{1}(|y|^2>2)$. In the following Lemma \ref{Lem:IMS_Mass_Away_From_Electron}, we are going to quantify the localization error with respect to the localization functions
\begin{align*}
    K_{1,s,\lambda,u}(\rho): & =\chi_u\! \left(\int \! g\! \left(s^{-1}y\right)\mathrm{d}\rho(y)\leq \lambda\right),\\
      K_{2,s,\lambda,u}(\rho): & =\sqrt{1-K_{1,s,\lambda,u}(\rho)^2}.
\end{align*}

\begin{lem}
    \label{Lem:IMS_Mass_Away_From_Electron}
    For any $\lambda,u,z>0$, there exists a constant $C>0$ such that for all $s\geq 1$
    \begin{align*}
        \sum_{j=1}^2\left\langle K_{j,s,\lambda,u}(\rho_{Y_n})\Psi_\alpha , \mathbb{H}_0 K_{j,s,\lambda,u}(\rho_{Y_n})\Psi_\alpha\right\rangle_{ \mathcal{F}}\leq E_\alpha + C P_\alpha(s,\lambda-2u)^{1-z} s^{-\frac{1}{2}}.
    \end{align*}
\end{lem}
\begin{proof}
    We start with the observation that $[K_{j,s,\lambda,u},\mathcal{N}]=0$ and therefore
    \begin{align}
    \label{Eq:Large_s_particle_number}
        \sum_{j=1}^2\left\langle K_{j,s,\lambda,u}(\rho_{Y_n})\Psi_\alpha , \mathcal{N} K_{j,s,\lambda,u}(\rho_{Y_n})\Psi_\alpha\right\rangle_{ \mathcal{F}}=\left\langle \Psi_\alpha, \mathcal{N} \Psi_\alpha\right\rangle_{\mathcal{F}} .
    \end{align}
    Furthermore, we can write
      \begin{align}
      \nonumber
     &    \sum_{j=1}^2 \left\langle \Psi_\alpha,K_{j,s,\lambda,u}(\rho_{Y_n})a(v)^* K_{j,s,\lambda,u}(\rho_{Y_n})\Psi_\alpha\right\rangle_{\mathcal{F}}-\left\langle \Psi_\alpha,a(v)^*\Psi_\alpha\right\rangle_{\mathcal{F}}\\
     \nonumber
      & \ \ \ =\left\langle \Psi_\alpha,\left(\sum_{j=1}^2 K_{j,s,\lambda,u}(\rho_{Y_{n}})K_{j,s,\lambda,u}(\rho_{Y_{n-1}})-1\right)L(v)\Psi_\alpha\right\rangle_{L^2\! \left(\underset{n\in \mathbb N}{\bigcup}  \mathbb{R}^{3 n}\right)}\\
      \label{Eq:Lage_s_K_*_introduction}
      &\ \ \ =\left\langle \Psi_\alpha,K_*(Y_n)L(v)\Psi_\alpha\right\rangle_{L^2\! \left(\underset{n\in \mathbb N}{\bigcup}  \mathbb{R}^{3 n}\right)},
    \end{align}
    with $K_*:\underset{n\in \mathbb N}{\bigcup}  \mathbb{R}^{3 n}\longrightarrow \mathbb R$ defined as
    \begin{align*}
        K_*(Y_n):=\sum_{j=1}^2 K_{j,s,\lambda,u}(\rho_{Y_{n}})K_{j,s,\lambda,u}(\rho_{Y_{n-1}})-1=-\frac{1}{2}\sum_{j=1}^2 \left[K_{j,s,\lambda,u}(\rho_{Y_{n}})-K_{j,s,\lambda,u}(\rho_{Y_{n-1}})\right]^2.
    \end{align*}
    Trivially we have the inequality $|K_*|\leq 1$. We note that $K_*(Y_n)\neq 0$ implies $g\! \left(s^{-1}y_n\right)\neq 0$ and therefore $|y_n|\geq s$. Furthermore, $K_*(Y_n)\neq 0$ implies either $\int \! g\! \left(s^{-1}y\right)\mathrm{d}\rho_{Y_n}(y)>\lambda-u$ or $\int \! g\! \left(s^{-1}y\right)\mathrm{d}\rho_{Y_{n-1}}(y)>\lambda-u$, hence
    \begin{align*}
      Y_{n-1},Y_n\in A:=\underset{n\in \mathbb N}{\bigcup} \mathbb R^{3n}\setminus \Omega_{s,\lambda-2u}
    \end{align*}
   for $\alpha$ large enough such that $\alpha^{-2}\leq u$. Consequently, we have the estimate
    \begin{align}
    \nonumber
        & \ \ \ \ \  \ \ \ \ \ \left|\left\langle \Psi_\alpha,K_*(Y_n)L(v)\Psi_\alpha\right\rangle_{L^2\! \left(\underset{n\in \mathbb N}{\bigcup}  \mathbb{R}^{3 n}\right)}\right|\\
        \nonumber
        & \leq \sum_{n=1}^\infty \frac{\sqrt{n}}{\alpha}\int_{\mathbb R^{3n}}|\mathds{1}_A(Y_{n-1})\Psi_\alpha(Y_{n-1})|\, |\mathds{1}_A(Y_{n})\Psi_\alpha(Y_{n})|\, \chi(|y_n|>s)v(y_n)\mathrm{d}Y_n\\
        \label{Eq:Classic_CS}
        & \leq \left\|\sqrt{\mathcal{N}}\mathds{1}_A\Psi_\alpha\right\|\, \left\|\chi(|\cdot |>s)v\right\|\, \left\|\mathds{1}_A\Psi_\alpha\right\|.
    \end{align}
    Let us compute first
    \begin{align}
    \label{Eq:1_A_Estimate}
      \left\|\mathds{1}_A\Psi_\alpha\right\|^2 & =P(s,\lambda-2u)\leq P(s,\lambda-2u)^{1-z}, \\
      \label{Eq:chi_cdot_s_Estimate}
     \left\|\chi(|\cdot |>s)v\right\|^2  & =\int \chi(|y|>s)v(y)^2\mathrm{d}y\lesssim \int \frac{\chi(|y|>s)}{|y|^4}\mathrm{d}y\lesssim s^{-1}.
    \end{align}
    Furthermore, we obtain by Hölder's inequality for the measure $|\Psi_\alpha(Y_n)|^2 \mathrm{d}Y_n$ and the random variables $Y_n\mapsto \frac{n}{\alpha^2}$ and $Y_n\mapsto \mathds{1}_A(Y_n)$, defined on the space $\underset{n\in \mathbb N}{\bigcup} \mathbb R^{3n}$, for all $z>0$
    \begin{align}
    \nonumber 
        & \ \ \ \ \  \left\|\sqrt{\mathcal{N}}\mathds{1}_A\Psi_\alpha\right\|^2=\sum_{n=1}^\infty \int_{\mathbb R^{3n}}\frac{n}{\alpha^2}\mathds{1}_A(Y_{n})\,  |\Psi_\alpha(Y_{n})|^2\mathrm{d}Y_n\\
            \nonumber 
        & \leq \left(\sum_{n=1}^\infty \int_{\mathbb R^{3n}}\mathds{1}_A(Y_{n})^{\frac{1}{1-z}}\,  |\Psi_\alpha(Y_{n})|^2\mathrm{d}Y_n\right)^{1-z}\left(\sum_{n=1}^\infty \int_{\mathbb R^{3n}}\left(\frac{n}{\alpha^2}\right)^{\frac{1}{z}}\,  |\Psi_\alpha(Y_{n})|^2\mathrm{d}Y_n\right)^{z}\\
        \label{Eq:Hoelder_I}
        & \ \ \ \ \ = P(s,\lambda-2u)^{1-z}\braket{\Psi_\alpha,\mathcal{N}^{\frac{1}{z}}\Psi_\alpha}_{\mathcal{F}}^z\lesssim P(s,\lambda-2u)^{1-z},
    \end{align}
    where we have used $\braket{\Psi_\alpha,\mathcal{N}^{\frac{1}{z}}\Psi_\alpha}_{\mathcal{F}}\lesssim 1$ in the final estimate, see Lemma \ref{Lem:Moments_of_particle_number}. Combining Eq.~(\ref{Eq:Classic_CS}), Eq.~(\ref{Eq:1_A_Estimate}), Eq.~(\ref{Eq:chi_cdot_s_Estimate}) and Eq.~(\ref{Eq:Hoelder_I}) yields for a suitable constant $C>0$
    \begin{align}
     \label{Eq:Lage_s_K_*_Estimate}
        \left|\left\langle \Psi_\alpha,K_*(Y_n)L(v)\Psi_\alpha\right\rangle_{L^2\! \left(\underset{n\in \mathbb N}{\bigcup}  \mathbb{R}^{3 n}\right)}\right|\leq C P(s,\lambda-2u)^{1-z} s^{-\frac{1}{2}}.
    \end{align}
    Finally, we have to understand the localization error of the $\mathcal P^2$ term in $\mathbb H$. For this purpose we apply the IMS formula again, see \cite[Theorem A.1]{LSol}, respectively \cite[Proposition 6.1]{LNSS} and \cite[Lemma 3.3]{BS1}, in order to compute
    \begin{align}
    \nonumber
      &  \ \ \ \ \ \ \   \sum_{j=1}^2\left\langle K_{j,s,\lambda,u}(\rho_{Y_n})\Psi_\alpha , \mathcal P^2 K_{j,s,\lambda,u}(\rho_{Y_n})\Psi_\alpha\right\rangle_{ \mathcal{F}}\\
      \nonumber
         & =\left\langle \Psi_\alpha, \mathcal P^2 \Psi_\alpha\right\rangle_{\mathcal{F}}+\frac{1}{2}\sum_{j=1}^2 \left\langle \Psi_\alpha, \left[\left[K_{j,s,\lambda,u}(\rho_{Y_n}),\mathcal P^2\right],K_{j,s,\lambda,u}(\rho_{Y_n})\right] \Psi_\alpha\right\rangle_{\mathcal{F}}\\
         \label{Large_s_P_f_introduction}
         &   = \left\langle \Psi_\alpha, \mathcal P^2 \Psi_\alpha\right\rangle_{\mathcal{F}}-\frac{1}{2}\sum_{j=1}^2 \left\langle \Psi_\alpha, \left[\mathcal P,K_{j,s,\lambda,u}(\rho_{Y_n})\right]^2 \Psi_\alpha\right\rangle_{\mathcal{F}},
    \end{align}
   where we have used that $K_{j,s,\lambda,u}(\rho_{Y_n})$ commutes with the multiplication operator
    \begin{align*}
        \left[\mathcal P, \! K_{j,s,\lambda,u}(\rho_{Y_n})\right] \! = \! \frac{1}{i} \! \sum_{k=1}^n\partial_{y_k}K_{j,s,\lambda,u}(\rho_{Y_n})  \! = \! \frac{1}{is}\Theta_j\! \left(\int \! g\! \left(s^{-1}y\right) \! \mathrm{d}\rho_{Y_n}(y)\right) \!  \! \int \! \nabla g\! \left(s^{-1}y\right) \! \mathrm{d}\rho_{Y_n}(y),
    \end{align*}
   and $\Theta_j$ is defined as
    \begin{align*}
        \Theta_1(z): & =\nabla_z \chi_{u}(z\leq \lambda),\\
        \Theta_2(z): & = \nabla_z \sqrt{1-\chi_{u}(z\leq \lambda)^2}.
    \end{align*}
    We observe that $\Theta_j\! \left(\int \! g\! \left(s^{-1}y\right)\mathrm{d}\rho(y)\right)\neq 0$ implies $\int \! g\! \left(s^{-1}y\right)\mathrm{d}\rho(y)\geq \lambda-u> \lambda-2u$, and therefore $Y_n\notin \Omega_{s,\lambda-2u}$ or equivalently $Y_n \in A=\underset{n\in \mathbb N}{\bigcup} \mathbb R^{3n}\setminus \Omega_{s,\lambda-2u}$. Applying again Hölder's inequality yields for $z>0$
    \begin{align}
    \nonumber
      &  \left|\left\langle \Psi_\alpha, \left[\mathcal P,K_{j,s,\lambda,u}(\rho_{Y_n})\right]^2 \Psi_\alpha\right\rangle_{\mathcal{F}}\right|=\left|\left\langle \Psi_\alpha, \mathds{1}_A\left[\mathcal P,K_{j,s,\lambda,u}(\rho_{Y_n})\right]^2 \Psi_\alpha\right\rangle_{\mathcal{F}}\right|\\
      \nonumber
      & \leq \left\langle \Psi_\alpha, \mathds{1}_A^{\frac{1}{1-z}} \Psi_\alpha\right\rangle_{\mathcal{F}}^{1-z} \left\langle \Psi_\alpha, \left[\mathcal P,K_{j,s,\lambda,u}(\rho_{Y_n})\right]^{\frac{2}{z}} \Psi_\alpha\right\rangle_{\mathcal{F}}^z\\
      \nonumber
      & \leq s^{-2}P(s,\lambda-2u)^{1-z} \left(\sum_{j=1}^2\|\nabla \Theta_j\|^2_\infty\right)\|\nabla g\|_\infty^2 \left\langle \Psi_\alpha, \left(\int \mathrm{d}\rho_{Y_{n}}\right)^{\frac{2}{z}} \Psi_\alpha\right\rangle_{\mathcal{F}}^z\\
        \label{Large_s_P_f_Estimate}
      & = s^{-2} P(s,\lambda-2u)^{1-z} \left(\sum_{j=1}^2\|\nabla \Theta_j\|^2_\infty\right)\|\nabla g\|_\infty^2 \left\langle \Psi_\alpha, \mathcal{N}^{\frac{2}{z}} \Psi_\alpha\right\rangle_{\mathcal{F}}^z\lesssim s^{-2}P(s,\lambda-2u)^{1-z},
    \end{align}
      where we have used $\braket{\Psi_\alpha,\mathcal{N}^{\frac{2}{z}}\Psi_\alpha}_{\mathcal{F}}\lesssim 1$, see Lemma \ref{Lem:Moments_of_particle_number}, which we are going to prove in the subsequent Section \ref{Sec:Asymptotic_concentration_of_Probability}. Combining Eq.~(\ref{Eq:Large_s_particle_number}), Eq.~(\ref{Eq:Lage_s_K_*_introduction}), Eq.~(\ref{Eq:Lage_s_K_*_Estimate}), Eq.~(\ref{Large_s_P_f_introduction}) and Eq.~(\ref{Large_s_P_f_Estimate}) yields for a suitable constant $C>0$
      \begin{align*}
        & \sum_{j=1}^2\left\langle K_{j,s,\lambda,u}(\rho_{Y_n})\Psi_\alpha , \mathbb{H}_0 K_{j,s,\lambda,u}(\rho_{Y_n})\Psi_\alpha\right\rangle_{ \mathcal{F}} =\sum_{j=1}^2\left\langle K_{j,s,\lambda,u}(\rho_{Y_n})\Psi_\alpha , \mathcal P^2 K_{j,s,\lambda,u}(\rho_{Y_n})\Psi_\alpha\right\rangle_{ \mathcal{F}}\\
        & \! + \!  \sum_{j=1}^2  \! \! \left\langle K_{j,s,\lambda,u}(\rho_{Y_n})\Psi_\alpha , \mathcal{N} K_{j,s,\lambda,u}(\rho_{Y_n})\Psi_\alpha\right\rangle_{ \mathcal{F}} \! - \! 2\sum_{j=1}^2 \! \mathfrak{Re} \! \left\langle K_{j,s,\lambda,u}(\rho_{Y_n})\Psi_\alpha , a(v)^* K_{j,s,\lambda,u}(\rho_{Y_n})\Psi_\alpha\right\rangle_{ \mathcal{F}}\\
        & \ \  \leq \left\langle\Psi_\alpha , \mathcal P^2 \Psi_\alpha\right\rangle_{ \mathcal{F}}+\left\langle \Psi_\alpha , \mathcal{N} \Psi_\alpha\right\rangle_{ \mathcal{F}}-2\mathfrak{Re}\left\langle \Psi_\alpha , a(v)^* \Psi_\alpha\right\rangle_{ \mathcal{F}}+CP(s,\lambda-2u)^{1-z}\! \left(s^{-\frac{1}{2}}+s^{-2}\right)\\
        & \ \ =E_\alpha + CP(s,\lambda-2u)^{1-z}\! \left(s^{-\frac{1}{2}}+s^{-2}\right).
      \end{align*}
\end{proof}

Combining Lemma \ref{Lem:Bad_quantile_support_Bad_Energy} and Lemma \ref{Lem:IMS_Mass_Away_From_Electron}, we can provide strong bounds on the probability $P_\alpha(s,\lambda)$ in the following Lemma \ref{Lem:Quantile_Mass_Estimate_by_IMS}, which we will use together with a layer-cake representation to verify Theorem \ref{Th:Powers_of_Median}.

\begin{lem}
\label{Lem:Quantile_Mass_Estimate_by_IMS}
    For any $\lambda,\gamma>0$, there exists a constant $C>0$ such that for all $s>0$
    \begin{align*}
        P_\alpha(s,\lambda)\leq C s^{-\gamma}.
    \end{align*}
\end{lem}
\begin{proof}
   Note that $P_\alpha(s,\lambda)$ is a probability, i.e. $P_\alpha(s,\lambda)\leq 1$. By induction, it is therefore enough to verify that for any $\lambda,u,z>0$ there exist  constants $C_,s_0$ such that for all $s\geq s_0$
    \begin{align}
    \label{Eq:Spatial_Probability_Iterative}
        P_\alpha(\sqrt{3}s,\lambda+3u)\leq C P_\alpha(s,\lambda)^{1-z} s^{-\frac{1}{2}}.
    \end{align}
    Let us define the functions $K_j:=K_{j,s,\lambda+2u,u}$, introduced above Lemma \ref{Lem:IMS_Mass_Away_From_Electron}, and note that 
    \begin{align}
    \label{Eq:Applied_K_IMS_Estimate}
          \sum_{j=1}^2\left\langle K_{j}(\rho_{Y_n})\Psi_\alpha , \mathbb{H}_0 K_{j}(\rho_{Y_n})\Psi_\alpha\right\rangle_{ \mathcal{F}}\leq E_\alpha + C P_\alpha(s,\lambda)^{1-z} s^{-\frac{1}{2}},
    \end{align}
    by Lemma \ref{Lem:IMS_Mass_Away_From_Electron} for a sufficiently large constant $C>0$. Furthermore, we introduce the state
    \begin{align*}
        \widehat{\Psi}_\alpha:=\frac{1}{\|K_{2}(\rho_{Y_n})\Psi_\alpha\|} K_{2}(\rho_{Y_n})\Psi_\alpha.
    \end{align*}
    We observe that $Y_n\notin \Omega_{\sqrt{3}s,\lambda+3u}$ implies that
    \begin{align*}
        \int g(s^{-1}y)\mathrm{d}\rho_{Y_n}(y)\geq \int \chi(|y|>\sqrt{3}s)\mathrm{d}\rho_{Y_n}(y)\geq \lambda+3u,
    \end{align*}
    and hence $K_{2}(\rho_{Y_n})=\sqrt{1-K_{1,s,\lambda+2u,u}(\rho_{Y_n})^2}=0$ for $Y_n\notin \Omega_{\sqrt{3}s,\lambda+3u}$. We deduce that
    \begin{align*}
         P_\alpha(\sqrt{3}s,\lambda+3u)\leq \|K_{2}(\rho_{Y_n})\Psi_\alpha\|^2=1-\|K_{1}(\rho_{Y_n})\Psi_\alpha\|^2.
    \end{align*}
   Consequently, we obtain together with Eq.~(\ref{Eq:Applied_K_IMS_Estimate}) and the operator inequality $\mathbb H\geq E_\alpha$
    \begin{align*}
        & \ \ \  \ \ \  \ \ \  P_\alpha(\sqrt{3}s,\lambda+3u)\left\langle \widehat \Psi_\alpha , \left(\mathbb{H}_0-E_\alpha\right) \widehat \Psi_\alpha\right\rangle_{ \mathcal{F}}\\
        & \leq \|K_{2}(\rho_{Y_n})\Psi_\alpha\|^2\left\langle \widehat \Psi_\alpha , \left(\mathbb{H}_0-E_\alpha\right) \widehat \Psi_\alpha\right\rangle_{ \mathcal{F}} \leq  C P_\alpha(s,\lambda)^{1-z} s^{-\frac{1}{2}}.
    \end{align*}
    In order to conclude the proof of Eq.~(\ref{Eq:Spatial_Probability_Iterative}), we are going to verify that 
    \begin{align}
    \label{Eq:Spatial_proof_A_Contradiction}
        \left\langle \widehat \Psi_\alpha , \left(\mathbb{H}_0-E_\alpha\right) \widehat \Psi_\alpha\right\rangle_{ \mathcal{F}}\geq \frac{\lambda}{8},
    \end{align}
for $\alpha$ large enough by contradiction. In the following we will assume that Eq.~(\ref{Eq:Spatial_proof_A_Contradiction}) is violated, i.e. we assume
\begin{align}
   \label{Eq:Spatial_proof_A_Contradiction_Assumption}
    \left\langle \widehat \Psi_\alpha , \mathbb{H}_0 \widehat \Psi_\alpha\right\rangle_{ \mathcal{F}}< E_\alpha+\frac{\lambda}{8}.
\end{align}
We further introduce the auxiliary states $\widehat \Phi_\alpha\in L^2\! \left(\mathbb{R}^3\right)\otimes \mathcal{F}\subseteq  L^2\! \left(\mathbb{R}^3\times \underset{n\in \mathbb N}{\bigcup}  \mathbb{R}^{3 n}\right)$ as
 \begin{align*}
     \widehat \Phi_\alpha(x;Y_n):=\tau_{\epsilon}(x)\widehat \Psi_\alpha(y_1-x,\dots ,y_n-x),
 \end{align*}
 where $\tau_\epsilon(x)=\epsilon^{\frac{3}{2}}\tau\! \left(\epsilon x\right)$ and $\tau$ is a $[0,1]$-valued smooth function with compact support and $\int \tau^2\mathrm{d}x=1$, see Subsection \ref{Subsec:Conventions_and_Definitions}. Observe that there exists a constant $C>0$ such that
 \begin{align*}
    & \ \ \ \  \left\langle \widehat \Phi_\alpha,  (-\Delta_x) \widehat \Phi_\alpha \right\rangle_{ L^2\! \left(\mathbb{R}^3\right)\otimes \mathcal{F}}= \left\langle \tau_\epsilon\otimes  \widehat \Psi_\alpha,  \left(\frac{1}{i}\nabla_x - \mathcal P\right)^2 \tau_\epsilon\otimes  \widehat \Psi_\alpha \right\rangle_{ L^2\! \left(\mathbb{R}^3\right)\otimes \mathcal{F}}\\
     &  \leq (1+\epsilon)  \left\langle  \widehat \Psi_\alpha, \mathcal P^2   \widehat \Psi_\alpha \right\rangle_{ \mathcal{F}}+\left(1+\epsilon^{-1}\right)\! \left\langle \tau_\epsilon, (-\Delta_x) \tau_\epsilon \right\rangle_{ L^2\! \left(\mathbb{R}^3\right)} \leq \left\langle  \widehat \Psi_\alpha,  \mathcal P^2   \widehat \Psi_\alpha \right\rangle_{ \mathcal{F}}+C'\! \left(\epsilon+\epsilon^2\right),
     \end{align*}
     where we have used the assumption in Eq.~(\ref{Eq:Spatial_proof_A_Contradiction_Assumption}) and the fact that $\mathcal{P}^2\leq 2\mathbb H_0 +C$ for a suitable constant $C>0$, see Eq.~(\ref{Eq:Particle_Number_Energy_bound}). Moreover, we observe that
     \begin{align*}
     &  \ \ \ \  \left\langle \widehat \Phi_\alpha,  \Big\{\mathcal{N}-a^*(v_x)-a(v_x)\Big\} \widehat \Phi_\alpha \right\rangle_{ L^2\! \left(\mathbb{R}^3\right)\otimes \mathcal{F}}=\left\langle \widehat \Psi_\alpha,  \Big\{\mathcal{N}-a^*(v)-a(v)\Big\} \widehat \Psi_\alpha \right\rangle_{ \mathcal{F}}.
 \end{align*}
Therefore we obtain for $\epsilon$ small enough such that $C'\! \left(\epsilon+\epsilon^2\right)\leq \frac{\lambda}{8}$
 \begin{align}
 \label{Eq:Energy_Estimate_for_Asymptotic_Concentration}
      \left\langle \widehat \Phi_\alpha,  \mathbb H \widehat \Phi_\alpha \right\rangle_{ L^2\! \left(\mathbb{R}^3\right)\otimes \mathcal{F}}\leq \left\langle \widehat \Psi_\alpha, \mathbb H_0  \widehat \Psi_\alpha\right\rangle_\mathcal{F}+\frac{\lambda}{8}< E_\alpha+\frac{\lambda}{4}.
 \end{align}
 Note that $ \widehat \Phi_\alpha(x;Y_n)\neq 0$ implies that $\widetilde Y_n:=(y_1-x,\dots ,y_n-x)$ satisfies $K_2(\widetilde Y_n)\neq 0$, and therefore $\widetilde Y_n\notin \Omega^{(n)}_{s,\lambda}$, or equivalently $Y_n\notin \Omega^{(n)}_{s,\lambda}(x)$. Consequently, $ \widehat \Phi$ satisfies the assumptions of Lemma \ref{Lem:Bad_quantile_support_Bad_Energy}, which yields the desired contradiction to Eq.~(\ref{Eq:Energy_Estimate_for_Asymptotic_Concentration})
 \begin{align*}
       \left\langle \widehat \Phi_\alpha,  \mathbb H \widehat \Phi_\alpha \right\rangle_{ L^2\! \left(\mathbb{R}^3\right)\otimes \mathcal{F}}\geq E_\alpha+\frac{\lambda}{4},
 \end{align*}
 and therefore conclude the proof of Eq.~(\ref{Eq:Spatial_proof_A_Contradiction}).
\end{proof}

\begin{proof}[Proof of Theorem \ref{Th:Powers_of_Median}.]
    Since the case $m=0$ is trivial, we are going to assume in the following that $m\geq 1$. Defining $A^n$ as the set of all $Y_n\in \mathbb R^{3n}$ such that $\int \mathrm{d}\rho_{Y_n}\geq d$, we obtain  
    \begin{align}
    \label{Eq:Layer_Cake}
      & \ \ \left\langle \Psi_\alpha,\left|m_q(\chi* \rho_{Y_{n}})\right|^m \chi\! \left(\int \mathrm{d}\rho_{Y_{n-1}}\geq d\right)\Psi_\alpha\right\rangle_{\mathcal{F}}= \sum_{n=0}^\infty \int_{A^n}\left|m_q(\chi* \rho_{Y_{n}})\right|^m|\Psi_\alpha(Y_n)|^2\, \mathrm{d}Y_n\\
      \nonumber
      &  \ \  \ \  \ \  \ \  \ \  \ \  \ \  \ \ =\sum_{n=0}^\infty \int_{A^n}\int_0^{\left|m_q(\chi* \rho_{Y_{n}})\right|^m} \! \! |\Psi_\alpha(Y_n)|^2\, \mathrm{d}t\mathrm{d}Y_n\\
      \nonumber 
      & \ \  \ \  \ \  \ \  =\int_0^{\infty}\sum_{n=0}^\infty \int_{A^n}\chi\! \left(t^{\frac{1}{m}}\leq \left|m_q(\chi* \rho_{Y_{n}})\right|\right)|\Psi_\alpha(Y_n)|^2\, \mathrm{d}Y_n\mathrm{d}t
    \end{align}
   using a layer-cake representation. In the following let $R$ be large enough such that $\chi(y)=0$ in case $|x|>R$ and let us show that for any $Y_n\in A^n$, $|m_q(\chi* \rho_{Y_n})|\geq s$ implies 
    \begin{align}
    \label{Eq:Layer_Cake_Upper_Bound_Tool}
    \int \chi\! \left(|y|\geq \frac{s}{3}-R\right)\mathrm{d}\rho_{Y_n}\geq \left(\frac{1}{2}-q\right)d.
    \end{align}
    In order to prove Eq.~(\ref{Eq:Layer_Cake_Upper_Bound_Tool}), observe that in the case $|m_q(\chi* \rho_{Y_n})|\geq s$, at least one component 
    \begin{align*}
        m_q(\rho_{Y_n})=\left(m_q\left((\rho_{Y_n})_1\right),m_q\left((\rho_{Y_n})_2\right),m_q\left((\rho_{Y_n})_3\right)\right)
    \end{align*}
    satisfies $|m_q\left((\chi * \rho_{Y_n})_j\right)|\geq \frac{s}{3}$. W.l.o.g. assume $m_q\left((\chi* \rho_{Y_n})_j\right)\geq \frac{s}{3}$. Therefore the quantile satisfies $x^{\frac{1}{2}+q}\left((\chi* \rho_{Y_n})_j\right)\geq \frac{s}{3}$, which implies
    \begin{align*}
        & \int \chi\! \left(|y|\geq \frac{s}{3}-R\right)\mathrm{d}\rho_{Y_n} \geq  \int \chi\! \left(|y|\geq \frac{s}{3}\right)\mathrm{d}(\chi* \rho_{Y_n})\geq \int \mathrm{d}(\chi* \rho_{Y_n})-\int\chi\! \left(y< \frac{s}{3}\right)\mathrm{d}(\chi* \rho_{Y_n})\\
         & \ \ \ \geq \int \mathrm{d}(\chi* \rho_{Y_n})-\left(\frac{1}{2}+q\right)\int \mathrm{d}(\chi* \rho_{Y_n})=\left(\frac{1}{2}-q\right)\int \mathrm{d}\rho_{Y_n}\geq \left(\frac{1}{2}-q\right)d.
    \end{align*}
    Having Eq.~(\ref{Eq:Layer_Cake_Upper_Bound_Tool}) at hand for any $Y_n\in A^n$ with $|m_q(\rho_{Y_n})|\geq s$, yields together with the layer-cake representation Eq.~(\ref{Eq:Layer_Cake}), the definition of the probability $P_\alpha(s,\lambda)$ in Eq.~(\ref{Eq:Probability_t_lambda}) and $\lambda_*:=\left(\frac{1}{2}-q\right)d>0$ the upper bound
    \begin{align*}
        \left\langle \Psi_\alpha,\left|m_q(\chi *\rho_{Y_{n}})\right|^m \chi\! \left(\int \mathrm{d}\rho_{Y_{n-1}}\geq d\right) \! \Psi_\alpha\right\rangle_{\mathcal{F}}\leq \int_0^\infty \!  P_\alpha\! \left(\frac{1}{3}t^{\frac{1}{m}}-R,\lambda_*\right)\mathrm{d}t.
    \end{align*}
    Choosing $\gamma:=2m$ in Lemma \ref{Lem:Quantile_Mass_Estimate_by_IMS} gives us the estimate
    \begin{align*}
       P_\alpha\! \left(\frac{1}{3}t^{\frac{1}{m}}-R,\lambda_*\right)\leq Ct^{-2} 
    \end{align*}
for $t\geq t_0$ and suitable constants $C,t_0>0$. Together with the observation that $P_\alpha\! \left(s,\lambda_*\right)\leq 1$ for all $s\in \mathbb R$ we conclude
    \begin{align*}
        \int_0^\infty \!  P_\alpha\! \left(\frac{1}{3}t^{\frac{1}{m}}-R,\lambda_*\right)\mathrm{d}t\leq t_0+C\int_{t_0}^\infty t^{-2}\mathrm{d}t=t_0+\frac{C}{t_0}.
    \end{align*}
\end{proof}

\section{Asymptotic concentration of Probability}
\label{Sec:Asymptotic_concentration_of_Probability}
 In the following we want to control moments of $\mathcal{N}$ in the ground state $\Psi_\alpha$, see Lemma \ref{Lem:Moments_of_particle_number}, and demonstrate that the probability measure $|\Psi_\alpha(Y_n)|^2\mathrm{d}Y_n$ is mostly supported on sets of the form $\Omega_\lambda$ defined in Eq.~(\ref{Eq:Omega_def}), for $\lambda>0$, as $\alpha\rightarrow \infty$ goes to infinity, see Lemma \ref{Lem:Most_Of_Mass}. The asymptotic concentration of $|\Psi_\alpha(Y_n)|^2\mathrm{d}Y_n$, as stated in Lemma \ref{Lem:Most_Of_Mass}, is an important input for Lemma \ref{Lem:Rather_Explicit_Expression}, where we show essential bounds for the residual term $\mathcal{E}_0$, and Lemma \ref{Lem:E_4}, where we show corresponding bounds for the residual term $\mathcal{E}_4$. In preparation for Lemma \ref{Lem:Moments_of_particle_number}, we first establish an upper bound on the particle number $\mathcal{N}$ in terms of the fiber Hamilton operator $\mathbb H_0$ in the subsequent Lemma \ref{Lemma:LY_Fiber}. The proof is based on the strategy developed in \cite{LY}, see also \cite{BS1}.

 \begin{lem}
 \label{Lemma:LY_Fiber}
     There exists a constant $C>0$ such that for $K>0$
     \begin{align}
     \label{Eq:Particle_Number_Energy_bound}
 & \  \ \ \  \mathcal P^2+\mathcal{N}\leq   2 \mathbb H_0+C,\\
      \label{Eq:Particle_Number_Energy_bound_Full_Space}
   & \ \ -\Delta_x+\mathcal{N}\leq 2\mathbb H^K+C,\\
     \label{Eq:Comparision_full_energy_VS_cut-off_Energy}
       &    \mathbb H_0^K\leq \left(1+CK^{-\frac{1}{2}}\right)\mathbb H_0+K^{-\frac{1}{2}}C. 
     \end{align}
     Furthermore, for all $m\in \mathbb N$ there exists a constant $C$ such that
     \begin{align}
          \label{Eq:Particle_Number_a_Energy_bound}
       \pm \left(\mathcal{N}^{2m}a(v)+a(v)^*\mathcal{N}^{2m}\right)\leq C \mathcal{N}^m \left( \mathcal P^2+\mathcal{N} +1\right) \mathcal{N}^m .
     \end{align}
 \end{lem}
 \begin{proof}
     Recall the definition of $v^\Lambda$ and $w^\Lambda$ from Eq.~(\ref{Eq:Def_v_Lambda}) and Eq.~(\ref{Eq:Def_w_Lambda}), and that 
     \begin{align*}
        v=v^\Lambda+\frac{1}{i}\nabla\cdot w^\Lambda,
     \end{align*}
see Lemma \ref{Lem:Semiclassical_objects_properties}. Consequently, we can write the operator $a^*(v)$ as
\begin{align}
\label{Eq:a_Star_Decomposition}
    a^*(v)=a^*(v^\Lambda)+[\mathcal P,a^*(w^\Lambda)].
\end{align}
Regarding the first term on the right side of Eq.~(\ref{Eq:a_Star_Decomposition}) we use the fact that $\|v^\Lambda\|^2\lesssim \Lambda$, see Lemma \ref{Lem:Semiclassical_objects_properties}, in order to estimate for $\epsilon>0$
\begin{align}
\label{Eq:v_Lambda_Phi_Estimate}
   \pm \left(a^*(v^\Lambda)+a(v^\Lambda)\right)\lesssim \epsilon\! \left(\mathcal{N}+\alpha^{-2}\right)+\epsilon^{-1}\Lambda\leq \epsilon\! \left(\mathcal{N}+1\right)+\epsilon^{-1}\Lambda.
\end{align}
Furthermore, we use the fact that $\|w^\Lambda\|^2\lesssim \frac{1}{\Lambda}$, see Lemma \ref{Lem:Semiclassical_objects_properties}, in order to estimate for $\epsilon>0$
\begin{align}
\label{Eq:w_Lambda_Phi_Estimate}
    \pm\left([\mathcal P,a^*(w^\Lambda)]+[\mathcal P,a^*(w^\Lambda)]^*\right)\lesssim \epsilon \mathcal P^2+\epsilon^{-1}\Lambda^{-1}\left(\mathcal{N}+\alpha^{-2}\right)\leq \epsilon \mathcal P^2+\epsilon^{-1}\Lambda^{-1}\left(\mathcal{N}+1\right).
\end{align}
Combining what we have so far yields for a suitable constant $C>0$ the operator inequality
\begin{align*}
    \mathbb H_0 \! = \! \mathcal P^2 \! + \! \mathcal{N} \! - \! a^*(v) \! - \! a(v) \! \geq \!  (1 \! - \! C\epsilon)\mathcal P^2 \! + \! \left(1 \! - \! C\epsilon \! - \! C\epsilon^{-1}\Lambda^{-1}\right)\! \mathcal{N} \! - \! C\epsilon^{-1}\Lambda^{-1} \! - \! C\epsilon \! - \! C\epsilon^{-1}\Lambda.
\end{align*}
Choosing $\epsilon$ small enough such that $C\epsilon\leq \frac{1}{4}$ and $\Lambda$ large enough such that $C\epsilon^{-1}\Lambda^{-1}\leq \frac{1}{4}$, yields for a suitable constant $C>0$
\begin{align*}
    \mathbb H_0\geq \frac{1}{2}\mathcal P^2+\frac{1}{2}\mathcal{N}-C,
\end{align*}
which concludes the proof of Eq.~(\ref{Eq:Particle_Number_Energy_bound}). Eq.~(\ref{Eq:Particle_Number_Energy_bound_Full_Space}) can be verified analogously. Using Eq.~(\ref{Eq:a_Star_Decomposition}) we can furthermore write
\begin{align*}
    \mathbb H_0^K= \mathbb H+ [\mathcal P,a^*(w^K)]+[\mathcal P,a^*(w^K)]^*\leq  \mathbb H+C K^{-\frac{1}{2}}\left(\mathcal P^2+\mathcal{N}+1\right),
\end{align*}
where we have applied Eq.~(\ref{Eq:w_Lambda_Phi_Estimate}). Together with Eq.~(\ref{Eq:Particle_Number_Energy_bound}), this concludes the proof of Eq.~(\ref{Eq:Comparision_full_energy_VS_cut-off_Energy}). Regarding the proof of Eq.~(\ref{Eq:Particle_Number_a_Energy_bound}) we use again Eq.~(\ref{Eq:a_Star_Decomposition}), and proceed similarly as in Eq.~(\ref{Eq:v_Lambda_Phi_Estimate}) and Eq.~(\ref{Eq:w_Lambda_Phi_Estimate}) for the choice $\Lambda:=1$, yielding
\begin{align*}
   \pm \left(\mathcal{N}^{2m}a(v)+a(v)^*\mathcal{N}^{2m}\right)\leq C \mathcal{N}^m \left(\mathcal P^2 +\mathcal{N}+1\right) \mathcal{N}^m.
\end{align*}
 \end{proof}

With Lemma \ref{Lemma:LY_Fiber} at hand, we are in a position to verify the following Lemma \ref{Lem:Moments_of_particle_number}.
\begin{lem}
\label{Lem:Moments_of_particle_number}
   For all $\ell\in \mathbb N$, there exists a constant $C_\ell$, such that 
    \begin{align*}
        \braket{\Psi_\alpha,\mathcal{N}^\ell\Psi_\alpha}_{\mathcal{F}}\leq C_\ell, \ \ \ \ \braket{\Psi_\alpha,\mathcal{N}^\ell \mathcal P^2\Psi_\alpha}_{\mathcal{F}}\leq C_\ell.
    \end{align*}
\end{lem}
\begin{proof}
    Let us first define for $\ell\in \mathbb{N}$ the quantities
    \begin{align*}
        \Gamma_{\ell,\alpha}:=\left\langle \Psi_\alpha, \mathcal{N}^\ell\left(\mathcal P^2+\mathcal{N}\right)\mathcal{N}^\ell\Psi_\alpha\right\rangle_{\mathcal{F}}.
    \end{align*}
    Based on the operator inequality $\mathbb{H}_0-E_\alpha\geq \frac{1}{2}\mathcal P^2+\frac{1}{2}\mathcal{N}-C$, see Lemma \ref{Lemma:LY_Fiber} together with the fact that $-c\leq E_\alpha\leq 0$ for a suitable constant $c>0$, we have
    \begin{align}
    \nonumber
        \Gamma_{\ell,\alpha} & \lesssim \left\langle \Psi_\alpha,\mathcal{N}^\ell\left(\mathbb{H}_0 \! - \! E_\alpha\right)\mathcal{N}^\ell\Psi_\alpha\right\rangle_{\mathcal{F}} \! + \! C\left\langle \Psi_\alpha,\mathcal{N}^{2\ell}\Psi_\alpha\right\rangle_{\mathcal{F}}\\
            \nonumber
        &  = \frac{1}{2}\left\langle \Psi_\alpha,\left[\mathcal{N}^\ell,\left[\mathbb{H}_0 \! - \! E_\alpha,\mathcal{N}^\ell\right]\right]\Psi_\alpha\right\rangle_{\mathcal{F}} \! + \!C\left\langle \Psi_\alpha,\mathcal{N}^{2\ell}\Psi_\alpha\right\rangle_{\mathcal{F}} \\
        \label{Eq:Gamma_Bound}
      &  = -\frac{\ell^2}{2\alpha^4}\left\langle \Psi_\alpha,\left(\mathcal{N}^{2(\ell-1)}a(v)+a(v)^*\mathcal{N}^{2(\ell-1)}\right)\Psi_\alpha\right\rangle_{\mathcal{F}} \! + \!C\left\langle \Psi_\alpha,\mathcal{N}^{2\ell}\Psi_\alpha\right\rangle_{\mathcal{F}},
    \end{align}
    where we have used $\left(\mathbb{H}_0-E_\alpha\right)\Psi=0$ and the fact that
    \begin{align*}
       \left[\mathcal{N}^\ell,\left[\mathbb{H}_0 \! - \! E_\alpha,\mathcal{N}^\ell\right]\right]= \left[\mathcal{N}^\ell,\left[a(v)^*+a(v),\mathcal{N}^\ell\right]\right]=-\frac{\ell^2}{\alpha^4}\left(\mathcal{N}^{2(\ell-1)}a(v)+a(v)^*\mathcal{N}^{2(\ell-1)}\right).
    \end{align*}
    Combining Eq.~(\ref{Eq:Particle_Number_a_Energy_bound}) and Eq.~(\ref{Eq:Gamma_Bound}) therefore yields for a suitable constant $C>0$
    \begin{align*}
        \Gamma_{\ell,\alpha}\leq C\left\langle \Psi_\alpha,\mathcal{N}^{\ell-1}\left(\mathcal P^2+\mathcal{N}\right)\mathcal{N}^{\ell-1}\Psi_\alpha\right\rangle_{\mathcal{F}}+C\left\langle \Psi_\alpha,\mathcal{N}^{2(\ell-1)}\Psi_\alpha\right\rangle_{\mathcal{F}}+C\left\langle \Psi_\alpha,\mathcal{N}^{2\ell}\Psi_\alpha\right\rangle_{\mathcal{F}}.
    \end{align*}
    Clearly we have $C\mathcal{N}^{2(\ell-1)}+C\mathcal{N}^{2\ell}\leq \frac{1}{2}\mathcal{N}^{2\ell+1}+\widetilde C$ for a large enough constant $\widetilde C$, hence
    \begin{align*}
       \Gamma_{\ell,\alpha}\leq C\Gamma_{\ell-1,\alpha}+\frac{1}{2}\Gamma_{\ell,\alpha}+\widetilde C,
    \end{align*}
    or equivalently $ \Gamma_{\ell,\alpha}\leq 2C\Gamma_{\ell-1,\alpha}+2\widetilde C$. This concludes the proof by induction, since we have, using the operator inequality  $\mathbb{H}_0-E_\alpha\geq \frac{1}{2}\mathcal P^2+\frac{1}{2}\mathcal{N}-C$, that
    \begin{align*}
        \Gamma_{0,\alpha}=\left\langle \Psi_\alpha,\left(\mathcal P^2+\mathcal{N}\right)\Psi_\alpha\right\rangle_{\mathcal{F}}\leq 2\left\langle \Psi_\alpha, \left(\mathbb{H}-E_\alpha\right)\Psi_\alpha\right\rangle_{\mathcal{F}}+2C=2C.
    \end{align*}
\end{proof}

In the following we want to verify that the probability measure $|\Psi_\alpha|^2\mathrm{d}Y_n$ concentrates on sets of the form $\Omega_{\lambda}$ as $\alpha\rightarrow \infty$, where $\Omega_\lambda$ is defined in Eq.~(\ref{Eq:Omega_def}). For this purpose let us define for $\lambda>0$ the probability $P_\alpha(\lambda)$ of being outside of such a set 
\begin{align}
\label{Eq:Probability_P_lambda}
    P_\alpha(\lambda):= \sum_{n=0}^\infty \underset{\mathbb R^{3n}\setminus \Omega_{\lambda}^{(n)}}{\int}|\Psi_\alpha(Y_n)|^2\, \mathrm{d}Y_n.
\end{align}
In order to find good estimates on $P_\alpha(\lambda)$, we first need the subsequent auxiliary Lemma \ref{Lem:Localization_IMS} and Lemma \ref{Lem:Bad_Support_High_Energy}. Lemma \ref{Lem:Localization_IMS} provides control on the localization error with respect to the localization functions 
\begin{align*}
   F_{1,\lambda,u}(\rho) : & =\chi_{u^2}\! \left(\left(\int \mathrm{d}\rho-\|\varphi^\mathrm{Pek}\|^2\right)^2 + \left(\underset{|x-y|>R_*}{\int \int}\mathrm{d}\rho\mathrm{d}\rho-\delta_*\right)^2 \leq \lambda^2\right), \\
     F_{2,\lambda,u}(\rho) : & = \sqrt{1-F_{1,\lambda,u}(\rho)^2},
\end{align*}
as well as the localization functions
\begin{align*}
      G_{1,D}(\rho) : & =\chi_1\! \left(\int \mathrm{d}\rho\leq D\right), \\
     G_{2,D}(\rho) : & = \sqrt{1-G_{1,D}(\rho)^2}.
\end{align*}
\begin{lem}
    \label{Lem:Localization_IMS}
    For $\lambda> 0$ and $u>0$ there exists a constant $C>0$, such that 
    \begin{align}
    \label{Eq:Energy_Decomposition_in_IMS-Evironment}
        \sum_{j=1}^2\left\langle F_{j,\lambda,u}(\rho_{Y_n})\Psi_\alpha , \mathbb{H}_{0} F_{j,\lambda,u}(\rho_{Y_n})\Psi_\alpha\right\rangle_{\mathcal{F}}\leq E_\alpha + C\sqrt[4]{P_\alpha\! \left(\frac{\lambda-u}{2}\right)}\alpha^{-4}\leq E_\alpha + C\alpha^{-4}.
    \end{align}
    Furthermore, there exists a $C>0$, such that for any state $\Psi\in L^2\! \left(\mathbb{R}^3\right)\otimes \mathcal{F}$ and $K\geq 0$
    \begin{align}
     \label{Eq:Energy_K_Decomposition_in_IMS-Evironment}
         \sum_{j=1}^2\left\langle G_{j,D}(\rho_{Y_n})\Psi , \mathbb{H}^K G_{j,D}(\rho_{Y_n})\Psi\right\rangle_{L^2\! \left(\mathbb{R}^3\right)\otimes \mathcal{F}}\leq \left\langle \Psi , \mathbb{H}^K \Psi\right\rangle_{L^2\! \left(\mathbb{R}^3\right)\otimes \mathcal{F}} + CK^{\frac{1}{2}}\alpha^{-4}.
    \end{align}
\end{lem}
\begin{proof}
Using $\sum_{j=1}^2F_{j,\lambda,u}^2=1$ and $[F_{j,\lambda,u},\mathcal P]=[F_{j,\lambda,u},\mathcal{N}]=0$, we observe that
\begin{align*}
           \Big\| \mathcal P \Psi_\alpha\Big\|^2 & = \sum_{j=1}^2  \Big\| \mathcal P F_{j,\lambda,u}(\rho_{Y_n})\Psi_\alpha\Big\|^2, \\ 
     \left\langle    \Psi_\alpha, \mathcal{N}   \Psi_\alpha \right\rangle_{\mathcal{F}} & = \sum_{j=1}^2 \left\langle   F_{j,\lambda,u}(\rho_{Y_n})\Psi_\alpha, \mathcal{N}    F_{j,\lambda,u}(\rho_{Y_n}) \Psi_\alpha \right\rangle_{\mathcal{F}}.
\end{align*}
In order to verify Eq.~(\ref{Eq:Energy_Decomposition_in_IMS-Evironment}), it is therefore enough to verify
    \begin{align}
      \label{Eq:a_Decomposition_in_IMS-Evironment}
       \left|\sum_{j=1}^2\left\langle F_{j,\lambda,u}(\rho_{Y_n})\Psi_\alpha, a(v)^* F_{j,\lambda,u}(\rho_{Y_n})\Psi_\alpha \right\rangle_{\mathcal{F}}-\left\langle \Psi_\alpha,a(v)^* \Psi_\alpha \right\rangle_{\mathcal{F}}\right|\leq C\sqrt[4]{P_\alpha\! \left(\frac{\lambda-u}{2}\right)}\alpha^{-4} \leq C\alpha^{-4}.
    \end{align}
   The proof of Eq.~(\ref{Eq:a_Decomposition_in_IMS-Evironment}) is a consequence of the IMS identity for operators of the form $H(\rho_{Y_n})$, see \cite[Lemma 3.3]{BS1} and \cite[Theorem A.1]{LSol}, respectively \cite[Proposition 6.1]{LNSS}, for the particle number operator $\mathcal{N}=\int \mathrm{d}\rho_{Y_n}$ specifically, which we will carry out in detail in the following. We compute
    \begin{align*}
     & \ \ \ \  \ \ \ \  \sum_{j=1}^2 \left\langle \Psi_\alpha,F_{j,\lambda,u}(\rho_{Y_n})a(v)^* F_{j,\lambda,u}(\rho_{Y_n})\Psi_\alpha\right\rangle_{\mathcal{F}}-\left\langle \Psi_\alpha,a(v)^*\Psi_\alpha\right\rangle_{\mathcal{F}}\\
      & =\left\langle \Psi_\alpha,\left(\sum_{j=1}^2 F_{j,\lambda,u}(\rho_{Y_{n}})F_{j,\lambda,u}(\rho_{Y_{n-1}})-1\right)L(v)\Psi_\alpha\right\rangle_{L^2\! \left(\underset{n\in \mathbb N}{\bigcup}  \mathbb{R}^{3 n}\right)}=\left\langle \Psi_\alpha,F_*(Y_n)L(v)\Psi_\alpha\right\rangle_{L^2\! \left(\underset{n\in \mathbb N}{\bigcup}  \mathbb{R}^{3 n}\right)},
    \end{align*}
    with $F_*:\underset{n\in \mathbb N}{\bigcup}  \mathbb{R}^{3 n}\longrightarrow \mathbb R$ defined as
    \begin{align*}
        F_*(Y_n):=\sum_{j=1}^2 F_{j,\lambda,u}(\rho_{Y_{n}})F_{j,\lambda,u}(\rho_{Y_{n-1}})-1=-\frac{1}{2}\sum_{j=1}^2 \left[F_{j,\lambda,u}(\rho_{Y_{n}})-F_{j,\lambda,u}(\rho_{Y_{n-1}})\right]^2.
    \end{align*}
Using the decomposition of $a(v)^*$ for $\Lambda:=1$ in Eq.~(\ref{Eq:a_Star_Decomposition}) we can estimate
\begin{align}
\label{Eq:IMS_Error_Estimate}
  & \  \ \  \ \  \ \  \ \  \ \left\langle \Psi_\alpha,F_*(Y_n)L(v)\Psi_\alpha\right\rangle_{L^2\! \left(\underset{n\in \mathbb N}{\bigcup}  \mathbb{R}^{3 n}\right)}=\left\langle \Psi_\alpha,F_*(Y_n)L(v^1)\Psi_\alpha\right\rangle_{L^2\! \left(\underset{n\in \mathbb N}{\bigcup}  \mathbb{R}^{3 n}\right)}\\
  \nonumber
   &   \ \  \ +\left\langle \Psi_\alpha,F_*(Y_n)\mathcal P L(w^1)\Psi_\alpha\right\rangle_{L^2\! \left(\underset{n\in \mathbb N}{\bigcup}  \mathbb{R}^{3 n}\right)}-\left\langle \Psi_\alpha,F_*(Y_n) L(w^1)\mathcal P\Psi_\alpha\right\rangle_{L^2\! \left(\underset{n\in \mathbb N}{\bigcup}  \mathbb{R}^{3 n}\right)}\\
   \nonumber
   & \leq \left\|\Psi_\alpha\right\| \left\|F_*(Y_n)L(v^1)\Psi_\alpha\right\| +\left\|\mathcal P\Psi_\alpha\right\| \left\|F_*(Y_n)L(w^1)\Psi_\alpha\right\| +\left\|L(w^1)^*F_*(Y_n)\Psi_\alpha\right\| \left\|\mathcal P \Psi_\alpha\right\| ,
\end{align}
 where we have used that $\mathcal P$ and $F_*$ commute. Observe that $F_*(Y_n)\neq 0$ implies $Y_n\notin \Omega_{\frac{\lambda-u}{\sqrt{2}}}$ or $Y_{n-1}\notin \Omega_{\frac{\lambda-u}{\sqrt{2}}}$, and given that $\alpha$ is large enough such that $(1+(\|\varphi^\mathrm{Pek}\|^2+\lambda))\alpha^{-2}\leq u$ we obtain in both cases $Y_{n-1}\notin \Omega_{\frac{\lambda-u}{2}}$. Using $\|h\|\lesssim 1$ for $h\in \{v^1,w^1\}$ and $\|F_*\|_\infty\lesssim \alpha^{-4}$, see Eq.~(\ref{Eq:F_estimate_easy}) in Section \ref{Sec:Analysis_of_the_Error_Terms}, yields for $h\in \{v^1,w^1\}$
 \begin{align*}
     & \ \ \ \ \ \  \left\|F_*(Y_n)L(u)\Psi_\alpha\right\|^2=\sum_{n=1}^\infty \frac{n}{\alpha^2}\int_{\mathbb R^{3n}}F_*(Y_n)^2|h(y_n)|^2 |\Psi_\alpha(Y_{n-1})|^2\mathrm{d}Y_n\\
     & \lesssim \alpha^{-8}\sum_{n=1}^\infty \frac{n}{\alpha^2}\int_{\mathbb R^{3n}\setminus \Omega_{\frac{\lambda-u}{2}}^{(n)}}|h(y_n)|^2 |\Psi_\alpha(Y_{n-1})|^2\mathrm{d}Y_n= \alpha^{-8}\|h\|^2\braket{\Psi_\alpha,(\mathcal{N}+1)\mathds{1}_{\mathbb R^{3n}\setminus \Omega_{\frac{\lambda-u}{2}}^{(n)}}(Y_n)\Psi_\alpha}\\
     & \lesssim  \alpha^{-8}\left\|(\mathcal{N}+1)\Psi_\alpha\right\| \left\|\mathds{1}_{\mathbb R^{3n}\setminus \Omega_{\frac{\lambda-u}{2}}^{(n)}}(Y_n)\Psi_\alpha\right\|= \alpha^{-8}\left\|(\mathcal{N}+1)\Psi_\alpha\right\|\sqrt{P_\alpha\! \left(\frac{\lambda-u}{2}\right)}.
 \end{align*}
 In a similar fashion we have $ \left\|L(u)^*F_*(Y_n)\Psi_\alpha\right\|^2\lesssim \alpha^{-8}\left\|\mathcal{N}\Psi_\alpha\right\|\sqrt{P_\alpha\! \left(\frac{\lambda-u}{2}\right)}$. Together with the fact that $\|\mathcal P\Psi\|\lesssim 1$ and $\|\mathcal{N}\Psi\|\lesssim 1$, see Lemma \ref{Lem:Moments_of_particle_number}, we obtain by Eq.~(\ref{Eq:IMS_Error_Estimate})
 \begin{align*}
     \left|\left\langle \Psi_\alpha,F_*(Y_n)L(v)\Psi_\alpha\right\rangle_{L^2\! \left(\underset{n\in \mathbb N}{\bigcup}  \mathbb{R}^{3 n}\right)}\right|\lesssim \sqrt{\alpha^{-8}\sqrt{P_\alpha\! \left(\frac{\lambda-u}{2}\right)}} = \sqrt[4]{P_\alpha\! \left(\frac{\lambda-u}{2}\right)}\alpha^{-4}.
 \end{align*}
 Since $P_\alpha\! \left(\frac{\lambda-u}{2}\right)$ is a probability, i.e. $P_\alpha\! \left(\frac{\lambda-u}{2}\right)\leq 1$, this concludes the proof of Eq.~(\ref{Eq:a_Decomposition_in_IMS-Evironment}), and therefore the proof of Eq.~(\ref{Eq:Energy_Decomposition_in_IMS-Evironment}). Regarding Eq.~(\ref{Eq:Energy_K_Decomposition_in_IMS-Evironment}), we proceed in a similar fashion as we did for Eq.~(\ref{Eq:Energy_Decomposition_in_IMS-Evironment}) and estimate
 \begin{align}
 \nonumber
    & \sum_{j=1}^2\left\langle G_{j,D}(\rho_{Y_n})\Psi , \mathbb{H}^K G_{j,D}(\rho_{Y_n})\Psi\right\rangle_{L^2\! \left(\mathbb{R}^3\right)\otimes \mathcal{F}} - \left\langle \Psi , \mathbb{H}^K \Psi\right\rangle_{L^2\! \left(\mathbb{R}^3\right)\otimes \mathcal{F}}\\
    \label{Eq:G_IMS_Estimate}
    & \ \  \leq \left\langle \Psi , G_*(\rho_{Y_n})L(v^K) \Psi\right\rangle_{L^2\! \left(\mathbb{R}^3\right)\otimes \mathcal{F}}\leq \left\|G_*(\rho_{Y_n}) \Psi\right\| \left\|L(v^K) \mathds{1}_{A}\Psi\right\| ,
 \end{align}
 with $A:=\underset{n\in \mathbb N}{\bigcup} \{Y_n\in \mathbb R^{3n}:\int \mathrm{d}\rho_{Y_n}\leq D+1\}$ and  $G_*:\underset{n\in \mathbb N}{\bigcup}  \mathbb{R}^{3 n}\longrightarrow \mathbb R$ defined as
    \begin{align*}
        G_*(Y_n):=\frac{1}{2}\sum_{j=1}^2 \left[G_{j,D}(\rho_{Y_{n}})-G_{j,D}(\rho_{Y_{n-1}})\right]^2.
    \end{align*}
    Note that we have used the fact that $G_*(Y_n)\neq 0$ implies $Y_{n-1}\in A$ in Eq.~(\ref{Eq:G_IMS_Estimate}). Furthermore, we have $\|G_*\|_\infty\lesssim \alpha^4$, and therefore we can estimate the right hand side of Eq.~(\ref{Eq:G_IMS_Estimate}) by
    \begin{align*}
        \left\|G_*(\rho_{Y_n}) \Psi\right\| \left\|L(v^K) \mathds{1}_{A}\Psi\right\| & \lesssim \alpha^{-4}\left\|L(v^K) \mathds{1}_{A}\Psi\right\|=\alpha^{-4}\|v^K\|\sqrt{\left\langle \mathds{1}_{A}\Psi, (\mathcal{N}+1)\mathds{1}_{A}\Psi\right\rangle}\\
        & \leq \alpha^{-4}\|v^K\|\sqrt{D+2}\lesssim K^{\frac{1}{2}}\alpha^{-4}.
    \end{align*}
\end{proof}

The following Lemma \ref{Lem:Bad_Support_High_Energy} shows that any state $\Psi\in L^2\! \left(\mathbb{R}^3\right)\otimes \mathcal{F}$ supported outside of $\Omega_\lambda$, with $\lambda>0$, cannot have an energy close to the ground state energy $E_\alpha$. While the proof is essentially contained in \cite{BS1}, we will work it out in detail here for the sake of completeness.
\begin{lem}
\label{Lem:Bad_Support_High_Energy}
    Let $\Psi\in L^2\! \left(\mathbb{R}^3\right)\otimes \mathcal{F}$ satisfy that $\Psi(x;Y_n)\neq 0$ implies $Y_n\notin \Omega_\lambda$ and define $K:=\alpha$. Then 
    \begin{align}
    \label{Eq:Bad_Support_High_Energy}
        \left\langle \Psi,\mathbb H^K \Psi\right\rangle_{L^2\! \left(\mathbb{R}^3\right)\otimes \mathcal{F}}>E_\alpha+\alpha^{-\frac{4}{29}}.
    \end{align}
\end{lem}
\begin{proof}
    We are going to verify Eq.~(\ref{Eq:Bad_Support_High_Energy}) by contradiction. Assume that Eq.~(\ref{Eq:Bad_Support_High_Energy}) is violated
    \begin{align*}
        \left\langle \Psi,\mathbb H^K \Psi\right\rangle_{L^2\! \left(\mathbb{R}^3\right)\otimes \mathcal{F}}\leq E_\alpha+\alpha^{-\frac{4}{29}}.
    \end{align*}
Then we obtain by Eq.~(\ref{Eq:Energy_K_Decomposition_in_IMS-Evironment}) for $\alpha$ large enough
    \begin{align}
    \label{Eq:Applied_G_IMS}
         \sum_{j=1}^2\left\langle G_{j,D}(\rho_{Y_n})\Psi , \mathbb{H}^K G_{j,D}(\rho_{Y_n})\Psi\right\rangle_{L^2\! \left(\mathbb{R}^3\right)\otimes \mathcal{F}}\leq E_\alpha+\alpha^{-\frac{4}{29}} + CK^{\frac{1}{2}}\alpha^{-4}\leq E_\alpha+2\alpha^{-\frac{4}{29}},
    \end{align}
    with the functions $G_{j,D}$ defined above Lemma \ref{Lem:Localization_IMS}. By Eq.~(\ref{Eq:Particle_Number_Energy_bound_Full_Space}) we immediately obtain for a suitable constant $C>0$
    \begin{align}
    \nonumber
      &  \left\langle G_{2,D}(\rho_{Y_n})\Psi , \mathbb{H}^K G_{2,D}(\rho_{Y_n})\Psi\right\rangle_{L^2\! \left(\mathbb{R}^3\right)\otimes \mathcal{F}}\geq \frac{1}{2} \left\langle G_{2,D}(\rho_{Y_n})\Psi , (\mathcal{N}-C) G_{2,D}(\rho_{Y_n})\Psi\right\rangle_{L^2\! \left(\mathbb{R}^3\right)\otimes \mathcal{F}}\\
          \label{Eq:Applied_MAss_Lower_Bound}
       & \ \ \ \ \geq \frac{1}{2}(D-C)\left\|G_{2,D}(\rho_{Y_n})\right\|^2=\frac{1}{2}(D-C)\left(1-\left\|G_{1,D}(\rho_{Y_n})\right\|^2\right),
    \end{align}
    where we have used the fact that $G_{2,D}(\rho_{Y_n})\neq 0$ implies that $\int \mathrm{d}\rho_{Y_n}\geq D$, i.e. we have the operator inequality $\mathcal{N}\geq D$ on the support of $G_{2,D}$. Let us define the state
    \begin{align*}
        \widetilde \Psi:=\frac{1}{\left\|G_{1,D}(\rho_{Y_n})\right\|}G_{1,D}(\rho_{Y_n}).
    \end{align*}
    Choosing $D$ large enough such that $\frac{1}{2}(D-C)\geq E_\alpha+2\alpha^{-\frac{4}{29}}$, we obtain by Eq.~(\ref{Eq:Applied_G_IMS}) and Eq.~(\ref{Eq:Applied_MAss_Lower_Bound}) the estimate
    \begin{align*}
        \left\langle \widetilde \Psi , \mathbb{H}^K \widetilde \Psi\right\rangle_{L^2\! \left(\mathbb{R}^3\right)\otimes \mathcal{F}}\leq E_\alpha+2\alpha^{-\frac{4}{29}}\leq e^\mathrm{Pek}+2\alpha^{-\frac{4}{29}}.
    \end{align*}
   Furthermore, we have $\mathcal{N}\leq D+1$ on the support of $\widetilde \Psi$, and therefore $\widetilde \Psi$ satisfies all assumptions of \cite[Theorem 3.2]{BS1}, which tells us that for all $m\in \mathbb N$ there exists a constant $C_m>0$ and a (Borel) probability measure $\mu$ on $\mathbb R^3$ such that 
   \begin{align}
   \label{Eq:Quantum_De_Finetti}
        \left|\left\langle \widetilde \Psi , J(\rho_{Y_n}) \widetilde \Psi\right\rangle_{L^2\! \left(\mathbb{R}^3\right)\otimes \mathcal{F}}-\int_{\mathbb R^3}J(\rho^\mathrm{Pek}_x)\mathrm{d}\mu(x)\right|\leq C_m \|j\|_{\infty} \alpha^{-\frac{2}{29}},
   \end{align}
  for any $J$ of the form $J(\rho)=\int \dots \int j(z_1,\dots ,z_m)\mathrm{d}\rho(z_1)\dots \mathrm{d}\rho(z_n)$ with bounded and measurable $j:\mathbb R^{3m}\longrightarrow \mathbb R$, where $\mathrm{d}\rho^\mathrm{Pek}_x(z):=\left|\varphi^\mathrm{Pek}(z-x)\right|^2\mathrm{d}z$. Consider the concrete choice
  \begin{align*}
      J(\rho):=\left(\int \mathrm{d}\rho-\int \mathrm{d}\rho^\mathrm{Pek}\right)^2+\left(\underset{|x-y|\leq R_*}{\int \int}\mathrm{d}\rho\mathrm{d}\rho-\delta_*\right)^2,
  \end{align*}
  which is a finite sum of admissible functions $J'$ in the sense described below Eq.~(\ref{Eq:Quantum_De_Finetti}). Furthermore, we have $J(\rho^\mathrm{Pek}_x)=J(\rho^\mathrm{Pek}_0)=0$ for all $x\in \mathbb R^3$, and therefore we obtain 
  \begin{align*}
      \left\langle \widetilde \Psi , J(\rho_{Y_n}) \widetilde \Psi\right\rangle_{L^2\! \left(\mathbb{R}^3\right)\otimes \mathcal{F}}\leq C\alpha^{-\frac{2}{29}}
  \end{align*}
  by Eq.~(\ref{Eq:Quantum_De_Finetti}). However, it follows immediately from the definition of $\Omega_\lambda$ that $J(\rho_{Y_n})\geq \lambda^2$ for $Y_n\notin \Omega_\lambda$. Together with the assumption that $\Psi(x;Y_n)\neq 0$ implies $Y_n\notin \Omega_\lambda$, this yields for $\alpha$ large enough such that $C\alpha^{-\frac{2}{29}}<\lambda^2$ the desired contradiction 
  \begin{align*}
      \lambda^2\leq \left\langle \widetilde \Psi , J(\rho_{Y_n}) \widetilde \Psi\right\rangle_{L^2\! \left(\mathbb{R}^3\right)\otimes \mathcal{F}}\leq C\alpha^{-\frac{2}{29}}<\lambda^2.
  \end{align*}
\end{proof}

The following Lemma \ref{Lem:Most_Of_Mass} shows that the measure $|\Psi_\alpha|^2\mathrm{d}Y_n$ is supported on the set $\Omega_{\lambda}$, up to a probability of the size $\alpha^{-5}$. While the estimate by $\alpha^{-5}$ is sufficient for our purpose, it is possible to improve this to $\alpha^{-m}$ for any $m\in \mathbb N$ with some additional effort.
\begin{lem}
\label{Lem:Most_Of_Mass}
Let $\lambda>0$. Then there exists a constant $C>0$ such that
    \begin{align*}
        P_\alpha(\lambda)\leq C \alpha^{-5}.
    \end{align*}
\end{lem}
\begin{proof}
Trivially, we have $P_\alpha(\lambda)\leq 1$ for all $\lambda\geq 0$, and therefore we are done by iteration, once we can verify that for all $\lambda> 0$ and $u>0$ there exists a constant $C$ such that
 \begin{align}
 \label{Eq:Probability_Iterative_Scheme}
     P_\alpha(2(\lambda+u))\leq C \sqrt[4]{P_\alpha(\lambda)}\alpha^{-4\left(1-\frac{1}{29}\right)} .
 \end{align}
 To be more precise, Eq.~(\ref{Eq:Probability_Iterative_Scheme}) implies for all $r<\frac{1}{1-\frac{1}{4}}4\left(1-\frac{1}{29}\right)$, such as $r:=5$, and $\lambda>0$ that there exists a constant $C_{r,\lambda}$ such that
 \begin{align*}
    P_\alpha(\lambda)\leq C_{r,\lambda} \alpha^{-r} .
 \end{align*}
  Let us define the functions $F_j:=F_{j,2\lambda+u,u}$, introduced above Lemma \ref{Lem:Localization_IMS}. By Lemma \ref{Lem:Localization_IMS} we obtain for a suitable constant $C>0$
\begin{align*}
    \sum_{j=1}^2\left\langle F_j(\rho_{Y_n})\Psi_\alpha , \mathbb{H}_{0} F_j(\rho_{Y_n})\Psi_\alpha\right\rangle_{\mathcal{F}}\leq  \left\langle \Psi_\alpha , \mathbb{H}_{0} \Psi_\alpha\right\rangle_{\mathcal{F}}+ C\sqrt[4]{P_\alpha(\lambda)}\alpha^{-4}=E_\alpha + C\sqrt[4]{P_\alpha(\lambda)}\alpha^{-4}.
\end{align*}
In combination with the operator inequality $\mathbb H_0\geq E_\alpha$, and the fact that 
\begin{align*}
   P_\alpha(2(\lambda+u))\leq \|F_2(\rho_{Y_n})\Psi_\alpha\|^2=1-\|F_1(\rho_{Y_n})\Psi_\alpha\|^2, 
\end{align*}
we obtain for the state 
\begin{align*}
  \widehat{\Psi}_\alpha:=\frac{1}{\|F_2(\rho_{Y_n})\Psi_\alpha\|}F_2(\rho_{Y_n})\Psi_\alpha  
\end{align*}
the estimate
\begin{align}
\label{Eq:Probability_Estimate_by_IMS_Asymptotic}
 P_\alpha(2(\lambda+u))\left\langle \widehat{\Psi}_\alpha,(\mathbb H_0-E_\alpha)\widehat{\Psi}_\alpha\right\rangle_{\mathcal{F}}\leq \|F_2(\rho_{Y_n})\Psi_\alpha\|^2 \left\langle \widehat{\Psi}_\alpha,(\mathbb H_0-E_\alpha)\widehat{\Psi}_\alpha\right\rangle_{\mathcal{F}}\leq C\sqrt[4]{P_\alpha(\lambda)}\alpha^{-4}.
\end{align}
In the following we want to show by contradiction that 
\begin{align}
\label{Eq:Contradiction_Argument_Asymptotic_Mass}
    \left\langle \widehat{\Psi}_\alpha,\mathbb H_0\widehat{\Psi}_\alpha\right\rangle_{\mathcal{F}}\geq E_\alpha+\frac{1}{3}\alpha^{-\frac{4}{29}},
\end{align}
i.e. we assume
\begin{align}
    \label{Eq:Asymptotic_proof_A_Contradiction_Assumption}
      \left\langle \widehat{\Psi}_\alpha,\mathbb H_0\widehat{\Psi}_\alpha\right\rangle_{\mathcal{F}}< E_\alpha+\frac{1}{3}\alpha^{-\frac{4}{29}}.
\end{align}
By Eq.~(\ref{Eq:Comparision_full_energy_VS_cut-off_Energy}) we therefore obtain for $K:=\alpha$ and $\alpha$ large enough
\begin{align*}
   \left\langle \widehat{\Psi}_\alpha,\mathbb H^K_0\widehat{\Psi}_\alpha\right\rangle_{\mathcal{F}}\leq E_\alpha+\frac{2}{3}\alpha^{-\frac{4}{29}}.
\end{align*}
Let us furthermore define the auxiliary state $\widehat \Phi_\alpha\in L^2\! \left(\mathbb{R}^3\right)\otimes \mathcal{F}\subseteq  L^2\! \left(\mathbb{R}^3\times \underset{n\in \mathbb N}{\bigcup}  \mathbb{R}^{3 n}\right)$ as
 \begin{align*}
     \widehat \Phi_\alpha(x;Y_n):=\tau_\epsilon(x)\widehat \Psi_\alpha(y_1-x,\dots ,y_n-x),
 \end{align*}
 where $\tau_\epsilon(x):=\epsilon^{\frac{3}{2}}\tau\! \left(\epsilon x\right)$ and $\tau$ is a $[0,1]$-valued smooth function with compact support and $\int \tau^2\mathrm{d}x=1$. Note that there exists a constant $C'>0$ such that for $\epsilon>0$
 \begin{align*}
    & \ \ \ \  \left\langle \widehat \Phi_\alpha,  (-\Delta_x) \widehat \Phi_\alpha \right\rangle_{ L^2\! \left(\mathbb{R}^3\right)\otimes \mathcal{F}}= \left\langle \tau_\epsilon\otimes \widehat \Psi_\alpha,  \left(\frac{1}{i}\nabla_x -\mathcal P\right)^2 \tau_\epsilon\otimes \widehat \Psi_\alpha \right\rangle_{ L^2\! \left(\mathbb{R}^3\right)\otimes \mathcal{F}}\\
     &  \leq (1+\epsilon)  \left\langle \widehat \Psi_\alpha, \mathcal P^2  \widehat \Psi_\alpha \right\rangle_{ \mathcal{F}}+\left(1+\epsilon^{-1}\right)\! \left\langle \tau_\epsilon, (-\Delta_x) \tau_\epsilon \right\rangle_{ L^2\! \left(\mathbb{R}^3\right)} \leq \left\langle \widehat \Psi_\alpha,  \mathcal P^2  \widehat \Psi_\alpha \right\rangle_{ \mathcal{F}}+C'\! \left(\epsilon+\epsilon^2\right),
 \end{align*}
   where we have used the assumption in Eq.~(\ref{Eq:Asymptotic_proof_A_Contradiction_Assumption}) and the fact that $\mathcal{P}^2\leq 2\mathbb H_0 +C$ for a suitable constant $C>0$, see Eq.~(\ref{Eq:Particle_Number_Energy_bound}). Moreover, we observe that
     \begin{align*}
     &  \ \ \ \  \left\langle \widehat \Phi_\alpha,  \Big\{\mathcal{N}-a^*(v^\Lambda_x)-a(v^\Lambda_x)\Big\} \widehat \Phi_\alpha \right\rangle_{ L^2\! \left(\mathbb{R}^3\right)\otimes \mathcal{F}}=\left\langle \widehat \Psi_\alpha,  \Big\{\mathcal{N}-a^*(v^\Lambda)-a(v^\Lambda)\Big\} \widehat \Psi_\alpha \right\rangle_{ \mathcal{F}}.
 \end{align*}
Therefore we obtain for $\epsilon:=\alpha^{-1}$ and $\alpha$ large enough
 \begin{align}
 \label{Eq:Energy_Estimate_for_Asymptotic_Concentration}
      \left\langle \widehat \Phi_\alpha,  \mathbb H^K \widehat \Phi_\alpha \right\rangle_{ L^2\! \left(\mathbb{R}^3\right)\otimes \mathcal{F}}\leq \left\langle \widehat \Psi_\alpha, \mathbb H^K_0  \widehat \Psi_\alpha\right\rangle_\mathcal{F}+C'\alpha^{-1}+C'\alpha^{-2}\leq E_\alpha+\alpha^{-\frac{4}{29}}.
 \end{align}
 It follows however from Lemma \ref{Lem:Bad_Support_High_Energy}, that any state satisfying Eq.~(\ref{Eq:Energy_Estimate_for_Asymptotic_Concentration}) cannot be supported outside of $\Omega_{\lambda}$ for $\lambda>0$ and $\alpha$ large enough. Since the state $ \widehat \Phi_\alpha$ is supported on the set $\left(\underset{n\in \mathbb N}{\bigcup}  \mathbb{R}^{3 n}\right)\setminus \Omega_{\sqrt{2}\lambda}$, this is the desired contradiction to Eq.~(\ref{Eq:Asymptotic_proof_A_Contradiction_Assumption}). Combining Eq.~(\ref{Eq:Probability_Estimate_by_IMS_Asymptotic}) and Eq.~(\ref{Eq:Contradiction_Argument_Asymptotic_Mass}) concludes the proof of Eq.~(\ref{Eq:Probability_Iterative_Scheme}).
 \end{proof}

\section{Analysis of the Error Terms $\mathcal{E}$}
\label{Sec:Analysis_of_the_Error_Terms}
It is the goal of this Section, to show that the residual terms $\mathcal{E}_1,\dots, \mathcal{E}_4$ appearing in the analysis of the quantum energy
\begin{align*}
    \left\langle \Psi_{\alpha,p},\mathbb H_p  \Psi_{\alpha,p}  \right\rangle_\mathcal{F}
\end{align*}
in Subsection \ref{Sec:Isolating_the_essential_Contribution} are indeed small compared to $\alpha^{-4}|p|^2$, where we rely heavily on the results of the previous Sections \ref{Sec:Spatial_concentration_of_Probability} and \ref{Sec:Asymptotic_concentration_of_Probability}. As a first step let us verify the following auxiliary Lemma \ref{Lema:F_eta_estimate}.

\begin{lem}
\label{Lema:F_eta_estimate}
    Let $f:\mathbb R^3\longrightarrow \mathbb R^3$ be as in the definition of $\mathcal{B}$ in Eq.~(\ref{Eq:Boost_Operator}). Then there exist constants $C,d>0$, such that for all $Y_n\in \mathbb{R}^{3n}$ and $n\in \mathbb{N}$ 
        \begin{align}
      \label{Eq:F_estimate_easy}
      &   \left|F(\rho_{Y_{n}})-F(\rho_{Y_{n-1}})\right|\leq \frac{C}{\alpha^2}, \ \ \ \ \  \left|\sqrt{1-F(\rho_{Y_{n}})^2}-\sqrt{1-F(\rho_{Y_{n-1}})^2}\right|\leq \frac{C}{\alpha^2},
    \end{align}  
 and in the case that $y_i\neq y_j$ for all $i,j\in \{1,\dots ,n\}$ 
   \begin{align}
         \nonumber 
      &  \ \ \ \ \ \ \   \ \ \ \ \ \ \  \ \ \  \left|G(\rho_{Y_{n}})-G(\rho_{Y_{n-1}})\right|\\
            \label{Eq:G_estimate_easy}
      & \leq \frac{C\left(\|f\|_\infty + \|\nabla f\|_\infty\right)}{\alpha^2}\left(1+\min\{|m_q(\chi* \rho_{Y_{n-1}})|,|m_q(\chi* \rho_{Y_{n}})|\}\right)\chi\! \left(\int \mathrm{d}\rho_{Y_{n-1}}\geq d\right),
\end{align}
and
\begin{align}
\nonumber
      & \ \ \ \ \ \ \   \ \ \ \ \  \left|G_i(\rho_{Y_{n}})  \mathcal{B}_j(\rho_{Y_{n}}) \! - \!  G_i(\rho_{Y_{n-1}})  \mathcal{B}_j(\rho_{Y_{n-1}})\right|\\
             \label{Eq:G_B_estimate_easy}
      &   \leq \frac{C\left(\|f\|_\infty \! + \!  \|\nabla f\|_\infty\right)}{\alpha^2}  \! \left(1+\min\{|m_q(\chi* \rho_{Y_{n-1}})|,|m_q(\chi* \rho_{Y_{n}})|\}\right)^2\chi\! \left(\int \mathrm{d}\rho_{Y_{n-1}}\geq d\right).
      \end{align}
    Furthermore, $0\leq \varphi_{\eta}(Y_n)\leq \frac{C}{\eta^4 \alpha^4}$ for $\eta\leq q$ and $Y_n$ satisfying $G(\rho_{Y_n})\neq G(\rho_{Y_{n-1}})$ or $F(\rho_{Y_n})\neq 0$, and $y_i\neq y_j$.
\end{lem}
\begin{proof}
    Let $\epsilon$ be small enough such that $\|\varphi^\mathrm{Pek}\|^2+\sigma+\kappa+\epsilon<\frac{(\|\varphi^\mathrm{Pek}\|^2-\sigma-\kappa-\epsilon)^2}{2}$ and let $\alpha$ be large enough such that $(1+2(\|\varphi^\mathrm{Pek}\|^2+\sigma+\kappa))\alpha^{-2}\leq \epsilon$. Note that $F(\rho_{Y_{n}})-F(\rho_{Y_{n-1}})\neq 0$ implies that $\rho_{Y_{n-1}}$ or $\rho_{Y_{n}}$ is an element of $\Omega_{\sigma+\kappa}$. In both cases we obtain 
    \begin{align}
    \label{Eq:Reasonable_Mass_Conditions}
       d:=\|\varphi^\mathrm{Pek}\|^2-\sigma-\kappa-\epsilon \leq \int \mathrm{d}\rho_{Y_{n-1}}\leq \|\varphi^\mathrm{Pek}\|^2+\sigma+\kappa+\epsilon,\ \ \ \underset{|x-y|>R}{\int \int}\mathrm{d}\nu\mathrm{d}\nu\leq \delta_*+\sigma+\kappa+\epsilon.
    \end{align}
    Since the function $\chi_{\kappa^2}(\cdot \leq \delta^2)$, as well as $\sqrt{1-\chi_{\kappa^2}(\cdot \leq \delta^2)^2}$, in the definition of $F$ in Eq.~(\ref{Eq:Def_F_cut_off_functional}) is smooth, Eq.~(\ref{Eq:F_estimate_easy}) follows immediately from the observation that
    \begin{align*}
      &  \ \ \ \ \  \ \ \ \ \ \int\mathrm{d}\rho_{Y_n}-\int\mathrm{d}\rho_{Y_{n-1}}=\frac{1}{\alpha^2}, \\
       &  0\leq \underset{|x-y|>R}{\int \int}\mathrm{d}\rho_{Y_{n}}\mathrm{d}\rho_{Y_{n}}-\underset{|x-y|>R}{\int \int}\mathrm{d}\rho_{Y_{n-1}}\mathrm{d}\rho_{Y_{n-1}}=\frac{2}{\alpha^2}\int_{|x-y_n|}\mathrm{d}\rho_{Y_{n-1}}\leq \frac{2(\|\varphi^\mathrm{Pek}\|^2+\sigma+\kappa+\epsilon)}{\alpha^2}. 
    \end{align*}
    Regarding the proof of Eq.~(\ref{Eq:G_estimate_easy}), we first note that Eq.~(\ref{Eq:Reasonable_Mass_Conditions}) holds as well in case 
    \begin{align*}
    G(\rho_{Y_{n}})-G(\rho_{Y_{n-1}})\neq 0,    
    \end{align*}
    and therefore we have for all such $Y_n$ that 
    \begin{align}
    \label{Eq:Sharp_Median_Bound}
        \left|m_{q_*}\! \left(\rho_{Y_{n}}\right)-m_{q_*}\! \left(\rho_{Y_{n-1}}\right)\right|\leq \frac{C}{\alpha^2 q_*}
    \end{align}
   by \cite[Lemma 3.10]{BS1} for a suitable constant $C>0$ and $q_*\leq q$, and furthermore we have by \cite[Lemma 3.9]{BS1} that the quantiles in the definition of $H^{\chi}_{q,\rho_{Y_{n-1}}}$ in Eq.~(\ref{Eq:H_Function}) satisfy 
    \begin{align}
        \label{Eq:Quantile_dif_est.}
        \left|x^{+}_{j,q}(\rho_{Y_{n-1}})-x^{-}_{j,q}(\rho_{Y_{n-1}})\right|\leq 2R_*.
    \end{align}
Especially we see that $(\rho_{Y_{n-1}},\rho_{Y_{n-1}},\rho_{Y_{n}})$ is an admissible triple in the sense of Lemma \ref{Lem:Eq:Total_Variation_Upper_Bound}, and therefore we obtain for a suitable constant $C>0$
    \begin{align}
        \label{Eq:M_q_diff_Formula}
        \left|m_q\! \left(\chi* \rho_{Y_{n}}\right) \! - \! m_q\! \left(\chi* \rho_{Y_{n-1}}\right)\right|\leq \|g * (\rho_{Y_{n}} \! - \! \rho_{Y_{n-1}})\|_{\mathrm{TV}}^2  \! +  \! \left|\int H^{\chi}_{q,\rho_{Y_{n-1}}}\mathrm{d}(\rho_{Y_{n}} \! - \! \rho_{Y_{n-1}})\right|\leq \frac{C}{\alpha^2},
    \end{align}
    where we have used $\|g * (\rho_{Y_{n}}-\rho_{Y_{n-1}})\|_{\mathrm{TV}}\leq \|\rho_{Y_{n}}-\rho_{Y_{n-1}}\|_{\mathrm{TV}}=\frac{1}{\alpha^2}$ and $\|H^{\chi}_{q,\rho_{Y_{n-1}}}\|_\infty \lesssim 1$, which is a consequence of Eq.~(\ref{Eq:Quantile_dif_est.}). We compute
    \begin{align}
    \label{Eq:G_diff_Formula}
     & \     G(\rho_{Y_{n}})-G(\rho_{Y_{n-1}})=\left[F(\rho_{Y_{n}})-F(\rho_{Y_{n-1}})\right]\! \mathcal{B}(\rho_{Y_n})+F(\rho_{Y_{n-1}})\! \left[m_q\! \left(\chi* \rho_{Y_{n}}\right)-m_q\! \left(\chi* \rho_{Y_{n-1}}\right)\right]\\
     \nonumber 
     & \ \  \ \  \ \  \ \  \ \  \ \  \ \  \ \  +   \alpha^{-2}F(\rho_{Y_{n-1}})f\! \left(y_n \! - \! m_q\! \left(\chi* \rho_{Y_{n}}\right)\right) \\
     & \ \  \ \  \ \  \ \  \ \  \ \  \ \  \ \   + \! F(\rho_{Y_{n-1}})\int \left[f(y \! - \! m_q\! \left(\chi* \rho_{Y_{n}}\right)) \! - \! f(y \! - \! m_q\! \left(\chi* \rho_{Y_{n-1}}\right))\right]\mathrm{d}\rho_{Y_{n-1}}.
    \end{align}
Using Eq.~(\ref{Eq:F_estimate_easy}) and Eq.~(\ref{Eq:M_q_diff_Formula}), as well as
\begin{align*}
 & \ \ \ \left| \mathcal{B}(\rho_{Y_n})\right| \leq \left|m_q\! \left(\chi* \rho_{Y_{n}}\right)\right|+\|f\|_\infty \|\rho_{Y_{n}}\|_{\mathrm{TV}}  ,\\
 & \|\rho_{Y_{n}}\|_{\mathrm{TV}}=\|\rho_{Y_{n-1}}\|_{\mathrm{TV}}+\alpha^{-2}\leq \|\varphi^\mathrm{Pek}\|^2+\sigma+\kappa+\epsilon
\end{align*}
 and
\begin{align*}
&   \int \left[f(y \! - \! m_q\! \left(\chi* \rho_{Y_{n}}\right)) \! - \! f(y \! - \! m_q\! \left(\chi* \rho_{Y_{n-1}}\right))\right]\mathrm{d}\rho_{Y_{n-1}}\\
  & \ \ \leq \|\nabla f\|_\infty \|\rho_{Y_{n-1}}\|_{\mathrm{TV}}\left|m_q\! \left(\chi* \rho_{Y_{n}}\right) \! - \! m_q\! \left(\chi* \rho_{Y_{n-1}}\right)\right|,  
\end{align*}
concludes the proof of Eq.~(\ref{Eq:G_estimate_easy}). Eq.~(\ref{Eq:G_B_estimate_easy}) can be verified analogously. Regarding the estimate on $\varphi_\eta$, let us use the fact that $\int_{\mathbb R^3}\tau_\eta(y)\mathrm{d}y=1$ in order to compute
\begin{align*}
    \varphi_\eta(Y_n) & =\frac{1}{2} \int_{\mathbb{R}^3}\left[\tau_{\eta}\! \left(m_\eta(\rho_{Y_{n-1}})+x\right)-\tau_{\eta}\! \left(m_\eta(\rho_{Y_{n}})+x\right)\right]^2\mathrm{d}x\\
    & \leq \frac{\|\nabla \tau_\eta\|^2}{2}\left|m_{\eta}\! \left(\rho_{Y_{n}}\right)-m_{\eta}\! \left(\rho_{Y_{n-1}}\right)\right|^2\leq C\eta^{-4}\alpha^{-4},
\end{align*}
for a suitable $C>0$, where we have used Eq.~(\ref{Eq:Sharp_Median_Bound}) and the fact that $\|\nabla \tau_\eta\|=\eta^{-1}\|\nabla \tau\|$.
\end{proof}

With Lemma \ref{Lema:F_eta_estimate} at hand we are going to verify that the expressions $\mathcal{E}_1,\dots ,\mathcal{E}_4$ can be considered as being small in magnitude. It will be the content of the following Lemma \ref{Lem:E_1} to establish estimates on 
\begin{align*}
    \mathcal{E}_1 =\frac{1}{i} \Big\langle \Psi_{\alpha},\xi_p(Y_n)L(\nabla\cdot w^\Lambda)\Psi_{\alpha}\Big\rangle_{L^2\! \left(\underset{n\in \mathbb N}{\bigcup}\mathbb R^{3n}\right)},
\end{align*}
introduced in Eq.~(\ref{Eq:Def_E_1}). 

\begin{lem}
\label{Lem:E_1}
   Let $f:\mathbb R^3\longrightarrow \mathbb R^3$ be as in the definition of $\mathcal{B}$ in Eq.~(\ref{Eq:Boost_Operator}). Then there exists a constant $C>0$ such that for $\Lambda>0$
    \begin{align*}
       |\mathcal{E}_1|\leq C\left(1+\|f\|_\infty \! + \!  \|\nabla f\|_\infty\right)^2\frac{ |p|^2}{\sqrt{\Lambda} \alpha^4}.
    \end{align*}
\end{lem}
\begin{proof}
   We start with the simple identity
    \begin{align}
    \nonumber
      \mathcal{E}_1 & =\frac{1}{i} \Big\langle \Psi_{\alpha},\xi_p(Y_n)L(\nabla\cdot w^\Lambda)\Psi_{\alpha}\Big\rangle_{L^2\! \left(\underset{n\in \mathbb N}{\bigcup}\mathbb R^{3n}\right)}=\Big\langle \Psi_{\alpha},\xi_p(Y_n)\big[\mathcal{P},L(w^\Lambda)\big]\Psi_{\alpha}\Big\rangle_{L^2\! \left(\underset{n\in \mathbb N}{\bigcup}\mathbb R^{3n}\right)}\\
       \label{Eq:Split_T_2}
       &=\left\langle \mathcal P\Psi_\alpha, \xi_p L(w^\Lambda)\Psi_\alpha\right\rangle_{L^2\! \left(\underset{n\in \mathbb N}{\bigcup}  \mathbb{R}^{3 n}\right)}-\left\langle \Psi_\alpha, \xi_p  L(w^\Lambda) \mathcal P\Psi_\alpha\right\rangle_{L^2\! \left(\underset{n\in \mathbb N}{\bigcup}  \mathbb{R}^{3 n}\right)}\\
       \nonumber
       & \ \ \ \ \ \ \ \  \ \ \ \ \ +\left\langle \Psi_\alpha, [\xi_p, \mathcal P] L(w^\Lambda)\Psi_\alpha\right\rangle_{L^2\! \left(\underset{n\in \mathbb N}{\bigcup}  \mathbb{R}^{3 n}\right)}.
    \end{align}
    Regarding the first term on the right hand side of Eq.~(\ref{Eq:Split_T_2}) we use that $\|\mathcal P \Psi_\alpha\|\lesssim 1$, see Lemma \ref{Lem:Moments_of_particle_number}, and estimate
    \begin{align*}
       & \left|\left\langle \mathcal P\Psi_\alpha, \xi_p L(w^\Lambda)\Psi_\alpha\right\rangle_{L^2\! \left(\underset{n\in \mathbb N}{\bigcup}  \mathbb{R}^{3 n}\right)}\right|^2 \leq \|\mathcal P \Psi_\alpha\|^2 \, \|\xi_p L(w^\Lambda)\Psi_\alpha\|^2\lesssim \|\xi_p L(w^\Lambda)\Psi_\alpha\|^2\\
       & \ \ =\sum_{n=1}^{D_\alpha}\frac{n}{\alpha^2}\int_{\mathbb{R}^{3n}}\left|w^\Lambda(y_{n})\xi_p(Y_{n})\Psi_\alpha(Y_{n-1})\right|^2 \mathrm{d}Y_{n},
    \end{align*}
    where we only consider $n\leq D_\alpha:=(\|\varphi^\mathrm{Pek}\|^2+\sigma+\kappa)\alpha^2 +1$, since
    \begin{align*}
      \mathrm{supp}\! \left(\xi_p\right)\cap \mathbb{R}^{3n}\subseteq   \Omega^{(n)}_{\sigma+\kappa}\cup\Omega^{(n-1)}_{\sigma+\kappa}\times \mathbb R^3 =\emptyset  
    \end{align*}
 in case $n>D_\alpha$. By Eq.~(\ref{Eq:G_estimate_easy}) we furthermore have 
    \begin{align}
    \nonumber
       & \ \ \  \  \left|\xi_p(Y_{n})\right| \! = \! \Big|p \! \cdot \! \left[G(\rho_{Y_{n}})-G(\rho_{Y_{n-1}})\right]\Big|^2 \! \left|\int_0^1 \!  \int_0^t \!  e^{is p\cdot \left[G(\rho_{Y_{n}})-G(\rho_{Y_{n-1}})\right]}\mathrm{d}s \mathrm{d}t\right| \\
         \label{Eq:xi_Control}
       &\lesssim  \left(\|f\|_\infty \! + \!  \|\nabla f\|_\infty\right)^2 \frac{|p|^2}{\alpha^4} \Big(1 \! + \! \min\{|m_q(\chi*\rho_{Y_{n-1}})|,|m_q(\chi*\rho_{Y_{n}})|\}\Big)^2 \chi\! \left(\int \mathrm{d}\rho_{Y_{n-1}}\geq d\right) ,
    \end{align}  
    for any $Y_{n}$ satisfying $y_i\neq y_j$. Consequently
    \begin{align*}
      & \ \ \ \ \ \ \ \ \ \ \ \ \   \frac{1}{\left(\|f\|_\infty \! + \!  \|\nabla f\|_\infty\right)^2}\left|\left\langle \mathcal P\Psi_\alpha, \xi_p L(w^\Lambda)\Psi_\alpha\right\rangle_{L^2\! \left(\underset{n\in \mathbb N}{\bigcup}  \mathbb{R}^{3 n}\right)}\right| \\
      & \lesssim \left(\sum_{n=1}^\infty \int_{\mathbb{R}^{3n}}|w^\Lambda(y_{n})|^2 \Big(1+|m_q(\chi*\rho_{Y_{n-1}})|\Big)^4 \chi\! \left(\int \mathrm{d}\rho_{Y_{n-1}}\geq d\right)\frac{|p|^4}{\alpha^8}\left|\Psi_\alpha(Y_{n-1})\right|^2 \mathrm{d}Y_{n}\right)^{\frac{1}{2}}\\
      & \leq \frac{|p|^2}{\alpha^4}\|w^\Lambda\| \left\langle \Psi_\alpha, \left(1+|m_q(\chi*\rho_{Y_{n}})|\right)^4\chi\! \left(\int \mathrm{d}\rho_{Y_{n}}\geq d\right)\Psi_\alpha \right\rangle_{L^2\! \left(\underset{n\in \mathbb N}{\bigcup}  \mathbb{R}^{3 n}\right)}^{\frac{1}{2}}\lesssim \frac{|p|^2}{\sqrt{\Lambda}\alpha^4 },
    \end{align*}
    where we have used that $\|w^\Lambda\|\lesssim \frac{1}{\sqrt{\Lambda}}$, see Lemma \ref{Lem:Semiclassical_objects_properties}. In order to estimate the second term on the right hand side of Eq.~(\ref{Eq:Split_T_2}) we use again $\|\mathcal P \Psi_\alpha\|\lesssim 1$ and compute
     \begin{align*}
        & \ \ \ \ \  \left|\left\langle \Psi_\alpha, \xi_p  L(w^\Lambda) \mathcal P\Psi_\alpha\right\rangle_{L^2\! \left(\underset{n\in \mathbb N}{\bigcup}  \mathbb{R}^{3 n}\right)}\right|^2 \lesssim \|L(w^\Lambda)^* \xi_p^* \Psi_\alpha\|^2\\
         & =\sum_{n=1}^{D_\alpha} \frac{n}{\alpha}\int_{\mathbb{R}^{3(n-1)}}\left|\int_{\mathbb{R}^3}w^\Lambda(y_{n})\xi^*_p(Y_{n})\Psi_\alpha(Y_{n})\mathrm{d}y_{n}\right|^2\mathrm{d}Y_{n-1} \\
         & \lesssim \|w^\Lambda\|^2 \int_{\mathbb{R}^{3(n-1)}}\int_{\mathbb{R}^3}\left|\xi^*_p(Y_{n})\Psi_\alpha(Y_{n})\right|^2\mathrm{d}y_{n}\mathrm{d}Y_{n-1}\\
         & \lesssim \left(\|f\|_\infty \! + \!  \|\nabla f\|_\infty\right)^2 \frac{|p|^4}{\Lambda \alpha^8}\left\langle \Psi_\alpha, \left(1+|m_q(\chi*\rho_{Y_{n}})|\right)^4\chi\! \left(\int \mathrm{d}\rho_{Y_{n}}\geq d\right)\Psi_\alpha \right\rangle_{L^2\! \left(\underset{n\in \mathbb N}{\bigcup}  \mathbb{R}^{3 n}\right)}\\
         & \lesssim \left(\|f\|_\infty \! + \!  \|\nabla f\|_\infty\right)^2 \frac{|p|^4}{\Lambda \alpha^8},
     \end{align*}
     where we have used Eq.~(\ref{Eq:xi_Control}) and Theorem \ref{Th:Powers_of_Median} again in the last estimate. Regarding the third term on the right hand side of Eq.~(\ref{Eq:Split_T_2}), we proceed in a similar fashion as we did for the first term and estimate
     \begin{align*}
         \left|\left\langle \Psi_\alpha, [\xi_p, \mathcal P] L(w^\Lambda)\Psi_\alpha\right\rangle_{L^2\! \left(\underset{n\in \mathbb N}{\bigcup}  \mathbb{R}^{3 n}\right)}\right|\lesssim \|w^\Lambda\| \left\|[\xi_p, \mathcal P]^* \Psi_\alpha\right\|\lesssim \frac{1}{\sqrt{\Lambda}}\left\|[\xi_p, \mathcal P]^*\Psi_\alpha\right\|.
     \end{align*}
     An explicit computation reveals that $[\xi_p, \mathcal P]$ is a multiplication operator given by
     \begin{align*}
      & \ \   \ \   \ \   \ \   [\xi_p, \mathcal P](Y_n)=i\sum_{j=1}^n\partial_{y_j}\xi_p(Y_n)\\
      &=-ip\cdot \left[G(\rho_{Y_{n-1}})-G(\rho_{Y_{n}})\right] \int_0^1 \!  e^{it p\cdot \left[G(\rho_{Y_{n-1}})-G(\rho_{Y_{n}})\right]}\mathrm{d}t \left(F(\rho_{Y_{n-1}})-F(\rho_{Y_{n}})\right)p.
     \end{align*}
     Using Eq.~(\ref{Eq:xi_Control}) and $\left|F(\rho_{Y_{n-1}})-F(\rho_{Y_{n}})\right|\lesssim \frac{1}{\alpha^2}$, see Eq.~(\ref{Eq:F_estimate_easy}), and Theorem \ref{Th:Powers_of_Median} yields
     \begin{align*}
         \left\|[\xi_p, \mathcal P]^* \Psi_\alpha\right\| & \lesssim \left(\|f\|_\infty \! + \!  \|\nabla f\|_\infty\right)\frac{|p|^2}{\alpha^4} \left\|\left(1+|m_q(\chi*\rho_{Y_{n}})|\right)\chi\! \left(\int \mathrm{d}\rho_{Y_{n}}\geq d\right)\Psi_\alpha\right\|\\
         & \lesssim \left(\|f\|_\infty \! + \!  \|\nabla f\|_\infty\right)\frac{|p|^2}{\alpha^4}.
     \end{align*}
\end{proof}

In the subsequent Lemma \ref{Lem:E_2,3} we will analyse the residual terms
\begin{align*}
    \mathcal{E}_2 & = \Big\langle \!  \Psi_{\alpha},\left(\xi_p(Y_n)+\frac{1}{2}\left[p \! \cdot \! G(\rho_{Y_{n}})-p \! \cdot \! G(\rho_{Y_{n-1}})\right]^2\right)L(v^\Lambda)\Psi_{\alpha} \! \Big\rangle_{L^2\! \left(\underset{n\in \mathbb N}{\bigcup}\mathbb R^{3n}\right)} ,\\
     \mathcal{E}_3 &  =-\frac{1}{2}\Big\langle \!  \Psi_{\alpha}, \! \varphi_\eta(Y_n)\left[p \! \cdot \! G(\rho_{Y_{n}}) \! - \! p \! \cdot \! G(\rho_{Y_{n-1}})\right]^2  \! L(v^\Lambda)\Psi_{\alpha} \! \Big\rangle_{L^2\! \left(\underset{n\in \mathbb N}{\bigcup}\mathbb R^{3n}\right)}   
\end{align*}
introduced in Eq.~(\ref{Eq:Def_E_2}) and Eq.~(\ref{Eq:Def_E_3}) respectively.
\begin{lem}
\label{Lem:E_2,3}
     Let $f:\mathbb R^3\longrightarrow \mathbb R^3$ be as in the definition of $\mathcal{B}$ in Eq.~(\ref{Eq:Boost_Operator}) and $\eta=\alpha^{-\beta}$ as in Eq.~(\ref{Eq:Definition_eta_from_beta}). Then there exists a constant $C>0$ such that for $\Lambda>0$ and $\eta\leq q$ 
    \begin{align*}
       |\mathcal{E}_2|\leq \frac{C\left(\|f\|_\infty \! + \!  \|\nabla f\|_\infty\right)^3 \sqrt{\Lambda} |p|^3}{ \alpha^6}, \ \ \ \ \   |\mathcal{E}_3|\leq \frac{C\left(\|f\|_\infty \! + \!  \|\nabla f\|_\infty\right)^2  \sqrt{\Lambda} |p|^2}{\eta^4 \alpha^8}.
    \end{align*}
\end{lem}
\begin{proof}
    Regarding the term $\mathcal{E}_2$, note that we have
    \begin{align*}
        \mathrm{supp}\! \left(\xi_p(Y_n)+\frac{1}{2}\left[p \! \cdot \! G(\rho_{Y_{n}})-p \! \cdot \! G(\rho_{Y_{n-1}})\right]^2\right)\cap \mathbb R^{3n}\subseteq  \Omega^{(n)}_{\sigma+\kappa}\cup\Omega^{(n-1)}_{\sigma+\kappa}\times \mathbb R^3 =:A^n ,
    \end{align*}
    and therefore
    \begin{align}
    \label{Eq:T_3_Factorization}
        |\mathcal{E}_2|\leq 2\|L(v^\Lambda)\mathds{1}_{A}\Psi_\alpha\|\, \| \! \left(\xi_p+\frac{1}{2}\left[p \! \cdot \! G(\rho_{Y_{n}})-p \! \cdot \! G(\rho_{Y_{n-1}})\right]^2\right) \! \Psi_\alpha\|.
    \end{align}
    Using the fact that $A^n=\emptyset$ in case $n>D_\alpha:=(\|\varphi^\mathrm{Pek}\|^2+\sigma+\kappa)\alpha^2 +1$, we can easily estimate the first factor on the right hand side of Eq.~(\ref{Eq:T_3_Factorization}) by
    \begin{align*}
        \|L(v^\Lambda)\mathds{1}_{A}\Psi_\alpha\|^2 \! \leq \!   \sum_{n=1}^{D_\alpha} \!  \! \frac{n}{\alpha^2}\int_{\mathbb R^{3n}} \! \! \!   \! \left|v^\Lambda(y_{n})\Psi_\alpha(Y_{n-1})\right|^2 \mathrm{d}Y_{n} \! \lesssim  \|v^\Lambda\|^2 \! \lesssim  \! \Lambda,
    \end{align*}
    where we have used that $\|v^\Lambda\|^2\lesssim \Lambda$, see Lemma \ref{Lem:Semiclassical_objects_properties}. In order to estimate the second factor on the right hand side of Eq.~(\ref{Eq:T_3_Factorization}), note that we have by Eq.~(\ref{Eq:G_estimate_easy})
\begin{align*}
  & \ \ \ \ \   \ \    \left|\xi_p(Y_n)+\frac{1}{2}\left[p \! \cdot \! G(\rho_{Y_{n}})-p \! \cdot \! G(\rho_{Y_{n-1}})\right]^2\right|\\
    & =\big|p \! \cdot \! \left[G(\rho_{Y_{n}})-G(\rho_{Y_{n-1}})\right]\big|^3 \left|\int_0^1 \int_0^t\int_0^s e^{irp\cdot \left[G(\rho_{Y_{n}})-G(\rho_{Y_{n-1}})\right]}\mathrm{d}r\mathrm{d}s\mathrm{d}t\right|\\
    & \leq |p|^3 \left|G(\rho_{Y_{n}})-G(\rho_{Y_{n-1}})\right|^3 \lesssim \left(\|f\|_\infty \! + \!  \|\nabla f\|_\infty\right)^3\Big(1+|m_q(\chi*\rho_{Y_{n}})|\Big)^3\frac{|p|^3}{\alpha^6}\chi\! \left(\int \mathrm{d}\rho_{Y_{n}}\geq d\right)
\end{align*}
    for all $Y_n$ with $y_i\neq y_j$. Consequently, we obtain using Theorem \ref{Th:Powers_of_Median}
    \begin{align*}
      & \ \ \ \ \ \ \ \    |\mathcal{E}_2|\lesssim \sqrt{\Lambda} \| \! \left(\xi_p+\frac{1}{2}\left[p \! \cdot \! G(\rho_{Y_{n}})-p \! \cdot \! G(\rho_{Y_{n-1}})\right]^2\right) \! \Psi_\alpha\|\\
      & \lesssim \left(\|f\|_\infty \! + \!  \|\nabla f\|_\infty\right)^3\frac{\sqrt{\Lambda}|p|^3}{\alpha^6}\left\|\Big(1+|m_q(\chi*\rho_{Y_{n}})|\Big)^3\Psi_\alpha\right\| \lesssim \left(\|f\|_\infty \! + \!  \|\nabla f\|_\infty\right)^3\frac{\sqrt{\Lambda}|p|^3}{\alpha^6}.
    \end{align*}
Utilizing the fact that $\|\mathds{1}_{\widetilde A} \varphi_{\eta}\|_\infty\lesssim \frac{1}{\eta^4 \alpha^4}$, where $\widetilde A$ is the set of all $Y_n$ such that $G(\rho_{Y_n})\neq G(\rho_{Y_{n-1}})$, see Lemma \ref{Lema:F_eta_estimate}, we proceed similarly by 
\begin{align*}
    |\mathcal{E}_3|\lesssim \|v^\Lambda\| \|\mathds{1}_{\widetilde A} \varphi_{\eta}\|_\infty \|\left[p \! \cdot \! G(\rho_{Y_{n}})-p \! \cdot \! G(\rho_{Y_{n-1}})\right]^2\Psi_\alpha\|\lesssim \left(\|f\|_\infty \! + \!  \|\nabla f\|_\infty\right)^2 \sqrt{\Lambda}\frac{1}{\eta^4 \alpha^4}\frac{|p|^2}{\alpha^4}.
\end{align*}
\end{proof}

Finally, we are going to derive estimates on the residual term
\begin{align*}
& \ \ \ \  \ \ \ \  \ \ \ \  \mathcal{E}_4=-\frac{1}{2} \! \int_{\mathbb R^3}  \! \! \Big\langle \!  \tau_\eta(m_{\eta}(\rho_{Y_{n}}) \! + \! x)\Psi_{\alpha}, \! \left(F(\rho_{Y_{n}}) \! - \! F(\rho_{Y_{n-1}})\right)\\
\nonumber
 &\times \left(p \! \cdot \! G(\rho_{Y_{n}})p \! \cdot \! \mathcal{B}(\rho_{Y_{n}}) \! - \! p \! \cdot \! G(\rho_{Y_{n-1}})p \! \cdot \! \mathcal{B}(\rho_{Y_{n-1}})\right)  \! L(v^\Lambda) \tau_\eta(m_{\eta}(\rho_{Y_n}) \! + \! x) \, \Psi_{\alpha} \! \Big\rangle_{L^2\! \left(\underset{n\in \mathbb N}{\bigcup}\mathbb R^{3n}\right)}\! \mathrm{d}x
\end{align*}
introduced in Eq.~(\ref{Eq:Def_E_4}).

\begin{lem}
\label{Lem:E_4}
   Let $f:\mathbb R^3\longrightarrow \mathbb R^3$ be as in the definition of $\mathcal{B}$ in Eq.~(\ref{Eq:Boost_Operator}). Then there exists a constant $C>0$ such that
    \begin{align*}
      |\mathcal{E}_4|\leq \frac{C \left(\|f\|_\infty \! + \! \|\nabla f\|_\infty\right)\sqrt{\Lambda} |p|^2}{ \alpha^{6}}.
    \end{align*}
\end{lem}
\begin{proof}
Let us introduce the function $c(Y_n)$ as
\begin{align*}
     c(Y_n): =\frac{1}{2}\int_{\mathbb R^3}\tau_{\eta}\! \left(m_\eta(\rho_{Y_{n}})+x\right)\tau_{\eta}\! \left(m_\eta(\rho_{Y_{n-1}})+x\right)\mathrm{d}x,
\end{align*}
which allows us to write
\begin{align*}
& \ \ \  \ \ \ \ \ \  \mathcal{E}_4 =-\Big\langle \Psi_\alpha, c(Y_{n}) \! \left(F(\rho_{Y_{n}}) \! - \! F(\rho_{Y_{n-1}})\right) \\
    &  \times  \left(p \! \cdot \! G(\rho_{Y_{n}})p \! \cdot \! \mathcal{B}(\rho_{Y_{n}}) \! - \! p \! \cdot \! G(\rho_{Y_{n-1}})p \! \cdot \! \mathcal{B}(\rho_{Y_{n-1}})\right) \! L(v^\Lambda)  \Psi_\alpha\Big\rangle_{L^2\! \left(\underset{n\in \mathbb N}{\bigcup}  \mathbb{R}^{3 n}\right)}  \! .
\end{align*}
    Note that $|c(Y_{n})|\leq \frac{1}{2}$ for all $Y_n\in \mathbb R^{3n}$, $n\in \mathbb{N}$, and by Eq.~(\ref{Eq:G_B_estimate_easy}) we obtain for $y_i\neq y_j$
    \begin{align*}
       & \ \ \ \ \ \ \left|\left(p \! \cdot \! G(\rho_{Y_{n}})p \! \cdot \! \mathcal{B}(\rho_{Y_{n}}) \! - \! p \! \cdot \! G(\rho_{Y_{n-1}})p \! \cdot \! \mathcal{B}(\rho_{Y_{n-1}})\right)\right|\\
       &\leq \frac{C\left(\|f\|_\infty \! + \!  \|\nabla f\|_\infty\right)}{\alpha^2} |p|^2 \! \left(1 +|m_q(\chi*\rho_{Y_{n-1}})|\right)^2\chi\! \left(\int \mathrm{d}\rho_{Y_{n-1}}\geq d\right)=:\frac{C |p|^2}{\alpha^2}\zeta(Y_{n-1}).
    \end{align*}
 Consequently,
    \begin{align}
    \nonumber
      |\mathcal{E}_4| \!   & \lesssim  \! \frac{|p|^2}{\alpha^2}   \! \! \sum_{n=1}^{D_\alpha}   \!  \! \frac{\sqrt{n}}{\alpha}   \!  \! \underset{\mathbb{R}^{3n}}{\int}    \!    \left|F(\rho_{Y_{n}}) \! - \! F(\rho_{Y_{n-1}})\right|  \zeta(Y_{n-1})  |v^\Lambda(y_{n})| \left|\overline{\Psi_\alpha(Y_{n-1})}\Psi_\alpha(Y_{n})\right|\mathrm{d}Y_{n}\\
        \nonumber
        & \lesssim \!\frac{|p|^2}{\alpha^2} \! \sum_{n=1}^{D_\alpha}      \! \underset{\mathbb{R}^{3n}}{\int}  \!  \!  \!     \left|F(\rho_{Y_{n}}) \! - \! F(\rho_{Y_{n-1}})\right|  \zeta(Y_{n-1})  |v^\Lambda(y_{n})|\left|\overline{\Psi_\alpha(Y_{n-1})}\Psi_\alpha(Y_{n})\right|\mathrm{d}Y_{n}\\
        \label{Eq:T_5_Bound}
        & \leq \frac{|p|^2}{\alpha^2}  \|v^\Lambda\| \, \|  \zeta(Y_n)   \Psi_\alpha\|\, \|  \! \left(F(\rho_{Y_{n}}) \! - \! F(\rho_{Y_{n-1}})\right) \! \Psi_\alpha\|,
    \end{align}
      where we only consider $n\leq D_\alpha:=(\|\varphi^\mathrm{Pek}\|^2+\sigma+\kappa)\alpha^2 +1$, since
    \begin{align*}
      \mathrm{supp}\! \left(F(\rho_{Y_{n}}) \! - \! F(\rho_{Y_{n-1}})\right)\cap \mathbb{R}^{3n} & \subseteq \Omega^{(n)}_{\sigma+\kappa}\cup\Omega^{(n-1)}_{\sigma+\kappa}\times \mathbb R^3 =\emptyset  
    \end{align*}
 in case $n>D_\alpha$. Making use of Theorem \ref{Th:Powers_of_Median} we furthermore obtain   
 \begin{align*}
   \|  \zeta(Y_n)   \Psi_\alpha\|^2  & =\left(\|f\|_\infty \! + \!  \|\nabla f\|_\infty\right)^2\sum_{n=0}^\infty  \! \int_{\mathbb{R}^{3n}}  \left(1 \! + \! |m_q(\chi*\rho_{Y_{n}})|\right)^4 \chi\! \left(\int \mathrm{d}\rho_{Y_{n-1}}\geq d\right)|\Psi_\alpha(Y_n)|^2\mathrm{d}Y_n\\
      & \leq \left(\|f\|_\infty \! + \!  \|\nabla f\|_\infty\right)^2\sum_{n=0}^\infty  \! \int_{\mathbb{R}^{3n}}  \left(1 \! + \! |m_q(\chi*\rho_{Y_{n}})|\right)^4 \chi\! \left(\int \mathrm{d}\rho_{Y_{n}}\geq d\right)|\Psi_\alpha(Y_n)|^2\mathrm{d}Y_n\\
      &  \lesssim \left(\|f\|_\infty \! + \!  \|\nabla f\|_\infty\right)^2.
 \end{align*}
 Finally, we note that $\left|F(\rho_{Y_n})-F(\rho_{Y_{n-1}})\right|\lesssim \frac{1}{\alpha^2}$ by Eq.~(\ref{Eq:F_estimate_easy}) and $F(\rho_{Y_n})-F(\rho_{Y_{n-1}})\neq 0$ implies $Y_n\notin \Omega_{\frac{\sigma-\kappa}{\sqrt{2}}}$ or $Y_{n-1}\notin \Omega_{\frac{\sigma-\kappa}{\sqrt{2}}}$. In both cases we have $Y_{n}\notin \Omega_{\frac{\sigma-\kappa}{\sqrt{2}}-\alpha^{-2}}\subseteq \Omega_{\frac{\sigma-\kappa}{\sqrt{2}}-\epsilon}$ for $\epsilon>0$ and $\alpha$ large enough and therefore  
 \begin{align*}
    \|  \! \left(F(\rho_{Y_{n}}) \! - \! F(\rho_{Y_{n-1}})\right) \! \Psi_\alpha\|^2\lesssim \alpha^{-4}\sum_{n=0}^\infty \underset{\mathbb{R}^{3n}\setminus \Omega_{\frac{\sigma-\kappa}{\sqrt{2}}-\epsilon}^{(n)}} {\int} \! \!  \!  |\Psi_\alpha(Y_n)|^2\mathrm{d}Y_{n} \! \lesssim  \! \alpha^{-8},
 \end{align*}
 for $\epsilon$ small enough, see Lemma \ref{Lem:Most_Of_Mass}. This concludes the proof by Eq.~(\ref{Eq:T_5_Bound}), together with the fact that $\|v^\Lambda\|\lesssim \sqrt{\Lambda}$, see Lemma \ref{Lem:Semiclassical_objects_properties}. 
\end{proof}

\section{Bose-Einstein Condensation}
\label{Sec:Bose-Einstein_Condensation}

In the following we want to establish that the state $\Phi_\alpha$, defined in Eq.~(\ref{Eq:Def_Phi_State})
\begin{align*}
     \Phi_\alpha(x;Y_n)=\mu_\alpha^{-1}\tau_{\eta}\! \left(m_\eta (\rho_{Y_{n}})\right)F(\rho_{Y_{n}})\Psi_\alpha(y_1-x_1,\dots ,y_n-x),
\end{align*}
satisfies Bose-Einstein condensation, and as a consequence we are going to show that the typical empirical measures $\rho_{Y_n}$ of the phonon positions $Y_n=(y_1,\dots ,y_n)$ are close to the Pekar measure $\rho^\mathrm{Pek}=\rho_{\varphi^\mathrm{Pek}}$ with respect to the mollified total variation. As a first step, we are going to show in the subsequent Lemma \ref{Lem:Energy_Estimates_tilde} that the state $\Phi_\alpha$ defined in Eq.~(\ref{Eq:Def_Phi_State}) is a low energy state of the cut-off Hamiltonian $\mathbb H^K$, or equivalently that $\overline{\Psi}_\alpha=\mathcal{T}^* \Phi_\alpha$, defined in Eq.~(\ref{Eq:Localized_State_LLP})
\begin{align*}
    \overline{\Psi}_\alpha(x;Y_n)=\mu_\alpha^{-1}\tau_{\eta}\! \left(m_\eta (\rho_{Y_{n}})+x\right)F(\rho_{Y_{n}})\Psi_\alpha(Y_n),
\end{align*}
is a low energy state of $\mathbb{H}^{K}_{\frac{1}{i}\nabla_x}$.

\begin{lem}
\label{Lem:Energy_Estimates_tilde}
  Let $K:=\alpha^{s}$ with $\frac{8}{29}< s \leq 1$. Then we have
    \begin{align}
    \label{Eq:Energy_Localized_State_Pre}
       \left\langle   \Phi_\alpha,\mathbb{H}^{K} \Phi_\alpha \right\rangle_{L^2\! \left(\mathbb{R}^3\times \underset{n\in \mathbb N}{\bigcup}  \mathbb{R}^{3 n}\right)}\leq E_\alpha+\alpha^{-\frac{4}{29}}.
    \end{align}
\end{lem}
\begin{proof}
As a first step we note that $\overline{\Psi}_\alpha=\mathcal{T}^* \Phi_\alpha$ and $\mathcal{T}^*\mathbb{H}^{K}\mathcal{T}=\mathbb{H}^{K}_{\frac{1}{i}\nabla}$, and therefore
\begin{align}
\label{Eq:LLP_in_BEC_Section}
    \left\langle   \Phi_\alpha,\mathbb{H}^{K} \Phi_\alpha \right\rangle_{L^2\! \left(\mathbb{R}^3\times \underset{n\in \mathbb N}{\bigcup}  \mathbb{R}^{3 n}\right)}=\left\langle   \overline{\Psi}_\alpha,\mathbb{H}^{K}_{\frac{1}{i}\nabla_x} \overline{\Psi}_\alpha \right\rangle_{L^2\! \left(\mathbb{R}^3\times \underset{n\in \mathbb N}{\bigcup}  \mathbb{R}^{3 n}\right)}.
\end{align}
    Using the localization functions $F_1:=F$ and $F_2:=\sqrt{1-F^2}$, we obtain by Lemma \ref{Lem:Localization_IMS} for a suitable constant $C>0$
\begin{align*}
    \sum_{j=1}^2\left\langle F_j(\rho_{Y_n})\Psi_\alpha , \mathbb{H}_{0} F_j(\rho_{Y_n})\Psi_\alpha\right\rangle\leq  E_\alpha + C\alpha^{-4}.
\end{align*}
Since $\mathbb{H}_{0}\geq E_\alpha$ as an operator inequality, we obtain for the state
\begin{align*}
    \widetilde \Psi_\alpha:=\mu_\alpha^{-1} F_1(\rho_{Y_n})\Psi_\alpha,    
\end{align*}
with $\mu_\alpha=\|F_1(\rho_{Y_n})\Psi_\alpha\|=\sqrt{1-\| F_2(\rho_{Y_n})\Psi_\alpha\|^2}$ introduced below Eq.~(\ref{Eq:Localized_State_LLP}), the estimate
\begin{align*}
    \left\langle \widetilde \Psi_\alpha , \mathbb{H}_{0} \widetilde \Psi_\alpha\right\rangle\leq \mu_\alpha^{-2}\left(E_\alpha+C\alpha^{-4}-E_\alpha \| F_2(\rho_{Y_n})\Psi_\alpha\|^2\right)=E_\alpha+\frac{C\alpha^{-4}}{\mu_\alpha^2}\leq E_\alpha+2C\alpha^{-4},
\end{align*}
where we have used that $\mu_\alpha^2=1-\left\langle \Psi_\alpha,\left(1-F(\rho_{Y_n})^2\right)\Psi_\alpha\right\rangle\geq \frac{1}{2}$ by Lemma \ref{Lem:Most_Of_Mass}. Together with Eq.~(\ref{Eq:Comparision_full_energy_VS_cut-off_Energy}) we obtain for a suitable constant $C>0$ and $\alpha$ large enough
\begin{align}
\label{Eq:widetilde_Psi_Estimate}
    \left\langle \widetilde \Psi_\alpha , \mathbb{H}^K_{0} \widetilde \Psi_\alpha\right\rangle\leq  E_\alpha+\frac{1}{2}\alpha^{-\frac{4}{29}}.
\end{align}
In order to compute the term on the right hand side of Eq.~(\ref{Eq:LLP_in_BEC_Section}), let us first write
\begin{align}
\nonumber
    &  \left\langle   \overline{\Psi}_\alpha, \mathbb{H}^K_{\frac{1}{i}\nabla_x}    \overline{\Psi}_\alpha \right\rangle_{L^2\! \left(\mathbb{R}^3\times \underset{n\in \mathbb N}{\bigcup}  \mathbb{R}^{3 n}\right)} = \Big\| \Big(\mathcal P-\frac{1}{i}\nabla_x\Big)  \overline{\Psi}_\alpha\Big\|^2 +    \left\langle  \overline{\Psi}_\alpha, \mathcal{N}    \overline{\Psi}_\alpha \right\rangle_{L^2\! \left(\mathbb{R}^3\times \underset{n\in \mathbb N}{\bigcup}  \mathbb{R}^{3 n}\right)}\\
    \label{Eq:Writing_Out_LLP_Hamiltonian}
    & \ \ \ \ \ \ \ \  \ \ \ \ \ \ \ \  \ \ \ \ \ \ \ \ -2\mathfrak{Re}   \left\langle  \overline{\Psi}_\alpha, L(v^{K})    \overline{\Psi}_\alpha \right\rangle_{L^2\! \left(\mathbb{R}^3\times \underset{n\in \mathbb N}{\bigcup}  \mathbb{R}^{3 n}\right)}.
\end{align}
We note that $\overline{\Psi}_\alpha(x;Y_n)=\tau_\eta(m_\eta(\rho_{Y_n})+x)\widetilde \Psi_\alpha(Y_n)$, and that the operator $\tau_\eta(m_\eta(\rho_{Y_n})+x)$ commutes with $\mathcal P-\frac{1}{i}\nabla$, and therefore
\begin{align}
\nonumber
    & \Big\| \Big(\mathcal P-\frac{1}{i}\nabla_x\Big)\overline{\Psi}_\alpha\Big\|^2=\Big\| \tau_{\eta}\! \left(m_\eta(\rho_{Y_{n}})+x\right)\Big(\mathcal P-\frac{1}{i}\nabla_x\Big)\widetilde{\Psi}_\alpha \Big\|^2=\Big\| \tau_{\eta}\! \left(m_\eta(\rho_{Y_{n}})+x\right)\mathcal P\widetilde{\Psi}_\alpha \Big\|^2\\
    \nonumber
     &  \ \ = \sum_{n=0}^\infty \int_{\mathbb{R}^{3n}} \left(\int_{\mathbb{R}^3} \tau_{\eta}\! \left(m_\eta(\rho_{Y_{n}})+x\right)^2\mathrm{d}x\right) |\mathcal P\widetilde{\Psi}_\alpha(Y_n)|^2\mathrm{d}Y_{n}\\
     \label{Eq:P_f_Identity}
     &  \ \ =\sum_{n=0}^\infty \int_{\mathbb{R}^{3n}} |\mathcal P\widetilde{\Psi}_\alpha(Y_n)|^2\mathrm{d}Y_{n}= \Big\| \mathcal P\widetilde{\Psi}_\alpha \Big\|^2.
\end{align}
In a similar fashion we have
\begin{align}
\label{Eq:mathcal_N_Identity}
    \left\langle   \overline \Psi_\alpha, \mathcal{N}   \overline \Psi_\alpha \right\rangle_{L^2\! \left(\mathbb{R}^3\times \underset{n\in \mathbb N}{\bigcup}  \mathbb{R}^{3 n}\right)}=\left\langle    \widetilde \Psi_\alpha, \mathcal{N}     \widetilde \Psi_\alpha \right\rangle_{\mathcal{F}}.
\end{align}
Regarding $\left\langle   \overline \Psi_\alpha, L(v^{K})     \overline \Psi_\alpha \right\rangle_{L^2\! \left(\mathbb{R}^3\times \underset{n\in \mathbb N}{\bigcup}  \mathbb{R}^{3 n}\right)}$, recall the definition of $\varphi_\eta$ in Eq.~(\ref{Eq:Def_F_with_subscript})
\begin{align*}
    \varphi_\eta(Y_n) =\int_{\mathbb{R}^3}\left[\tau_{\eta}\! \left(m_\eta(\rho_{Y_{n-1}})+x\right)-\tau_{\eta}\! \left(m_\eta(\rho_{Y_{n}})+x\right)\right]\tau_{\eta}\! \left(m_\eta(\rho_{Y_{n-1}})+x\right)\mathrm{d}x,
\end{align*}
and use the fact that $0\leq \varphi_\eta(Y_n)\leq \frac{C}{\eta^4}\alpha^{-4}$ for a suitable constant $C>0$ and $Y_n$ in the support of $\widetilde \Psi_\alpha$, see Lemma \ref{Lema:F_eta_estimate}, in order to obtain
\begin{align}
\label{Eq:L_Identity}
 & \ \ \ \  \ \ \ \  \left|\left\langle   \overline \Psi_\alpha, L(v^{K})     \overline \Psi_\alpha \right\rangle_{L^2\! \left(\mathbb{R}^3\times \underset{n\in \mathbb N}{\bigcup}  \mathbb{R}^{3 n}\right)} -\left\langle   \widetilde \Psi_\alpha, L(v^{K})     \widetilde \Psi_\alpha \right\rangle_{L^2\! \left(\underset{n\in \mathbb N}{\bigcup}  \mathbb{R}^{3 n}\right)}\right|\\
 \nonumber
 &  \ \ \ \ \ \ \ \ \ \ \ \  =\left|\left\langle   \widetilde \Psi_\alpha, \varphi_\eta(Y_n)L(v^{K})     \widetilde \Psi_\alpha \right\rangle_{L^2\! \left(\underset{n\in \mathbb N}{\bigcup}  \mathbb{R}^{3 n}\right)}\right|\\
 \nonumber
 & \leq \frac{C}{\eta^4}\alpha^{-4} \! \! \! \sum_{n=1}^{(\|\varphi^\mathrm{Pek}\|^2+\sigma+\kappa)\alpha^2} \! \! \frac{\sqrt{n}}{\alpha}\int_{\mathbb R^{3n}}\left|\Psi_\alpha(Y_{n})\right| \left|\Psi_\alpha(Y_{n-1})\right| \left|v^K(y_n)\right| \mathrm{d}Y_n\lesssim \eta^{-4}\alpha^{-4}\|v^K\|\lesssim \eta^{-4}\alpha^{-4} \sqrt{K}.
\end{align}
By our choice of $\eta$ in Eq.~(\ref{Eq:Definition_eta_from_beta}) and the assumption $K\leq \alpha$, we obtain 
\begin{align*}
    \eta^{-4}\alpha^{-4} \sqrt{K}\leq \frac{1}{4}\alpha^{-\frac{4}{29}}
\end{align*}
for a suitable constant $C>0$ and $\alpha$ large enough. Combining Eq.~(\ref{Eq:Writing_Out_LLP_Hamiltonian}), Eq.~(\ref{Eq:P_f_Identity}), Eq.~(\ref{Eq:mathcal_N_Identity}) and Eq.~(\ref{Eq:L_Identity}) yields
\begin{align*}
   \left\langle   \overline \Psi_\alpha,\mathbb{H}^{K}_{\frac{1}{i}\nabla_x}  \overline\Psi_\alpha \right\rangle_{ \!  \!  \! L^2\! \left(\mathbb{R}^3\times \underset{n\in \mathbb N}{\bigcup}  \mathbb{R}^{3 n}\right)} \!  \!  & \leq \!  \Big\| \mathcal P\widetilde{\Psi}_\alpha \Big\|^2 \!   \! \! +   \! \left\langle    \widetilde \Psi_\alpha, \mathcal{N}     \widetilde \Psi_\alpha \right\rangle_{\mathcal{F}}  \!  \! \! - \! 2\mathfrak{Re}\left\langle   \widetilde \Psi_\alpha, L(v^{K})     \widetilde \Psi_\alpha \right\rangle_{ \!  \!  \! L^2\! \left(\underset{n\in \mathbb N}{\bigcup}  \mathbb{R}^{3 n}\right)} \!  \! + \!  \frac{1}{2}\alpha^{-\frac{4}{29}}\\
    & =  \left\langle \widetilde \Psi_\alpha , \mathbb{H}^K_{0} \widetilde \Psi_\alpha\right\rangle_{\mathcal{F}}+ \frac{1}{2}\alpha^{-\frac{4}{29}}.
\end{align*}
This concludes the proof by Eq.~(\ref{Eq:widetilde_Psi_Estimate}).
\end{proof}

The following result, concerning the Weyl operator $W_{\varphi^\mathrm{Pek}}$ with respect to the state $\varphi^\mathrm{Pek}$, see Eq.~(\ref{Eq:Def_Weyl_Operator}), is a direct consequence of Lemma \ref{Lem:Energy_Estimates_tilde} and \cite[Lemma 3.11]{BS1}, and establishes that $\Phi_\alpha$ satisfies Bose-Einstein condensation.

\begin{lem}
\label{Lem:BEC_Imported}
   Let $\Phi_\alpha$ be defined in Eq.~(\ref{Eq:Def_Phi_State}). Then there exists a constant $C>0$, such that 
    \begin{align*}
         \left\langle \Phi_\alpha, W^{-1}_{\varphi^{\mathrm{Pek}}}\mathcal{N}W_{\varphi^{\mathrm{Pek}}} \, \Phi_\alpha\right\rangle_{L^2(\mathbb R^3)\otimes \mathcal{F}}\leq C\alpha^{-\frac{2}{29}} .
    \end{align*}
\end{lem}
\begin{proof}
      By Lemma \ref{Lem:Energy_Estimates_tilde}, together with the trivial observation that $E_\alpha\leq e^\mathrm{Pek}$, we have for $K:=\alpha$
    \begin{align*}
          \left\langle  \Phi_\alpha, \mathbb{H}^K  \Phi_\alpha \right\rangle_{L^2\! \left(\mathbb{R}^3\times \underset{n\in \mathbb N}{\bigcup}  \mathbb{R}^{3 n}\right)}\leq E_\alpha+\alpha^{-\frac{4}{29}}\leq e^\mathrm{Pek}+\alpha^{-\frac{4}{29}},
    \end{align*}
    and furthermore note that $\mathrm{supp}(\Phi_\alpha)\subseteq \mathbb R^3\times \Big(\mathrm{supp}\big(\tau_\eta(m_\eta(\rho_{Y_n}))\big)\cap \mathrm{supp}\big(F\big)\Big)$, which especially implies that for all $Y_n\in \mathrm{supp}(\Phi_\alpha)$
    \begin{align*}
        \int \mathrm{d}\rho_{Y_n}\leq \|\varphi^\mathrm{Pek}\|^2+\sigma+\kappa, \ \ \ \left|m_\eta(\rho_{Y_n})\right|\leq \eta.
        \end{align*}
    Therefore, $\Phi_\alpha$ satisfies the assumptions of \cite[Lemma 3.11]{BS1} with corresponding constants $q:=\epsilon:=\alpha^{-\beta}\leq \alpha^{-\frac{2}{29}}$, see the definition of $\eta$ in Eq.~(\ref{Eq:Definition_eta_from_beta}), which concludes the proof.
\end{proof}

With Lemma \ref{Lem:BEC_Imported} at hand, we are in a position to verify that a typical empirical measure $\rho_{Y_n}$ is expected to be close to the Pekar measure $\rho^\mathrm{Pek}$ with respect to the mollified total variation.

\begin{lem}
\label{Lem:Estimate_Convoluted_Total_Variation}
    Let $\rho^{\mathrm{Pek}}$ be the Pekar measure $\mathrm{d}\rho^{\mathrm{Pek}}:=|\varphi^\mathrm{Pek}|^2\mathrm{d}x$ and let $\Phi_\alpha$ be defined in Eq.~(\ref{Eq:Def_Phi_State}). Then there exists an $\epsilon>0$ such that for $T\geq 1$
    \begin{align*}
        \sum_{n=0}^\infty \int_{\mathbb R^{3n}}\|g_T \! * \! (\rho_{Y_n}-\rho^{\mathrm{Pek}})\|_{\mathrm{TV}} \! \left(\int_{\mathbb R^3} |\Phi_\alpha(x;Y_n)|^2 \mathrm{d}x\right)\! \mathrm{d}Y_n\leq \alpha^{-\epsilon}T^{\frac{3}{2}}.
    \end{align*}
\end{lem}
\begin{proof}
    Let us introduce for $R>0$ the notation
    \begin{align*}
         \varphi_R: & =\chi(|\cdot| \leq R)\varphi^\mathrm{Pek},\\
         \mathrm{d}\widetilde \rho_R: & =\left|\varphi^\mathrm{Pek}_R\right|^2\mathrm{d}x,\\
         g_{T,x}(y): & =g_T(y-x),
    \end{align*}
and express the total variation between the convoluted measures as
    \begin{align}
    \label{Eq:Expression_Metric_Empirical_Measure}
        \|g_T \! * \! (\rho_{Y_n}-\widetilde \rho_R)\|_{\mathrm{TV}}=\int_{\mathbb R^3} \left|\int g_{T,x}\mathrm{d}\rho_{Y_n}-\int g_{T,x}\mathrm{d}\widetilde \rho_R\right|\mathrm{d}x.
    \end{align}
    Using an orthonormal basis $\{u_n:n\in \mathbb N\}$ and the notation 
    \begin{align*}
        G_{n,m}(x): &  =\braket{u_m,g_{T,x} u_n},\\
         \varphi_{R,n}: & =\braket{u_n, \varphi_R},\\
         b_n: & =a(u_n)-\varphi_{R,n},
    \end{align*}
   we can express by Eq.~(\ref{Eq:Multiplication_Operator_in_terms_of_a}) for fixed $x\in \mathbb R^3$ the multiplication operator by $Y_n\mapsto \int g_{T,x}\mathrm{d}\rho_{Y_n}$, acting on the Hilbert space $L^2\! \left(\mathbb{R}^3\times \underset{n\in \mathbb N}{\bigcup}  \mathbb{R}^{3 n}\right)$, as
    \begin{align*}
        & \int g_{T,x}\mathrm{d}\rho_{Y_n}-\int g_{T,x}\mathrm{d}\widetilde \rho_R=\sum_{n,m=0}^\infty G_{n,m}(x) \left(a_n^* a_m-\overline \varphi_{R,n}\varphi_{R,m}\right)\\
        & \ \ \ \ \ \  \ \ \  =\sum_{n,m=0}^\infty G_{n,m}(x) \left(b_n^* b_m+b_n^*\varphi_{R,m}+\overline \varphi_{R,n} b_m\right).
    \end{align*}
   By the Cauchy-Schwarz inequality we therefore obtain
    \begin{align*}
        & \left(\int g_x\mathrm{d}\rho_{Y_n}-\int g_x\mathrm{d}\widetilde \rho_R\right)^2\leq 3 \big(A_1+A_2+A_3\big),
    \end{align*}
    with $ A_1 :  =  \left(\sum_{n,m=0}^\infty G_{n,m}(x) b_n^* b_m\right)^2$ and
    \begin{align*}
        A_2 : & =  \left(\sum_{n,m=0}^\infty G_{n,m}(x) b_n^*\varphi_{R,m}\right)\left(\sum_{n,m=0}^\infty G_{n,m}(x) \overline \varphi_{R,n} b_m\right),\\ 
         A_3 : & =  \left(\sum_{n,m=0}^\infty G_{n,m}(x) \overline \varphi_{R,n} b_m\right)\left(\sum_{n,m=0}^\infty G_{n,m}(x) b_n^*\varphi_{R,m}\right).
    \end{align*}
    Using the canonical commutation relations $[b_m,b_n^*]=\alpha^{-2}\delta_{m,n}$ we obtain
    \begin{align*}
        A_3=A_2+\alpha^{-2}\|g_{T,x}\varphi_R\|^2.
    \end{align*}
    Furthermore, let $G'(x):L^2\! \left(\mathbb R^3\right)\longrightarrow L^2\! \left(\mathbb R^3\right)$ be the rank one operator, defined in coordinates as $G'_{n,m}(x)=\braket{u_n,g_x \varphi_R}\overline{\braket{u_m,g_x \varphi_R}}$. Clearly
    \begin{align*}
        \|G'(x)\|_\mathrm{op}= \|g_{T,x} \varphi_R\|^2,
    \end{align*}
    and therefore
    \begin{align*}
        A_2 = \sum_{n,m=0}^\infty G'_{n,m}(x) b_n^* b_m\leq \|g_{T,x} \varphi_R\|^2 \sum_{n=0}^\infty b_n^* b_n.
    \end{align*}
    In the following let $r>0$ be large enough such that $\mathrm{supp}(g)\subseteq B_r(0)$, and note that $g_x \varphi_R=0$ in case $|x|>R+r$. Summarizing what we have so far yields for $\lambda>0$
    \begin{align*}
       & \left(\int g_{T,x}\mathrm{d}\rho_{Y_n}-\int g_{T,x}\mathrm{d}\widetilde \rho_R\right)^2\leq 3 \left(\sum_{n,m=0}^\infty G_{n,m}(x) b_n^* b_m\right)^2+3\|g_{T,x} \varphi_R\|^2 \left(2\sum_{n=0}^\infty b_n^* b_n + \alpha^{-2}\right) \\
       & \ \ \ \leq  3 \left(\sum_{n,m=0}^\infty G_{n,m}(x) b_n^* b_m\right)^2+3\lambda\|g_{T,x} \varphi_R\|^4 +\chi(|x| \! \leq  \! R \! + \! r)3\lambda^{-1}\left(\sum_{n=0}^\infty b_n^* b_n + \alpha^{-2}\right)^2\\
         & \ \ \ \leq  3\left(\sum_{n,m=0}^\infty G_{n,m}(x) b_n^* b_m+\lambda \|g_{T,x} \varphi_R\|^2+\lambda^{-1} \chi(|x| \! \leq  \! R \! + \! r)\left(\sum_{n=0}^\infty b_n^* b_n+\alpha^{-2}\right)\right)^2,
    \end{align*}
    where we have used in the last estimate that all the components $\sum_{n,m=0}^\infty G_{n,m}(x) b_n^* b_m$, $\|g_{T,x} \varphi_R\|^2$ and $\chi(|x| \! \leq  \! R \! + \! r)\left(\sum_{n=0}^\infty b_n^* b_n+1\right)$ are non-negative and commute. Since the square root $t\mapsto \sqrt{t}$ is operator monotone, we therefore have by Eq.~(\ref{Eq:Expression_Metric_Empirical_Measure})
    \begin{align*}
        & \ \ \ \  \ \  \ \   \ \ \ \  \ \  \ \  \|g_T \! * \! (\rho_{Y_n} \! - \! \widetilde \rho_R)\|_{\mathrm{TV}}\\
        & \leq  \sqrt 3  \! \int_{\mathbb R^3} \! \left(  \sum_{n,m=0}^\infty G_{n,m}(x) b_n^* b_m \! + \! \lambda\|g_{T,x} \varphi_R\|^2 \! + \! \lambda^{-1}\chi(|x| \! \leq  \! R \! + \! r) \! \left(\sum_{n=0}^\infty b_n^* b_n \! + \! \alpha^{-2}\right) \! \right)\mathrm{d}x \\
         & =   \! \sqrt 3  \!  \sum_{n,m=0}^\infty   \! \left(\int_{\mathbb R^3}G_{n,m}(x)\mathrm{d}x\right)  b_n^* b_m   \! +   \! \lambda\left\langle \varphi_R , \int_{\mathbb R^3}g_{T,x}^2 \mathrm{d}x \, \varphi_R \right\rangle   \! +  \! \lambda^{-1}\frac{4\pi (R  \! +  \! r)^3}{3}\left(\sum_{n=0}^\infty b_n^* b_n \! + \! \alpha^{-2}\right)  \!   \! .
    \end{align*}
    Using the fact that $\int_{\mathbb R^3}G_{n,m}(x)\mathrm{d}x=\int_{\mathbb R^3}g_T(y)\mathrm{d}y\delta_{n,m}=\delta_{n,m}$, we obtain for a suitable constant $C>0$ and all $R\geq 1$ the operator inequality
    \begin{align*}
        \|g \! * \! (\rho_{Y_n} \! - \! \widetilde \rho_R)\|_{\mathrm{TV}} \leq C \left(\left(1+\lambda^{-1}R^3\right)\! \! \sum_{n=0}^\infty b_n^* b_n +\lambda T^3 + \alpha^{-2}R^3\right),
    \end{align*}
where we have used $\int_{\mathbb R^3}g_{T,x}^2 \mathrm{d}x=T^3\|g\|^2$. Furthermore, note that for a suitable $C>0$
\begin{align*}
  \|g \! * \! (\rho_{Y_n} \!  -  \!  \rho^\mathrm{Pek})\|_{\mathrm{TV}}  & \leq   \|g \! * \! (\rho_{Y_n}   \!  - \!   \widetilde \rho_R)\|_{\mathrm{TV}}+\|g \! * \! (\rho^\mathrm{Pek}   \!  -  \!  \widetilde \rho_R)\|_{\mathrm{TV}}\leq \|g \! * \! (\rho_{Y_n}  \!   - \!   \widetilde \rho_R)\|_{\mathrm{TV}}+\frac{C}{R} 
  \end{align*}
  and
  \begin{align*}
\sum_{n=0}^\infty b_n^* b_n    & \leq 2 \sum_{n=0}^\infty \left(a(u_n)-\varphi^\mathrm{Pek}_n\right)^* \left(a(u_n)-\varphi^\mathrm{Pek}_n\right)+2\|\varphi^\mathrm{Pek}-\varphi_R\|^2\\
   & \leq 2 \sum_{n=0}^\infty \left(a(u_n)-\varphi^\mathrm{Pek}_n\right)^* \left(a(u_n)-\varphi^\mathrm{Pek}_n\right)+\frac{C}{R},
\end{align*}
see Lemma \ref{Lem:Semiclassical_objects_properties}, and 
\begin{align*}
   \sum_{n=0}^\infty \left(a(u_n)-\varphi^\mathrm{Pek}_n\right)^* \left(a(u_n)-\varphi^\mathrm{Pek}_n\right)=  W^{-1}_{\varphi^{\mathrm{Pek}}}\mathcal{N}W_{\varphi^{\mathrm{Pek}}},
\end{align*}
where $W_{\varphi^{\mathrm{Pek}}}$ is the Weyl operator defined in Eq.~(\ref{Eq:Def_Weyl_Operator}). Consequently
\begin{align*}
   & \ \ \sum_{n=0}^\infty \int_{\mathbb R^{3n}}\|g \! * \! (\rho_{Y_n}-\rho^{\mathrm{Pek}})\|_{\mathrm{TV}} \! \left(\int_{\mathbb R^3} |\Phi_\alpha(x;Y_n)|^2 \mathrm{d}x\right)\! \mathrm{d}Y_n\\
    & \leq 2 C \left(1+\lambda^{-1}R^3\right) \! \! \left\langle \Phi_\alpha, W^{-1}_{\varphi^{\mathrm{Pek}}}\mathcal{N}W_{\varphi^{\mathrm{Pek}}} \, \Phi_\alpha\right\rangle +2C\left(\lambda T^3+\alpha^{-2}R^3+R^{-1}\right).
\end{align*}
By our choice of $\eta$ and Lemma \ref{Lem:BEC_Imported} there exist $\kappa,C>0$ such that
\begin{align*}
     \left\langle \Phi_\alpha, W^{-1}_{\varphi^{\mathrm{Pek}}}\mathcal{N}W_{\varphi^{\mathrm{Pek}}} \, \Phi_\alpha\right\rangle\leq C\alpha^{-\kappa}.
\end{align*}
Optimizing in $\lambda$ and $R$ concludes the proof.
\end{proof}

Finally, let us show in the following Lemma \ref{Lem:Electron_Wave_Function} that the state $\Phi_\alpha$ is close to a tensor state of the form $\psi^\mathrm{Pek}\otimes \Phi_\alpha'$ for some $\Phi_\alpha'\in \mathcal{F}$.

\begin{lem}
    \label{Lem:Electron_Wave_Function}
    Recall the definition of $\psi^\mathrm{Pek}$ and $Q^\mathrm{Pek}$ from Subsection \ref{Subsec:Semi-Classical_Objects}, and let
    \begin{align*}
        Q^\mathrm{Pek}_*:=Q^\mathrm{Pek}\otimes 1_\mathcal{F}.
    \end{align*}
   Furthermore let $\Phi_\alpha$ be the state defined in Eq.~(\ref{Eq:Def_Phi_State}). Then there exist $C,\epsilon>0$ such that
    \begin{align*}
        \|Q^\mathrm{Pek}_*\Phi_\alpha \|\leq C\alpha^{-\epsilon}.
    \end{align*}
\end{lem}
\begin{proof}
    Combining Eq.~(\ref{Eq:Energy_Localized_State_Pre}), for $K:=\alpha$, and the trivial observation that $E_\alpha\leq e^\mathrm{Pek}$ yields 
    \begin{align}
    \label{Eq:Comparison_With_Pekar_Energy}
        & \left\langle  \Phi_\alpha,-\Delta_x  \Phi_\alpha\right\rangle_{L^2(\mathbb R^3)\otimes \mathcal{F}}+ \left\langle  \Phi_\alpha,\mathcal{N}  \Phi_\alpha\right\rangle_{L^2(\mathbb R^3)\otimes \mathcal{F}}- 2\mathfrak{Re}\left\langle  \Phi_\alpha,a(v^K_{x})  \Phi_\alpha\right\rangle_{L^2(\mathbb R^3)\otimes \mathcal{F}}\\
        \nonumber
        & \ \ \ \ \ \ \ \ \ \ \  =\left\langle \Phi_\alpha, \mathbb H^K \Phi_\alpha \right\rangle_{L^2(\mathbb R^3)\otimes \mathcal{F}}\leq e^\mathrm{Pek}+\alpha^{-\frac{4}{29}}.
    \end{align}
    In order to analyse the term $\left\langle  \Phi_\alpha,a(v^K_{x})  \Phi_\alpha\right\rangle$, let us define the operators
    \begin{align*}
        b_j: & =a(u_n)-\left\langle u_n ,\varphi^{\mathrm{Pek}}\right \rangle,\\
        \widetilde{\mathcal{N}}: & =\sum_{j=0}^\infty b_j^* b_j=W^{-1}_{\varphi^{\mathrm{Pek}}}\mathcal{N}W_{\varphi^{\mathrm{Pek}}},
    \end{align*}
    where $\{u_n:n\in \mathbb N\}$ is an orthonormal basis of $L^2\! \left(\mathbb R^3\right)$ and $W_{\varphi^{\mathrm{Pek}}}$ is the Weyl transformation introduced in Eq.~(\ref{Eq:Def_Weyl_Operator}). With this at hand we can write
    \begin{align}
    \label{Eq:Interaction_Split_For_Electron}
        \left\langle  \Phi_\alpha,a(v^K_{x})  \Phi_\alpha\right\rangle_{L^2(\mathbb R^3)\otimes \mathcal{F}} \! = \!  \left\langle  \Phi_\alpha,\left\langle v^K_{x},\varphi^\mathrm{Pek}
        \right\rangle \Phi_\alpha \right\rangle_{L^2(\mathbb R^3)\otimes \mathcal{F}} \! + \! \sum_{j=0}^\infty \left\langle v^K_{x},u_j
        \right\rangle  \left\langle  \Phi_\alpha,b_j  \Phi_\alpha\right\rangle_{L^2(\mathbb R^3)\otimes \mathcal{F}}.
    \end{align}
    Since $\varphi^\mathrm{Pek}=\left|\psi^\mathrm{Pek}\right|^2*v$ and $\psi^\mathrm{Pek}\in H^2\! \left(\mathbb R^3\right)$, see \cite{Li,MS}, we have $\left\||k|^2 \widehat{\left|\psi^\mathrm{Pek}\right|^2}(k)\right\|_\infty<\infty$ and
    \begin{align*}
        \sup_{x\in \mathbb R^3}\left|\left\langle v_{x},\varphi^\mathrm{Pek}
        \right\rangle \! - \! \left\langle v^K_{x},\varphi^\mathrm{Pek}
        \right\rangle\right| \! \lesssim \!  \int_{|k|>K}\left|\widehat{\left|\psi^\mathrm{Pek}\right|^2}(k)\right| |\widehat{v}(k)|^2\mathrm{d}k \! \lesssim \! \int_{|k|>K}|k|^{-4}\mathrm{d}k \! \lesssim \!  \frac{1}{K} \! = \! \alpha^{-1}.
    \end{align*}
    Defining the multiplication operator $V(x):=-2\mathfrak{Re}\left\langle v_{x},\varphi^\mathrm{Pek}
        \right\rangle$ acting on $L^2 \! \left(\mathbb R^3\right)$, there consequently exists a constant $C>0$ such that
    \begin{align}
        \label{Eq:Potential_Electron_Cutoff}
        -2\mathfrak{Re}\left\langle  \Phi_\alpha,\left\langle v^K_{x},\varphi^\mathrm{Pek}
        \right\rangle \Phi_\alpha \right\rangle_{L^2(\mathbb R^3)\otimes \mathcal{F}}\geq \left\langle  \Phi_\alpha,V\otimes 1_\mathcal{F}\,  \Phi_\alpha \right\rangle-\frac{C}{\alpha}.
    \end{align}
    Regarding the second term in Eq.~(\ref{Eq:Interaction_Split_For_Electron}), let us define for $i\in\{1,2,3\}$ the elements 
    \begin{align*}
        \widetilde w^{K}_{x}:=\frac{\chi(1\leq |\nabla|\leq K)}{i\Delta}\nabla v_x
    \end{align*}
and compute in analogy to Lemma \ref{Lem:Semiclassical_objects_properties}
    \begin{align*}
        v^K_{x}=v^1_{x}+\left[\frac{1}{i}\nabla_x,\widetilde w^{K}_{x}\right],
    \end{align*}
  see also \cite{LY}. Following the approach in \cite[Section 2]{BS1} we have for $\lambda>0$ the estimates
    \begin{align*}
      \mathfrak{Re}\sum_{j=0}^\infty \left\langle v_{x}^1,u_j
        \right\rangle  b_j  & \leq \lambda \|v^1\|^2+\lambda^{-1}\widetilde{\mathcal{N}}, \\
      \mathfrak{Re}\sum_{j=0}^\infty \left[\frac{1}{i}\nabla_x,\left\langle \widetilde w^{K}_{x},u_j
        \right\rangle  b_j\right]   & \leq -\lambda \Delta_x+\lambda^{-1}\|\widetilde w^{K}\|^2\left(\widetilde{\mathcal{N}}+\alpha^{-2}\right).
    \end{align*}
Using that $\|v^1\|\lesssim 1$ and $\|\widetilde w^{K}\|^2\leq \| w^{K}\|^2\lesssim 1$, see Lemma \ref{Lem:Semiclassical_objects_properties}, as well as $\left\langle \Phi_\alpha, \widetilde{\mathcal{N}}\Phi_\alpha\right\rangle\leq C\alpha^{-\kappa}$, see Lemma \ref{Lem:BEC_Imported}, and 
\begin{align*}
     \left\langle \Phi_\alpha, \! -\Delta_x\Phi_\alpha\right\rangle_{ \! L^2(\mathbb R^3)\otimes \mathcal{F}} \! \leq \! \left\langle \Phi_\alpha, \! (\mathbb H^K \! + \! C)\Phi_\alpha\right\rangle_{ \! L^2(\mathbb R^3)\otimes \mathcal{F}} \! \leq \! e^\mathrm{Pek} \! + \! C \! + \! \alpha^{-\frac{4}{29}} \! ,
\end{align*}
see Eq.~(\ref{Eq:Particle_Number_Energy_bound_Full_Space}) and Eq.~(\ref{Eq:Comparison_With_Pekar_Energy}), we obtain for suitable $C,\kappa>0$
\begin{align}
\nonumber
  &  \ \ \ \ \ \ \ -\mathfrak{Re}\sum_{j=0}^\infty \left\langle v^K_{x},u_j \right\rangle  \left\langle  \Phi_\alpha,b_j  \Phi_\alpha\right\rangle_{L^2(\mathbb R^3)\otimes \mathcal{F}}\\
  \nonumber
  &\geq \lambda \|v^1\|^2+\lambda^{-1}\alpha^{-2} +\lambda^{-1}\left(1+\|w^{K}\|^2\right)\left\langle \Phi_\alpha,\widetilde{\mathcal{N}} \Phi_\alpha \right\rangle_{L^2(\mathbb R^3)\otimes\mathcal{F}}+\lambda  \left\langle \Phi_\alpha, \! -\Delta_x\Phi_\alpha\right\rangle_{ \! L^2(\mathbb R^3)\otimes \mathcal{F}}\\
  \label{Eq:Electron_Projector_modified_number}
        &  \ \ \ \ \ \ \ \ \ \geq -C\lambda-C\lambda^{-1}\alpha^{-\kappa}=-2C\alpha^{-\frac{\kappa}{2}},
\end{align}
where we have optimized in $\lambda$ in the last identity. In a similar fashion we obtain
\begin{align}
\label{Eq:Particle_Number_Plugg_in_varphi}
    \left\langle  \Phi_\alpha,\mathcal{N}  \Phi_\alpha\right\rangle_{L^2(\mathbb R^3)\otimes \mathcal{F}} \geq  \|\varphi^\mathrm{Pek}\|-C\alpha^{-\frac{\kappa}{2}}.
\end{align}
Combining Eq.~(\ref{Eq:Comparison_With_Pekar_Energy}), Eq.~(\ref{Eq:Interaction_Split_For_Electron}), Eq.~(\ref{Eq:Potential_Electron_Cutoff}), Eq.~(\ref{Eq:Electron_Projector_modified_number}) and Eq.~(\ref{Eq:Particle_Number_Plugg_in_varphi}), and defining the operator 
\begin{align*}
  h:=-\Delta_x+V  
\end{align*}
acting on $L^2\! \left(\mathbb R^3\right)$, yields for suitable constants $C,\kappa>0$
\begin{align}
\label{Eq:Electron_h_comparison}
    \left\langle \Phi_\alpha, h\otimes 1_\mathcal{F}\, \Phi_\alpha\right\rangle_{L^2(\mathbb R^3)\otimes \mathcal{F}}\leq e^\mathrm{Pek}-\|\varphi^\mathrm{Pek}\|^2+C\alpha^{-\kappa}.
\end{align}
Note that by the variational definition of $e^\mathrm{Pek}$ in Eq.~(\ref{Eq:Variational_Definition}), it is clear that $e^\mathrm{Pek}-\|\varphi^\mathrm{Pek}\|^2<0$ is the ground state energy of $h$ with the unique ground state $\psi^\mathrm{Pek}$, and since the essential spectrum of $h$ is given by $[0,\infty)$, there exists a $\delta>0$ such that 
\begin{align*}
  \sigma(h)\cap (e^\mathrm{Pek}-\|\varphi^\mathrm{Pek}\|^2-\delta,e^\mathrm{Pek}-\|\varphi^\mathrm{Pek}\|^2+\delta)=\{e^\mathrm{Pek}-\|\varphi^\mathrm{Pek}\|^2\},  
\end{align*}
or equivalently
\begin{align*}
h\geq e^\mathrm{Pek}-\|\varphi^\mathrm{Pek}\|^2+\delta Q^\mathrm{Pek}.
\end{align*}
Together with Eq.~(\ref{Eq:Electron_h_comparison}) we obtain 
\begin{align*}
    \|Q^\mathrm{Pek}_*\Phi_\alpha\|^2=\left\langle \Phi_\alpha, Q^\mathrm{Pek}_*\, \Phi_\alpha\right\rangle_{L^2(\mathbb R^3)\otimes \mathcal{F}}\leq \frac{C}{\delta}\alpha^{-\kappa}.
\end{align*}
\end{proof}

\appendix

\section{Properties of the Semi-Classical minimizers}
\label{Appendix:Properties_of_the_Semi-Classical_minimizers}
In the following we are going to prove the properties of the semi-classical minimizers $(\psi^\mathrm{Pek},\varphi^\mathrm{Pek})$ and the truncated interaction terms $v^\Lambda$ and $w^\Lambda$ stated in Lemma \ref{Lem:Semiclassical_objects_properties}. For the readers convenience, we are going to display the claimed identities and estimates again
   \begin{align}
   \label{Eq:non-named_COPY}
    &     \ \ \ \   \ \ |\psi^\mathrm{Pek}|^2*v=\varphi^\mathrm{Pek},\\
       \label{Eq:Decomposition_in_Lemma_of_interaction_COPY}
 &   \ \    \ \  \ \  v=v^\Lambda+\frac{1}{i}\nabla\cdot w^\Lambda,\\
   \label{Eq:In_Lemma_v_w_Lambda_COPY}
      &  \ \  \left\|v^\Lambda\right\|\leq C\Lambda^{\frac{1}{2}}, \ \ \ \ \left\|w^\Lambda\right\|\leq C\Lambda^{-\frac{1}{2}}, \\
      \label{Eq:High_Momentum_varphi_COPY}
     & \ \   \ \   \left\|\chi\!\left(|\nabla|>\Lambda\right)\varphi^\mathrm{Pek}\right\|\leq C\Lambda^{-\frac{1}{2}},\\
     \label{Eq:High_position_varphi_COPY}
      &    \ \    \ \ \left\|\chi\! \left(|y|>R\right)\varphi^\mathrm{Pek}\right\|\leq C R^{-\frac{1}{2}},\\
      \label{Eq:C_b_COPY}
      &   \left\|\frac{\partial_{y_j} \varphi^\mathrm{Pek}}{\varphi^\mathrm{Pek}}\right\|_\infty+\left\|\partial_{y_k}\frac{\partial_{y_j} \varphi^\mathrm{Pek}}{\varphi^\mathrm{Pek}}\right\|_\infty+\left\| \partial_{y_\ell} \partial_{y_k}\frac{\partial_{y_j} \varphi^\mathrm{Pek}}{\varphi^\mathrm{Pek}}\right\|_\infty<\infty.
   \end{align}
   
\begin{proof}[Proof of Lemma \ref{Lem:Semiclassical_objects_properties}.]
       Eq.~(\ref{Eq:Decomposition_in_Lemma_of_interaction_COPY}) follows immediately from the operator identity 
    \begin{align*}
        1=\chi(|\nabla|\leq \Lambda)+\chi(|\nabla|> \Lambda)=\chi(|\nabla|\leq \Lambda)+\left(\frac{1}{i}\nabla\cdot \frac{i\nabla}{\Delta}\right)\chi(|\nabla|> \Lambda),
    \end{align*}
    where we have used $\Delta=\nabla\cdot \nabla$. Furthermore, note that $\varphi^\mathrm{Pek}$ minimizes the energy functional
    \begin{align*}
        \varphi\mapsto \mathcal{E}^\mathrm{Pek}(\psi^\mathrm{Pek},\varphi)= \left\|\nabla \psi^\mathrm{Pek}\right\|^2-\left\||\psi^\mathrm{Pek}|^2*v\right\|^2+\left\|\varphi-|\psi^\mathrm{Pek}|^2*v\right\|^2,
    \end{align*}
    which obtains its unique minimum at $\varphi=|\psi^\mathrm{Pek}|^2*v$, hence $\varphi^\mathrm{Pek}=|\psi^\mathrm{Pek}|^2*v$. Regarding the proof of Eq.~(\ref{Eq:In_Lemma_v_w_Lambda_COPY}) note that the Fourier transformation of $v^\Lambda$ and $w^\Lambda$ are given by
    \begin{align*}
        \widehat{v}^\Lambda(k)=(2\pi)^{-\frac{1}{2}}\chi(|k|\leq \Lambda)|k|^{-1}, \ \ \ \widehat{w}^\Lambda(k)=(2\pi)^{-\frac{1}{2}}\chi(|k|>\Lambda)|k|^{-3}k.
    \end{align*}
    Consequently,
    \begin{align*}
        \left\|v^\Lambda\right\|^2 & =(2\pi)^{-1}\int_{\mathbb R^3}\chi(|k|\leq \Lambda)|k|^{-2}\mathrm{d}k=2\Lambda,\\
         \left\|w^\Lambda\right\|^2 & =(2\pi)^{-1}\int_{\mathbb R^3}\chi(|k|> \Lambda)|k|^{-4}\mathrm{d}k=2\Lambda^{-1}.
    \end{align*}
  When it comes to Eq.~(\ref{Eq:High_Momentum_varphi_COPY}), we compute
    \begin{align*}
       & \ \left\|\chi\!\left(|\nabla|>\Lambda\right)\varphi^\mathrm{Pek}\right\|=\left\|\chi\!\left(|\nabla|>\Lambda\right)|\psi^\mathrm{Pek}|^2*v\right\|=\left\||\psi^\mathrm{Pek}|^2*\left(\nabla\cdot w^\Lambda\right)\right\|\\
       & = \left\|\left(\nabla |\psi^\mathrm{Pek}|^2\right)*w^\Lambda\right\|\leq \|\nabla |\psi^\mathrm{Pek}|^2\|_{L^1\left(\mathbb R^3\right)}\|w^\Lambda\|\leq \|\psi^\mathrm{Pek}\|\|\nabla \psi^\mathrm{Pek}\|\|w^\Lambda\|.
    \end{align*}
    Since $\psi^\mathrm{Pek}\in H^1\left(\mathbb R^3\right)$ and $\|w^\Lambda\|\leq C\Lambda^{-\frac{1}{2}}$ by Eq.~(\ref{Eq:In_Lemma_v_w_Lambda_COPY}), this concludes the proof of Eq.~(\ref{Eq:High_Momentum_varphi_COPY}). Using Eq.~(\ref{Eq:non-named_COPY}) we clearly have
    \begin{align}
    \label{Eq:Pre_Lower_Bound}
        \varphi^\mathrm{Pek}(x)=\pi^{-\frac{3}{2}}\int_{\mathbb R^3}\frac{\psi^\mathrm{Pek}(y)^2}{|x-y|^2}\mathrm{d}y\geq \pi^{-\frac{3}{2}}\int_{\mathbb R^3}\chi(|y|\leq 1)\psi^\mathrm{Pek}(y)^2\mathrm{d}y \frac{1}{(|x|+1)^2}.
    \end{align}
    By \cite{Li} we know that $\varphi^\mathrm{Pek},\psi^\mathrm{Pek}>0$, hence Eq.~(\ref{Eq:Pre_Lower_Bound}) implies that for a suitable $c>0$ 
    \begin{align*}
         \varphi^\mathrm{Pek}(y)\geq \frac{c}{(1+|y|)^2}.
    \end{align*}
    Consequently, both Eq.~(\ref{Eq:High_position_varphi_COPY}) and Eq.~(\ref{Eq:C_b_COPY}) follows once we have established 
    \begin{align*}
        \varphi^\mathrm{Pek}(y)+|\partial_{y_j} \varphi^\mathrm{Pek}(y)|+|\partial_{y_k}\partial_{y_j} \varphi^\mathrm{Pek}(y)|+|\partial_{y_\ell}\partial_{y_k}\partial_{y_j} \varphi^\mathrm{Pek}(y)|\leq \frac{C}{(1+|y|)^2}
    \end{align*}
    for a suitable $C>0$. Again by Eq.~(\ref{Eq:non-named_COPY}), it is enough to verify that the functions
    \begin{align}
    \label{Eq:Exponetial_Decay_All_Functions}
        \psi^\mathrm{Pek}(y), \partial_{y_j} \psi^\mathrm{Pek}(y),\partial_{y_k}\partial_{y_j} \psi^\mathrm{Pek}(y),\partial_{y_\ell}\partial_{y_k}\partial_{y_j}\psi^\mathrm{Pek}(y) 
    \end{align}
are bounded from above by $Ce^{-\gamma |y|}$ for suitable constants $\gamma,C>0$. In the case of $\psi^\mathrm{Pek}$ this has been verified in \cite{MS}, and furthermore we know by \cite{Li} that $\psi^\mathrm{Pek}\in C^\infty(\mathbb R^3)$. Since the function $\psi^\mathrm{Pek}$ is radial, i.e. 
\begin{align}
    \label{Eq:Radial_Function}
    \psi^\mathrm{Pek}(y)=\psi^\mathrm{Pek}_*(|y|),
\end{align}
the exponential decay of the functions in Eq.~(\ref{Eq:Exponetial_Decay_All_Functions}) follows from the exponential decay of
\begin{align*}
    \psi^\mathrm{Pek},\nabla  \psi^\mathrm{Pek},\Delta  \psi^\mathrm{Pek},\Delta \left(\nabla  \psi^\mathrm{Pek}\right).
\end{align*}
In order to verify the exponential decay of $\Delta  \psi^\mathrm{Pek}$, note that $\psi^\mathrm{Pek}$ as the ground state of the operator $-\Delta_x-2v*\mathfrak{Re}\varphi^\mathrm{Pek}$ satisfies the equation
\begin{align}
\label{Eq:In_Appendix_Delta}
\Delta  \psi^\mathrm{Pek}=2\left(v*\left(v*|\psi^\mathrm{Pek}|^2\right)+\|\varphi^\mathrm{Pek}\|^2-e^\mathrm{Pek}\right)\psi^\mathrm{Pek}.
\end{align}
Since we already know that $\psi^\mathrm{Pek}$ decays exponentially, it follows that $v*(v*|\psi^\mathrm{Pek}|^2)$ is bounded and therefore the right hand side of Eq.~(\ref{Eq:In_Appendix_Delta}) decays exponentially. Equivalently, we obtain that the function $\psi^\mathrm{Pek}_*$ in Eq.~(\ref{Eq:Radial_Function}) satisfies for suitable $C,\gamma>0$
\begin{align*}
    \left|\frac{1}{r^2}\partial_r\left(r^2 \partial_r \psi^\mathrm{Pek}_*\right)\right|\leq C e^{-\gamma r}.
\end{align*}
Since $ \psi^\mathrm{Pek}_*$ decays exponentially as well, we obtain by interpolation that 
\begin{align*}
  |\nabla \psi^\mathrm{Pek}(y)|=|\partial_r \psi^\mathrm{Pek}_*(|y|)|  
\end{align*}
satisfies an exponential decay. Taking the gradient of Eq.~(\ref{Eq:In_Appendix_Delta}) furthermore yields
\begin{align*}
   \Delta  \nabla \psi^\mathrm{Pek}=2\left(v*\left(v*|\psi^\mathrm{Pek}|^2\right)+\|\varphi^\mathrm{Pek}\|^2-e^\mathrm{Pek}\right)\nabla \psi^\mathrm{Pek}+4\left( v*\left(v*\left(\psi^\mathrm{Pek}\nabla \psi^\mathrm{Pek}\right)\right)\right)\psi^\mathrm{Pek}.
\end{align*}
Since $\psi^\mathrm{Pek}\nabla \psi^\mathrm{Pek}$ decays exponentially it is clear that $v*\left(v*\left(\psi^\mathrm{Pek}\nabla \psi^\mathrm{Pek}\right)\right)$ is bounded. Furthermore, we have already seen that $v*\left(v*|\psi^\mathrm{Pek}|^2\right)$ is bounded and that $\nabla \psi^\mathrm{Pek}$ as well as $\psi^\mathrm{Pek}$ decay exponentially, which concludes the proof.
\end{proof}

\begin{center}
\textsc{Acknowledgments}
\end{center}
Funding from the ERC Advanced Grant ERC-AdG CLaQS, grant agreement n. 834782, is gratefully acknowledged. Furthermore, the author would like to thank David Mitrouskas and Robert Seiringer for fruitful discussions.

\end{document}